\documentclass[fleqn,10pt]{article}
\usepackage{latexsym, graphicx, epsfig, amsmath, amssymb,amsfonts}
\usepackage{natbib,amsthm,version}
\usepackage{amsbsy,bm,multirow,enumerate}
\usepackage[titletoc,page]{appendix}
\usepackage[mathscr]{eucal}
\usepackage{mathtools}
\usepackage{color}
\usepackage{subfigure}
\usepackage[utf8]{inputenc}
\usepackage[english]{babel}
\usepackage{float}
\usepackage{caption}
\usepackage{systeme}
\usepackage[affil-it]{authblk}

\newtheorem{theorem}{Theorem}[section]

\newtheorem{lemma}[theorem]{Lemma}
\newtheorem{remark}{Remark}[theorem]

\begin{document}

\title{ 
	A consistent and conservative Phase-Field method for compressible multiphase flows with the six-equation model
} 

\author{
	Ziyang Huang%
	\thanks{Email: \texttt{ziyanghuang@scut.edu.cn}; Corresponding author at the School of Marine Science and Engineering, Guangzhou International Campus, South China University of Technology, Guangzhou, Guangdong, China, 511400.}}

\affil{
	School of Marine Science and Engineering, Guangzhou International Campus, South China University of Technology, Guangzhou, Guangdong, China, 511400}

\date{Aug 26, 2026}

\maketitle


\begin{abstract}
In the present study, the consistent and conservative Phase-Field method is extended to the six-equation model for compressible multiphase flows.
Based solely on the conservation laws and the second law of thermodynamics, the six-equation model with the Phase-Field mechanism is first derived. In addition to satisfying Galilean invariance and consistency of reduction, the model is general to admit an arbitrary number of phases with different formulations of the Phase-Field mechanism. The isobaric closure by the pressure relaxation and the incompressible limit of the proposed model are analyzed. The derivation and analysis identify additional terms arising from the Phase-Field mechanism, which are absent from previous studies.
The consistent and conservative numerical approach is adapted to the proposed six-equation model with moderate modifications, retaining the interfacial equilibrium condition, conservation, and flexibility and robustness of incorporating different formulations of the Phase-Field mechanism with bound preservation.
As the new component of solving the six-equation model, both the pressure and pressure-temperature relaxations are theoretically analyzed in a general multiphase setup, with a proof of the existence and uniqueness of a thermodynamically admissible solution for these two relaxations.
Various two-phase compressible flow benchmarks are performed to demonstrate the method, and good agreement with exact solutions is achieved.
\end{abstract}

\vspace{0.05cm}

\section{Introduction}\label{Sec:Introduction}
Compressible multiphase flows, characterized by dynamical interactions between material interfaces and wave structures such as shocks and rarefactions, are ubiquitous in scientific research and engineering applications, and the diffuse-interface capturing methods \citep{SaurelPantano2018} have demonstrated their success in solving such problems. In comparison to methods that have an explicit definition of interface location, such as the front-tracking \cite{UnverdiTryggvason1992,Tryggvasonetal2001}, level-set \cite{OsherSethian1988,Sussmanetal1994,SethianSmereka2003,Gibouetal2018}, and volume-of-fluid (VOF) \cite{HirtNichols1981,ScardovelliZaleski1999,OwkesDesjardins2017} methods, material interfaces in the diffuse-interface capturing methods are implicitly defined by transition over a few grid cells of a scalar function that is usually related to the volume (or mass) fractions of individual phases. Interface-wave interactions are accurately captured by incorporating interface evolution into nonlinear wave propagation.

In the family of the diffuse-interface capturing methods, the five-equation models \citep{Allaireetal2002,Massonietal2002,Kapilaetal2001} are widely used \citep{CoralicColonius2014,FriessKokh2014,BeigJohnsen2015,RodriguezJohnsen2019}. On one hand, they are more flexible to use complex equations of state than the $\gamma$-based models \citep{Abgrall1996,Shyue1998,SaurelAbgrall1999-Gamma,AbgrallKarni2001,JohnsenColonius2006}. On the other hand, they directly impose the isobaric closure, which corresponds to the mechanical equilibrium required at material interfaces.
Compared to the five-equation model of Allaire et al. \citep{Allaireetal2002} and of Massoni et al. \citep{Massonietal2002}, the five-equation model of Kapila et al. \citep{Kapilaetal2001} has an extra non-conservative term resulting from the Baer-Nunziato model \citep{BaerNunziato1986} and its corresponding types (also known as the seven-equation models) with a hierarchy of relaxations \citep{SaurelAbgrall1999-Nonequilibrium,Saureletal2003,PerigaudSaurel2005,Saureletal2008,Petitpasetal2009}. This non-conservative term plays a critical role in bubble collapse dynamics \citep{Tiwarietal2013,Schmidmayeretal2020}, while poses numerical challenges in terms of shock computation, volume fraction boundedness, and wave propagation \citep{Petitpasetal2007,Saureletal2009,DalMasoetal1995}. To address these difficulties, the six-equation model is suggested \citep{Saureletal2009}, where only the velocity equilibrium is considered. In numerical practice, the possibly singular non-conservative term is removed from the hyperbolic evolution, but then the isobaric closure (mechanical equilibrium at material interfaces) is enforced via a separate pressure relaxation. Due to the usage of the phasic internal energy \citep{Saureletal2009}, calibration to the relaxation outcome is needed to conserve energy. It is demonstrated in \citep{Schmidmayeretal2020} that this six-equation model also captures bubble collapse correctly. Another six-equation model based on the phasic total energy is proposed \citep{PelantiShyue2014}, which automatically satisfies energy conservation after the relaxation. These six-equation models are further applied to cases with surface tension \citep{Schmidmayeretal2017}, deformable solids \citep{Favrieetal2009}, heat transfer and phase change resulting from additional relaxations of temperature and Gibbs free energy \citep{PelantiShyue2014,Pelanti2022}, and more than two phases \citep{Petitpasetal2009,Ndanouetal2015,PelantiShyue2019,Pelanti2026}.
However, there are theoretical and numerical gaps in terms of the six-equation models, which the present study attempts to fill.

Although the derivations of the two-phase six-equation models \citep{Saureletal2009,PelantiShyue2014} have been detailed by either performing the velocity relaxation on the Baer-Nunziato model \citep{Kapilaetal2001} or following the idea of matching the phasic and mixture momentum equations \citep{DeLorenzoetal2018,KuhnDesjardins2021}, the derivations of their $N$-phase extensions \citep{Petitpasetal2009,Ndanouetal2015,PelantiShyue2019} from first principles are missed. Furthermore, the $N$-phase models \citep{Petitpasetal2009,Ndanouetal2015} are unable to recover their two-phase correspondence \citep{Saureletal2009} when only two phases ($N=2$) are present, which fails the consistency of reduction. The present study attempts to rigorously derive the six-equation model with the Phase-Field mechanism (discussed below) under a general $N$-phase setup based solely on the conservation laws (for the phasic mass, mixture momentum, and phasic total energy) and the second law of thermodynamics. Our derivation illustrates that the interphase interaction terms as well as the evolution of volume fraction are direct consequences of satisfying the second law of thermodynamics.

As a new component of solving the six-equation model, the relaxations in the six-equation models are ordinary differential systems that are ultimately turned into nonlinear algebraic equations numerically. In general, iterative methods like the Newton-Raphson method are used to find the solution \citep{SaurelAbgrall1999-Nonequilibrium,Saureletal2009,Favrieetal2009,Ndanouetal2015}, while exact or approximate closed-form solutions are proposed for some specific equations of state \citep{Saureletal2008,PelantiShyue2014,Pelanti2022,Pelanti2026}. The effect of relaxation parameters is discussed based on numerical observations in \citep{Saureletal2009,Haegemanetal2024}. However, theoretical investigation of the thermodynamical admissibility (i.e., bounded volume fractions with summation to unity, positive mass, positive temperature, and real speed of sound) of the relaxation outcome is missed, although a qualitative discussion is provided in \citep{FonkwaKamga2026} following a brief description in \citep{Ndanouetal2015} for the pressure relaxation. The present study addresses this issue by analyzing the pressure and pressure-temperature relaxations with the equation of state by Le M{\'e}tayer et al. \citep{LeMetayeretal2005}, which covers the commonly used ideal-gas and stiffened-gas equations of state. Our proof shows that there exists a unique solution that is thermodynamically admissible under a general $N$-phase setup for these two relaxations.

Finally, the six-equation models share a key challenge of the diffuse-interface capturing methods for compressible multiphase flows; that is unbounded growth of material interface thickness due to numerical diffusion inherent to the capturing schemes \citep{Harten1977,Harten1978}. As a result, it becomes more ambiguous to distinguish different phases as computation goes on. Following its success in various incompressible multiphase flows \citep{Huang2021,Huangetal2020,Huangetal2020CAC,Huangetal2020N,Huangetal2020B,Huangetal2020NPMC,Huangetal2020Solid,Huangetal2021Contact} and compressible multiphase flows with the five-equation models \citep{Shuklaetal2010,Tiwarietal2013,Jainetal2020,JainMoin2022,HuangJohnsen2022,HuangJohnsen2023,HuangJohnsen2024}, the Phase-Field mechanism is introduced in the present study to address the issue of excessive interface thickening. Beyond its contributions to mass, momentum, and energy transport, the Phase-Field mechanism also exerts an influence on the interphase work, which is revealed by our model derivation. After enforcing the isobaric closure via the pressure relaxation, our six-equation model with the Phase-Field mechanism reduces to a five-equation model that is different from those in previous studies \citep{Shuklaetal2010,Tiwarietal2013,Jainetal2020,JainMoin2022,HuangJohnsen2022,HuangJohnsen2023,HuangJohnsen2024}. This difference stems from the appearance of new non-conservative terms associated with the Phase-Field mechanism in the volume fraction equation of the reduced five-equation model. To the best of our knowledge, such non-conservative terms are identified for the first time in this work; they are otherwise difficult, if not impossible, to directly derive from the five-equation models. Accordingly, when the Phase-Field mechanism is incorporated, the present six-equation model with the pressure relaxation is no longer equivalent and cannot serve as an alternative to the five-equation models \citep{Shuklaetal2010,Tiwarietal2013,Jainetal2020,JainMoin2022,HuangJohnsen2022,HuangJohnsen2023,HuangJohnsen2024} with the non-conservative term of Kapila et al. \citep{Kapilaetal2001}.
We note that while a seven-equation model was actually implemented in \citep{HatashitaJain2025}, the appendix presents a two-phase six-equation model with the Phase-Field mechanism. In that appendix, contributions of the Phase-Field mechanism to kinetic energy are split according to mass fractions, with limited explanations but no detailed derivations. Nevertheless, this two-phase model is mathematically equivalent to our model derived from first principles under the two-phase flow condition.

Together with the theoretical outcomes, the consistent and conservative numerical approach in \citep{HuangJohnsen2022,HuangJohnsen2023,HuangJohnsen2024} is modified and adapted to numerically solve the proposed six-equation model. Different from those \citep{Shuklaetal2010,Tiwarietal2013,Jainetal2020,JainMoin2022} that are designed for a specific formulation of the Phase-Field mechanism, the approach in \citep{HuangJohnsen2022,HuangJohnsen2023,HuangJohnsen2024} provides a unified bound-preserving framework to incorporate different formulations of the Phase-Field mechanism into compressible multiphase flows \citep{Huangetal2024} with possible high-order implementation \citep{Whiteetal2025} and adaptive mesh refinement \citep{Huangetal2025}. Both the mathematical model and numeral approach form a consistent and conservative Phase-Field method, which is extensively examined in various compressible two-phase flows including the bubble collapse problem in \citep{Schmidmayeretal2020}.

The rest of the paper is organized as follows.
In Section~\ref{Sec:GoverningEquations}, the proposed six-equation model with the Phase-Field mechanism is derived and analyzed.
In Section~\ref{Sec:NumericalScheme}, the adaptation of the consistent and conservative numerical approach to the proposed six-equation model is detailed, with emphasis on the analysis of the pressure and pressure-temperature relaxations.
In Section~\ref{Sec:Results}, various compressible two-phase flow benchmarks are performed to demonstrate the proposed model and approach, and the present study is concluded in Section~\ref{Sec:Conclusions}.

\section{Model derivation and analysis}\label{Sec:GoverningEquations}
In this section, we derive and analyze the six-equation model with the Phase-Field mechanism for compressible multiphase flows in a general $N$-phase setup. The derivation is based on the consistency conditions, which were originally proposed for incompressible multiphase flows \citep{Huangetal2020,Huangetal2020CAC,Huangetal2020N,Huangetal2020B,Huangetal2020NPMC,Huangetal2020Solid,Huangetal2021Contact,Huang2021} and recently applied to the five-equation models for compressible multiphase flows \citep{HuangJohnsen2022,HuangJohnsen2023,HuangJohnsen2024,Whiteetal2025,Huangetal2025}, and neglects the effects of surface tension and phase changes.
Starting with mass, momentum, and energy balance, the proposed model is completed with the second law of thermodynamics, followed by analysis of Galilean invariance, consistency of reduction \citep{BoyerMinjeaud2014,Dong2018,Huangetal2020N,Huangetal2020B}, isobaric closure via relaxation, and incompressible limit.

\subsection{Conservation laws}\label{Sec:ConservationLaw}
Without considering any phase change, the mass of individual phases is conserved, following
\begin{equation}\label{Eq:Mass-Phase}
\frac{\partial (\alpha_p \rho_p)}{\partial t}
+
\nabla \cdot \left( (\alpha_p\mathbf{u} - \mathbf{J}_p)\rho_p \right)
=
0,
\end{equation}
where $\mathbf{u}$ is the flow velocity, and $\alpha_p$, $\rho_p$, and $\mathbf{J}_p$ are the volume fraction, phasic density, and Phase-Field mechanism of Phase~$p$.
The mass flux now incorporates the Phase-Field mechanism that so far can be any formulation but needs to vanish at the domain boundary, i.e., has the zero-flux boundary condition. As a result, only the flow velocity can drive the phases entering/leaving the domain.
In the present study, $(\alpha_p\rho_p)$ is called the mass of Phase~$p$.
As the mixture density is defined as $\rho = \sum_{q=1}^N (\alpha_q\rho_q)$, Eq.~(\ref{Eq:Mass-Phase}) implies mass conservation of the multiphase mixture after summing Eq.~(\ref{Eq:Mass-Phase}) over all the phases.

The consistency conditions \citep{Huangetal2020,Huangetal2020N,Huangetal2020NPMC,Huangetal2020Solid,HuangJohnsen2022} require that momentum and energy are transported by the same mass flux in Eq.~(\ref{Eq:Mass-Phase}).
As a result, with the assumption of velocity equilibrium between phases, the mixture momentum balance from the Newton's second law follows
\begin{equation}\label{Eq:Momentum}
\frac{\partial (\rho \mathbf{u})}{\partial t}
+
\nabla \cdot \left(\sum_{q=1}^N (\alpha_q\mathbf{u} - \mathbf{J}_q)\rho_q \otimes \mathbf{u}\right)
=
\nabla \cdot \left(\sum_{q=1}^N \alpha_q \boldsymbol{\sigma}_q\right),
\end{equation}
where $\boldsymbol{\sigma}_p$ is the phasic stress tensor of Phase~$p$. Without considering any external (body) force, the mixture momentum is conserved.

Similarly, the total energy of individual phases from the first law of thermodynamics is governed by
\begin{equation}\label{Eq:Energy-Phase}
\frac{\partial (\alpha_p \rho_p E_p)}{\partial t}
+
\nabla \cdot \left( (\alpha_p \mathbf{u} - \mathbf{J}_p) \rho_p E_p \right)
=
\nabla \cdot (\alpha_p \mathbf{u} \cdot \boldsymbol{\sigma}_p)
-
\nabla \cdot (\alpha_p \mathbf{Q}_p)
+
\nabla \cdot \mathbf{J}_p^{E}
+
\sum_{q=1}^N W_{p \leftarrow q}^I
+
\sum_{q=1}^N Q_{p \leftarrow q}^I,
\end{equation}
where $E_p = e_p + \frac{1}{2} \mathbf{u} \cdot \mathbf{u}$ is the specific total energy, $e_p$ is the specific internal energy, $\mathbf{Q}_p$ is the heat flux, and $\mathbf{J}_p^E$ is the energy flux due to the Phase-Field mechanism, of Phase~$p$, and $\mathbf{W}_{p \leftarrow q}^I$ and $Q_{p \leftarrow q}^I$ are the work and heat transfer, respectively, applied to Phase~$p$ from Phase~$q$.
The notation for the interfacial quantities in the present study implies $\xi_{p \leftarrow q}^I=-\xi_{q \leftarrow p}^I$ and $\xi_{p,q}^I=\xi_{q,p}^I$. As a result, the mixture total energy is conserved without considering any external energy source, following
\begin{equation}\label{Eq:Energy}
\frac{\partial (\rho E)}{\partial t}
+
\nabla \cdot \left( \mathbf{u} \rho E - \sum_{q=1}^N (\mathbf{J}_q \rho_q E_q) \right)
=
\nabla \cdot (\mathbf{u} \cdot \boldsymbol{\sigma})
-
\nabla \cdot \mathbf{Q}
+
\nabla \cdot \mathbf{J}^{E},
\quad
\rho E = \sum_{q=1}^N (\alpha_q \rho_q E_q),
\quad
\mathbf{Q}=\sum_{q=1}^N (\alpha_q \mathbf{Q}_q)
\quad
\mathbf{J}^{E}=\sum_{q=1}^N \mathbf{J}_q^{E},
\end{equation}
from summing Eq.~(\ref{Eq:Energy-Phase}) over all the phases.

\subsection{Second law of thermodynamics}\label{Sec:SeconLaw}
With the mass, momentum, and energy balance in Section~\ref{Sec:ConservationLaw}, we apply the second law of thermodynamics to determine the evolution of volume fraction, the work from Phase~$q$ to Phase~$p$, and the constraints on the stress tensor, heat flux, and Phase-Field mechanism.
The second law of thermodynamics requires
\begin{equation}\label{Eq:Entropy-Phase}
\Sigma_p
=
\frac{\partial (\alpha_p \rho_p s_p) }{\partial t}
+
\nabla \cdot \left( (\alpha_p\mathbf{u}-\mathbf{J}_p) \rho_p s_p \right)
-
\nabla \cdot \left( -\frac{\alpha_p\mathbf{Q}_p}{T_p} \right)
-
\sum_{q=1}^N \frac{Q_{p \leftarrow q}^I}{T_q}
-
\nabla \cdot \mathbf{J}_p^{S}
\geqslant 0,
\end{equation}
where $s_p$ is the specific entropy, $T_p$ is the temperature, $\mathbf{J}_p^{S}$ is the entropy flux due the Phase-Field mechanism, and $\Sigma_p$ is the entropy production, of Phase~$p$.
We note that the entropy is again transported by the same mass flux in Eq.~(\ref{Eq:Mass-Phase}).
Using the mass, momentum, and energy balance (Eq.~(\ref{Eq:Mass-Phase}), Eq.~(\ref{Eq:Momentum}), and Eq.~(\ref{Eq:Energy-Phase}) in Section~\ref{Sec:ConservationLaw}) and the thermodynamic relation $de_p = T_p ds_p + \frac{P_p}{\rho_p^2} d\rho_p$, where $P_p$ is the thermodynamic pressure of Phase~$p$, Eq.~(\ref{Eq:Entropy-Phase}) becomes
\begin{equation}\label{Eq:SecondLaw}
\begin{split}
\Sigma_p
=
\frac{\alpha_p \left(
P_p \mathbf{I}
+
\boldsymbol{\sigma}_p
\right) : \nabla \mathbf{u}}{T_p}
-
\frac{\alpha_p \mathbf{Q}_p \cdot \nabla T_p}{T_p^2}
+
\sum_{q=1}^N \frac{(T_q -T_p) Q_{p \leftarrow q}^I}{T_p T_q}\\
+
\frac{\nabla \cdot \left(\mathbf{J}_p^{E} - (\alpha_p P_p) \sum_{q=1}^N \mathbf{J}_q \right)}{T_p}
+
\frac{\nabla (\alpha_p P_p) \cdot \sum_{q=1}^N \mathbf{J}_q}{T_p}
-
\nabla \cdot \mathbf{J}_p^{S}\\
+
\frac{1}{T_p} \sum_{q=1}^N \left(
W_{p \leftarrow q}^I
-
\mathbf{u} \cdot \left(
\frac{(\alpha_p \rho_p)}{\rho} \nabla \cdot (\alpha_q \boldsymbol{\sigma}_q)
-
\frac{(\alpha_q\rho_q)}{\rho} \nabla \cdot (\alpha_p \boldsymbol{\sigma}_p)
\right)\right.\\
\left.
-
\mathbf{u} \cdot \left(
\frac{(\alpha_p \rho_p)}{\rho} (\mathbf{J}_q \rho_q) \cdot  \nabla \mathbf{u}
-
\frac{(\alpha_q \rho_q)}{\rho} (\mathbf{J}_p \rho_p) \cdot \nabla \mathbf{u}
\right)
+
P^I_{p,q} \left(\frac{\mathcal{D}\alpha}{\mathcal{D}t}\right)^I_{p \leftarrow q}
\right)\\
+
\frac{1}{T_p} \left(
P_p \left(
\frac{\partial \alpha_p}{\partial t}
+
\nabla \cdot (\alpha_p \mathbf{u} - \mathbf{J}_p)
-
\alpha_p \nabla \cdot \left(
\mathbf{u}
-
\sum_{q=1}^N \mathbf{J}_q
\right)
\right)
-
\sum_{q=1}^N P^I_{p,q} \left(\frac{\mathcal{D}\alpha}{\mathcal{D}t}\right)^I_{p \leftarrow q}
\right)
\geqslant
0,
\end{split}
\end{equation}
where $P^I_{p,q}$ is the pressure at the interface separating Phases~$p$ and $q$, and $\left(\mathcal{D}\alpha/\mathcal{D}t\right)^I_{p \leftarrow q}$ denotes the time rate change of volume of Phase~$p$ due to Phase~$q$.

The terms on the second row of Eq.~(\ref{Eq:SecondLaw}) are related to the Phase-Field mechanism. The simplest option to satisfy the second law under different pressure is
\begin{equation}\label{Eq:SecondLaw-PhaseField}
\sum_{q=1}^N \mathbf{J}_q = \mathbf{0} \quad \mathrm{or} \quad \sum_{q=1}^N \nabla \cdot \mathbf{J}_q = 0.
\end{equation}
As a result, there is no need to introduce extra energy and entropy fluxes, i.e., $\mathbf{J}_p^E=\mathbf{0}$ and $\mathbf{J}_p^S=\mathbf{0}$.
The terms on the third and fourth rows of Eq.~(\ref{Eq:SecondLaw}) determine the work applied to Phase~$p$ from Phase~$q$, which reads
\begin{equation}\label{Eq:Work-Interface}
W_{p \leftarrow q}^I
=
\mathbf{u} \cdot \left(
\frac{(\alpha_p \rho_p)}{\rho} \nabla \cdot (\alpha_q \boldsymbol{\sigma}_q)
-
\frac{(\alpha_q \rho_q)}{\rho} \nabla \cdot (\alpha_p \boldsymbol{\sigma}_p)
\right)
+
\mathbf{u} \cdot \left(
\frac{(\alpha_p \rho_p)}{\rho} (\mathbf{J}_q \rho_q) \cdot  \nabla \mathbf{u}
-
\frac{(\alpha_q \rho_q)}{\rho} (\mathbf{J}_p \rho_p) \cdot \nabla \mathbf{u}
\right)
-
P_{p,q}^I \left(\frac{\mathcal{D}\alpha}{\mathcal{D}t}\right)^I_{p \leftarrow q}.
\end{equation}
It is again straightforward to verify that $W_{p \leftarrow q}^I$ in Eq.~(\ref{Eq:Work-Interface}) satisfies $W_{p \leftarrow q}^I=-W_{q \leftarrow p}^I$ and thus $\sum_{p,q=1}^N W_{p \leftarrow q}^I=0$, i.e., there is no net work produced between the phases. We note that $W_{p \leftarrow q}^I$ includes the contribution from the Phase-Field mechanism to satisfy the second law.
Due to the fact that the volume change of Phase~$p$ is equal to the summation of volume exchanges contributed from all other phases, we have
\begin{equation}\label{Eq:VolumeChange}
\frac{\partial \alpha_p}{\partial t}
+
\nabla \cdot (\alpha_p \mathbf{u} - \mathbf{J}_p)
-
\alpha_p \nabla \cdot \left(
\mathbf{u}
-
\sum_{q=1}^N \mathbf{J}_q
\right)
=
\sum_{q=1}^N \left(\frac{\mathcal{D}\alpha}{\mathcal{D}t}\right)^I_{p \leftarrow q}.
\end{equation}
Summing Eq.~(\ref{Eq:VolumeChange}) over $p$, we have $\partial (\sum_{q=1}^N \alpha_q)/\partial t=0$ due to $\left(\mathcal{D}\alpha/\mathcal{D}t\right)^I_{p \leftarrow q}=-\left(\mathcal{D}\alpha/\mathcal{D}t\right)^I_{q \leftarrow p}$. Therefore, $\sum_{q=1}^N \alpha_q=1$ is preserved.
Incorporating Eq.~(\ref{Eq:SecondLaw-PhaseField}), Eq.~(\ref{Eq:Work-Interface}), and Eq.~(\ref{Eq:VolumeChange}) into Eq.~(\ref{Eq:SecondLaw}), we obtain
\begin{equation}\label{Eq:SecondLaw-Final}
\Sigma_p
=
\frac{\alpha_p \boldsymbol{\tau}_p : \nabla \mathbf{u}}{T_p}
-
\frac{\alpha_p \mathbf{Q}_p \cdot \nabla T_p}{T_p^2}
+
\sum_{q=1}^N \frac{(T_q -T_p) Q_{p \leftarrow q}^I}{T_p T_q}
+
\frac{1}{T_p} \sum_{q=1}^N \left(P_p - P_{p,q}^I \right) \left(\frac{\mathcal{D}\alpha}{\mathcal{D}t}\right)^I_{p \leftarrow q}
\geqslant
0,
\end{equation}
where $\boldsymbol{\tau}_p = P_p \mathbf{I} + \boldsymbol{\sigma}_p$ is the stress tensor of Phase~$p$ excluding the thermodynamic pressure.
It is clear that the Newtonian viscous stress ($\boldsymbol{\tau}_p=\mu_p (\nabla \mathbf{u} + \nabla \mathbf{u}^T) + \mu_p^B (\nabla \cdot \mathbf{u}) \mathbf{I}$, where $\mu_p$ and $\mu_p^B$ are the shear and bulk viscosities of Phase~$p$), the Fourier's law of heat conduction ($\mathbf{Q}_p=-\kappa_p \nabla T_p$, where $\kappa_p$ is the heat conductivity of Phase~$p$), and the Newton's law of cooling ($Q_{p \leftarrow q}^I=h_{p,q}^I \delta_{p,q}^I(T_q -T_p)$, where $h_{p,q}^I$ is the convective heat transfer coefficient and $\delta_{p,q}^I$ is the surface delta function, of Phases$p$ and $q$) are admissible candidates.
In an extreme case where viscosity, heat conduction, and heat transfer are neglected, the last term in Eq.~(\ref{Eq:SecondLaw-Final}) must be non-negative, resulting in
\begin{equation}\label{Eq:VolumeChange-pq}
\left(\frac{\mathcal{D}\alpha}{\mathcal{D}t}\right)^I_{p \leftarrow q}
=
\zeta_{p,q}^I (P_p - P_q),
\quad
\zeta_{p,q}^I \geqslant 0,
\quad
\min(P_p,P_q) \leqslant P_{p,q}^I \leqslant \max(P_p,P_q).
\end{equation}
Combining Eq.~(\ref{Eq:SecondLaw-PhaseField}), Eq.~(\ref{Eq:VolumeChange}), and Eq.~(\ref{Eq:VolumeChange-pq}), the volume fraction equation is finalized.

\subsection{Proposed model}\label{Sec:Model}
Incorporating the analysis of the second law of thermodynamics into the conservation of mass, momentum, and energy, we obtain the final governing equations written in a vector form
\begin{equation}\label{Eq:GoverningEquations}
\frac{\partial \mathbf{U}}{\partial t}
+
\nabla \cdot \mathbf{F}^{HB}
=
\mathbf{S}^{HB}
+
\nabla \cdot \mathbf{F}^{DF}
+
\mathbf{S}^{DF}
+
\nabla \cdot \mathbf{F}^{PF}
+
\mathbf{S}^{PF}
+
\mathbf{S}^{RX},
\end{equation}
where $\mathbf{U}$ is the vector of conservative variables containing the phasic masses, mixture momentum, phasic total energies, and volume fractions
\begin{equation}\label{Eq:Conservative}
\mathbf{U}=\left[\left\{(\alpha_p \rho_p)\right\}_{p=1}^N,(\rho \mathbf{u}),\left\{(\alpha_p \rho_p E_p)\right\}_{p=1}^N,\left\{\alpha_p\right\}_{p=1}^N\right]^T,
\end{equation}
$\mathbf{F}^{HB}$ and $\mathbf{S}^{HB}$ are the hyperbolic flux vector and hyperbolic non-conservative vector
\begin{equation}\label{Eq:Hyperbolic}
\begin{split}
\mathbf{F}^{HB}&=\left[\left\{(\alpha_p\rho_p) \mathbf{u}\right\}_{p=1}^N,\rho \mathbf{u}\otimes \mathbf{u} + \sum_{q=1}^N (\alpha_q P_q) \mathbf{I},\left\{( (\alpha_p \rho_p E_p)  + (\alpha_p P_p) )\mathbf{u}\right\}_{p=1}^N,\left\{\alpha_p \mathbf{u}\right\}_{p=1}^N\right]^T,\\
\mathbf{S}^{HB}&=\left[\left\{0\right\}_{p=1}^N,\mathbf{0},\left\{\sum_{q=1}^N \mathbf{u} \cdot \left(
\frac{(\alpha_q \rho_q)}{\rho} \nabla (\alpha_p P_p)
-
\frac{(\alpha_p \rho_p)}{\rho} \nabla (\alpha_q P_q)
\right)\right\}_{p=1}^N,\left\{\alpha_p \nabla \cdot \mathbf{u}\right\}_{p=1}^N\right]^T,
\end{split}
\end{equation}
$\mathbf{F}^{DF}$ and $\mathbf{S}^{DF}$ are the diffusive flux vector and diffusive non-conservative vector
\begin{equation}\label{Eq:Diffusion}
\begin{split}
\mathbf{F}^{DF}&=\left[\left\{0\right\}_{p=1}^N,\sum_{q=1}^N (\alpha_q \boldsymbol{\tau}_q),\left\{\mathbf{u} \cdot (\alpha_p \boldsymbol{\tau}_p) - (\alpha_p \mathbf{Q}_p)\right\}_{p=1}^N,\left\{0\right\}_{p=1}^N\right]^T,\\
\mathbf{S}^{DF}&=\left[\left\{0\right\}_{p=1}^N,\mathbf{0},\left\{\sum_{q=1}^N \mathbf{u} \cdot \left(
\frac{(\alpha_p \rho_p)}{\rho} \nabla \cdot (\alpha_q \boldsymbol{\tau}_q)
-
\frac{(\alpha_q \rho_q)}{\rho} \nabla \cdot (\alpha_p \boldsymbol{\tau}_p)
\right)
\right\}_{p=1}^N,\left\{0\right\}_{p=1}^N\right]^T,
\end{split}
\end{equation}
$\mathbf{F}^{PF}$ and $\mathbf{S}^{PF}$ are the Phase-Field flux vector and Phase-Field non-conservative vector
\begin{equation}\label{Eq:PhaseField}
\begin{split}
\mathbf{F}^{PF}&=\left[\left\{(\mathbf{J}_p\rho_p)\right\}_{p=1}^N,\sum_{q=1}^N (\mathbf{J}_q\rho_q) \otimes \mathbf{u},\left\{(\mathbf{J}_p \rho_p E_p)\right\}_{p=1}^N,\left\{\mathbf{J}_p\right\}_{p=1}^N\right]^T,\\
\mathbf{S}^{PF}&=\left[\left\{0\right\}_{p=1}^N,\mathbf{0},\left\{\sum_{q=1}^N \mathbf{u} \cdot \left(
\frac{(\alpha_p \rho_p)}{\rho} (\mathbf{J}_q \rho_q) \cdot \nabla \mathbf{u}
-
\frac{(\alpha_q \rho_q)}{\rho} (\mathbf{J}_p \rho_p) \cdot \nabla \mathbf{u}
\right)\right\}_{p=1}^N,\left\{0\right\}_{p=1}^N\right]^T,
\end{split}
\end{equation}
and $\mathbf{S}^{RX}$ is the relaxation vector
\begin{equation}\label{Eq:Relaxation}
\mathbf{S}^{RX}=\left[\left\{0\right\}_{p=1}^N,\mathbf{0},\left\{-\sum_{q=1}^N \zeta_{p,q}^I P_{p,q}^I (P_p - P_q)+\sum_{q=1}^N Q_{p \leftarrow q}^I\right\}_{p=1}^N,\left\{\sum_{q=1}^N \zeta_{p,q}^I (P_p - P_q)\right\}_{p=1}^N\right]^T.
\end{equation}
Considering that each phase is a pure substance, any thermodynamical properties are uniquely determined by the phasic density and pressure from the state principle. After supplementing the equations of state $e_p=\mathsf{e}_p(P_p,\rho_p)$ and $T_p=\mathsf{T}_p(P_p,\rho_p)$ for each phase, the proposed model is complete. The phasic sound speed of Phase~$p$, denoted by $c_p$, is derived from the equations of state following
\begin{equation}\label{Eq:SoundSpeed-Phase}
c_p^2 = \frac{\frac{P_p}{\rho_p^2} - \left(\frac{\partial e_p}{\partial \rho_p}\right)_{P_p}}{\left(\frac{\partial e_p}{\partial P_p}\right)_{\rho_p}}.
\end{equation}

The proposed model is general to include an arbitrary number of phases with different material properties and equations of state and to admit any formulation of the Phase-Field mechanism under the requirement in Eq.~(\ref{Eq:SecondLaw-PhaseField}) from the second law of thermodynamics. In what follows, we further analyze the properties of the proposed model in addition to conservation and second law of thermodynamics.

\subsection{Galilean invariance}\label{Sec:Galilean}
Galilean invariance is another important property that must be satisfied by multiphase flow models.
As required by continuum mechanics \citep{Reddy2013}, stress tensor, heat flux, and heat transfer should be objective (or frame indifferent). Commonly, the Phase-Field mechanism depends on the volume fractions and their spatial derivatives, which are also frame indifferent. With these normal requirements, the proof of Galilean invariance of the proposed model is straightforward using the Galilean transformation, and is supplemented in \ref{Appendix:Galilean}, as the procedure is essentially the same as those in \citep{Huangetal2020NPMC,HuangJohnsen2022}.

As a direct consequence of Galilean invariance, the proposed model satisfies the kinematic, mechanical, and thermal equilibria at isolated interfaces (or the uniform flow condition). Specifically, the proposed model admits the following solution
\begin{equation}\label{Eq:UniformFlow}
\rho_p = \rho_{p}^{(0)},\quad
\mathbf{u} = \mathbf{u}^{(0)},\quad
P_p = P^{(0)},\quad
T_p = T^{(0)},\quad
\alpha_p = \alpha_p^{(0)}(\mathbf{x}-\mathbf{u}_0t),\quad
\forall t>0,
\end{equation}
no matter how the volume fractions are initialized, where $\mathbf{u}^{(0)}$, $\rho_{p}^{(0)}$, $P^{(0)}$, and $T^{(0)}$ are admissible uniform initial conditions.
It is obvious that Eq.~(\ref{Eq:UniformFlow}) with $\mathbf{u}_0=\mathbf{0}$ is a solution of the proposed model. Due to Galilean invariance, Eq.~(\ref{Eq:UniformFlow}) with any constant $\mathbf{u}_0$ is also a solution.
As a result, velocity, pressure, and temperature are continuous across isolated material interfaces.
Furthermore, if the interfacial heat transfer is neglected (i.e., $Q_{p \leftarrow q}^I=0$) and each phase initially has a different uniform temperature $T_p^{(0)}$, then the temperature solution in Eq.~(\ref{Eq:UniformFlow}) becomes $T_p=T_p^{(0)}$; the uniform temperature of each phase is preserved.

\subsection{Consistency of reduction}\label{Sec:Reduction}
Consistency of reduction \citep{BoyerLapuerta2006,BoyerMinjeaud2014,Dong2018,Huangetal2020N} in the local sense \citep{Huangetal2020B} is another important property that must be satisfied by multiphase flow models, so that locally absent phases are not fictitiously produced to affect the multiphase dynamics of present phases. As a result, the single-phase dynamics is recovered in bulk-phase regions.

We first consider the case where the Phase-Field mechanism is excluded, i.e., $\mathbf{J}_p=\mathbf{0}$ for $1 \leqslant p \leqslant N$. When Phase~$q$ is locally absent, i.e., $\alpha_q=0$ and $|\nabla \alpha_q| = 0$, at a specific location, we need to have $\zeta_{p,q}^I=0$ for $1 \leqslant p \leqslant N$, so that $\partial \alpha_q/\partial t=0$ and Phase~$q$ has no effect on the volume fraction equation in Eq.~(\ref{Eq:GoverningEquations}) for Phase $p \neq q$. Then, it is straightforward to show that the mass, momentum, and energy equations in Eq.~(\ref{Eq:GoverningEquations}) are not affected by Phase~$q$ either. As a result, Phase~$q$ remains locally absent, and the $N$-phase system automatically reduces to the corresponding $(N-1)$-phase system at that location. One can repeat this process until there is a single phase, and the proposed model in Eq.~(\ref{Eq:GoverningEquations}) reduces to the Navier-Stokes equations, ensuring that the single-phase dynamics is correctly recovered.
We note that, in the absence of the Phase-Field mechanism, the proposed model in Eq.~(\ref{Eq:GoverningEquations}) is equivalent to the one in \citep{PelantiShyue2019}, an $N$-phase extension of the two-phase six-equation models based on the phasic total energy \citep{PelantiShyue2014} and on the phasic internal energy \citep{Saureletal2009}, while the $N$-phase model in \citep{Petitpasetal2009,Ndanouetal2015} is reduction \textit{inconsistent} with these two two-phase models.

To preserve the consistency of reduction of the proposed model in Eq.~(\ref{Eq:GoverningEquations}) after including the Phase-Field mechanism, $\mathbf{J}_p$ for $1 \leqslant p \leqslant N$ itself should be reduction consistent, following the requirement in \citep{Dong2018,Huangetal2020B,HuangJohnsen2022}. Specifically, suppose that $\mathbf{J}_p$ depends on $\nabla^m \alpha_r$ for $1 \leqslant r \leqslant N$ and $0 \leqslant m \leqslant M$, when Phase~$q$ is locally absent, i.e., $|\nabla^m \alpha_q|=0$ for $0 \leqslant m \leqslant M$, it is required that $\mathbf{J}_q = \mathbf{0}$ and $\mathbf{J}_{p \neq q}$ reduces to the corresponding $(N-1)$-phase formulation.
Therefore, as long as the chosen Phase-Field mechanism is locally reduction consistent, the proposed model in Eq.~(\ref{Eq:GoverningEquations}) shares the same property.

\subsection{Isobaric closure via relaxation}\label{Sec:PressureRelaxation}
In this section, we analyze the isobaric closure of the proposed model in Eq.~(\ref{Eq:GoverningEquations}), achieved by setting the relaxation coefficient $\zeta_{p,q}^I$ to infinity, for $1 \leqslant p,q \leqslant N$.
We start with obtaining the phasic internal energy equation from Eq.~(\ref{Eq:GoverningEquations}), which reads
\begin{equation}\label{Eq:InternalEnergy-Phase}
\frac{\partial (\alpha_p \rho_p e_p)}{\partial t}
+
\nabla \cdot \left( (\alpha_p \mathbf{u} - \mathbf{J}_p) \rho_p e_p \right)
=
\alpha_p \boldsymbol{\sigma}_p : \nabla \mathbf{u}
-
\nabla \cdot (\alpha_p \mathbf{Q}_p)
-
\sum_{q=1}^N \zeta_{p,q}^I P_{p,q}^I (P_p - P_q)
+
\sum_{q=1}^N Q_{p \leftarrow q}^I.
\end{equation}
Then, using the equations of state, we obtain the phasic pressure equation
\begin{equation}\label{Eq:Pressure-Phase}
\begin{split}
\frac{\partial P_p}{\partial t}
=
-
\rho_p c_p^2 \nabla \cdot \mathbf{u}
+
\frac{
\alpha_p \boldsymbol{\tau}_p : \nabla \mathbf{u}
-
\nabla \cdot (\alpha_p \mathbf{Q}_p)
+
\sum_{q=1}^N Q_{p \leftarrow q}^I
}{(\alpha_p \rho_p) \left(\frac{\partial e_p}{\partial P_p}\right)_{\rho_p}}\\
-
\sum_{q=1}^N \frac{1}{\alpha_p} 
\frac{
\frac{P_{p,q}^I}{\rho_p}
-
\rho_p \left(\frac{\partial e_p}{\partial \rho_p}\right)_{P_p}
}{\left(\frac{\partial e_p}{\partial P_p}\right)_{\rho_p}} \zeta_{p,q}^I (P_p - P_q)
-
\frac{\alpha_p \mathbf{u} - \mathbf{J}_p}{\alpha_p} \cdot \nabla P_p.
\end{split}
\end{equation}
The pressure relaxation is performed on the pressure difference between Phases~$p$ and $q$ governed by
\begin{equation}\label{Eq:PressureDifference}
\begin{split}
\frac{\partial (P_p-P_q)}{\partial t}
=
-
\rho_p c_p^2 \nabla \cdot \mathbf{u}
+
\rho_q c_q^2 \nabla \cdot \mathbf{u}\\
+
\frac{
\alpha_p \boldsymbol{\tau}_p : \nabla \mathbf{u}
-
\nabla \cdot (\alpha_p \mathbf{Q}_p)
+
\sum_{q=1}^N Q_{p \leftarrow q}^I
}{(\alpha_p \rho_p) \left(\frac{\partial e_p}{\partial P_p}\right)_{\rho_p}}
-
\frac{
\alpha_q \boldsymbol{\tau}_q : \nabla \mathbf{u}
-
\nabla \cdot (\alpha_q \mathbf{Q}_q)
+
\sum_{p=1}^N Q_{q \leftarrow p}^I
}{(\alpha_q \rho_q) \left(\frac{\partial e_q}{\partial P_q}\right)_{\rho_q}}\\
-
\sum_{q=1}^N \frac{1}{\alpha_p} 
\frac{
\frac{P_{p,q}^I}{\rho_p}
-
\rho_p \left(\frac{\partial e_p}{\partial \rho_p}\right)_{P_p}
}{\left(\frac{\partial e_p}{\partial P_p}\right)_{\rho_p}} \zeta_{p,q}^I (P_p - P_q)
+
\sum_{p=1}^N \frac{1}{\alpha_q} 
\frac{
\frac{P_{q,p}^I}{\rho_q}
-
\rho_q \left(\frac{\partial e_q}{\partial \rho_q}\right)_{P_q}
}{\left(\frac{\partial e_q}{\partial P_q}\right)_{\rho_q}} \zeta_{q,p}^I (P_q - P_p)\\
-
\frac{\alpha_p \mathbf{u} - \mathbf{J}_p}{\alpha_p} \cdot \nabla P_p
+
\frac{\alpha_q \mathbf{u} - \mathbf{J}_q}{\alpha_q} \cdot \nabla P_q.
\end{split}
\end{equation}
We consider $\zeta_{p,q}^I = \frac{1}{\epsilon_{p,q}^I} \rightarrow +\infty$ for all $p$ and $q$, and any quantity can be asymptotically expended as $f=f_{p,q}^{(0)} + \sum_{n_a=1}^{\infty} (\epsilon_{p,q}^I)^{n_a} f_{p,q}^{(n_a)}$.
As a result, Eq.~(\ref{Eq:PressureDifference}) requires $P_p = P_q = P$ for all $p$ and $q$ on the leading order of the asymptotic expansion, which is the isobaric closure. Due to the second law of thermodynamics discussed in Section~\ref{Sec:SeconLaw}, $P_{p,q}^I=P$ is true for its leading order. Using the fact that $\sum_{p,q=1}^N \zeta_{p,q}^I (P_p - P_q)=0$, we finally obtain
\begin{equation}\label{Eq:PressureRelaxation}
\begin{split}
\lim_{\zeta_{p,q}^I \rightarrow +\infty, \forall p,q} \sum_{q=1}^N \zeta_{p,q}^I (P_p - P_q)
=
\alpha_p \left(\frac{\rho c^2}{\rho_p c_p^2} - 1\right) \nabla \cdot \mathbf{u}
+
\frac{1}{\rho_p c_p^2} \left(\mathbf{J}_p - \sum_{q=1}^N \frac{\alpha_p \rho c^2}{\rho_q c_q^2} \mathbf{J}_q \right) \cdot \nabla P\\
+
\frac{1}{\rho_p c_p^2} \left(
\frac{
\alpha_p \boldsymbol{\tau}_p : \nabla \mathbf{u}
-
\nabla \cdot (\alpha_p \mathbf{Q}_p)
+
\sum_{q=1}^N Q_{p \leftarrow q}^I
}{\rho_p \left(\frac{\partial e_p}{\partial P_p}\right)_{\rho_p}}
-
\sum_{q=1}^N
\frac{\alpha_p \rho c^2}{\rho_q c_q^2}
\frac{
\alpha_q \boldsymbol{\tau}_q : \nabla \mathbf{u}
-
\nabla \cdot (\alpha_q \mathbf{Q}_q)
+
\sum_{p=1}^N Q_{q \leftarrow p}^I
}{\rho_q \left(\frac{\partial e_q}{\partial P_q}\right)_{\rho_q}}
\right),
\end{split}
\end{equation}
from Eq.~(\ref{Eq:PressureDifference}), where $c$ is the Wood sound speed \citep{Wood1930} defined by $\rho c^2 = \left(\sum_{q=1}^N \frac{\alpha_q}{\rho_q c_q^2}\right)^{-1}$. We note that all the quantities on the right-hand side of Eq.~(\ref{Eq:PressureRelaxation}) is on their leading order of the asymptotic expansion.

As a result, we obtain a reduced five-equation model including the phasic mass equation in Eq.~(\ref{Eq:Mass-Phase}), the mixture momentum equation in Eq.~(\ref{Eq:Momentum}), the mixture total energy equation in Eq.~(\ref{Eq:Energy}), and the volume fraction equation in Eq.~(\ref{Eq:VolumeChange}) with $\sum_{q=1}^N \left(\frac{\mathcal{D}\alpha}{\mathcal{D}t}\right)^I_{p \leftarrow q}=\lim_{\zeta_{p,q}^I \rightarrow +\infty, \forall p,q} \sum_{q=1}^N \zeta_{p,q}^I (P_p - P_q)$ in Eq.~(\ref{Eq:PressureRelaxation}).
The first term in Eq.~(\ref{Eq:PressureRelaxation}) is the $N$-phase extension of the well-known non-conservative term of the five-equation model of Kapila et al. \citep{Kapilaetal2001}.
The third term shows the effect of the stress tensor (excluding the thermodynamic pressure), heat conduction, and heat transfer. If only the heat transfer is taken into account, i.e., $\boldsymbol{\tau}_p=\mathbf{0}$ and $\mathbf{Q}_p = \mathbf{0}$, our formulation is the same as that derived in \citep{Petitpasetal2009}.
The second term appears due to the Phase-Field mechanism. To the best of our knowledge, this term was not reported in previous five-equation models with the Phase-Field mechanism \citep{Shuklaetal2010,Tiwarietal2013,Jainetal2020,JainMoin2022,HuangJohnsen2022,HuangJohnsen2023,HuangJohnsen2024}. Therefore, the present six-equation model with the pressure relaxation should not be simply considered as an alternative of the five-equation models \citep{Shuklaetal2010,Tiwarietal2013,Jainetal2020,JainMoin2022,HuangJohnsen2022,HuangJohnsen2023,HuangJohnsen2024} with the non-conservative term of Kapila et al. \citep{Kapilaetal2001}.

\subsection{Incompressible limit}\label{Sec:Incompressible}
In this section, we discuss the behavior of the proposed model in Eq.~(\ref{Eq:GoverningEquations}) when some or all of the phases are incompressible.
Suppose that Phase~$p$ is incompressible, i.e., $\rho_p$ is a constant, its mass equation becomes
\begin{equation}\label{Eq:Incompressible-Mass}
\frac{\partial \alpha_p}{\partial t}
+
\nabla \cdot \left( \alpha_p\mathbf{u} - \mathbf{J}_p \right)
=
0,
\end{equation}
which actually governs the evolution of its volume fraction. As a result, the volume fraction equation of Phase~$p$ becomes
\begin{equation}\label{Eq:Incompressible-VolumeFraction}
\alpha_p \nabla \cdot \mathbf{u}
=
-\sum_{q=1}^N \zeta_{p,q}^I (P_p - P_q).
\end{equation}
The equation of state for incompressible substances is $\rho_p (e_p-e_p^\infty)=\rho_p C_p^V T_p$, where $e_p^\infty$ is the reference internal energy and $C_p^V$ is the specific heat capacity at constant volume. It is more common to use the internal energy equation for incompressible substances, and Eq.~(\ref{Eq:InternalEnergy-Phase}) becomes
\begin{equation}\label{Eq:Incompressible-InternalEnergy}
\frac{\partial (\alpha_p \rho_p C_p^V T_p)}{\partial t}
+
\nabla \cdot \left( (\alpha_p \mathbf{u} - \mathbf{J}_p) \rho_p C_p^V T_p \right)
=
\alpha_p \boldsymbol{\tau}_p : \nabla \mathbf{u}
+
\sum_{q=1}^N \zeta_{p,q}^I (P_p - P_{p,q}^I) (P_p - P_q)
-
\nabla \cdot (\alpha_p \mathbf{Q}_p)
+
\sum_{q=1}^N Q_{p \leftarrow q}^I.
\end{equation}

When all the phases are incompressible, summing Eq.~(\ref{Eq:Incompressible-VolumeFraction}) over $p$, we obtain
\begin{equation}\label{Eq:Incompressible-DivergenceFree}
\nabla \cdot \mathbf{u}
=
0,
\end{equation}
which is the divergence-free condition for incompressible flows.
As $\sum_{q=1}^N \mathbf{J}_q=\mathbf{0}$, required by the second law of thermodynamics in Eq.~(\ref{Eq:SecondLaw-PhaseField}), $\sum_{q=1} \alpha_q = 1$ is implied in Eq.~(\ref{Eq:Incompressible-Mass}) together with Eq.~(\ref{Eq:Incompressible-DivergenceFree}).
Combining Eq.~(\ref{Eq:Incompressible-VolumeFraction}) and Eq.~(\ref{Eq:Incompressible-DivergenceFree}), we obtain $\sum_{q=1}^N \zeta_{p,q}^I (P_p - P_q)=0$, for $1 \leqslant p \leqslant N$, implying the isobaric closure. Using the single-temperature assumption and summing Eq.~(\ref{Eq:Incompressible-InternalEnergy}) over $p$, we obtain
\begin{equation}\label{Eq:Incompressible-HeatTransfer}
\frac{\partial \left(\sum_{p=1}^N(\alpha_p \rho_p C_p^V) T\right)}{\partial t}
+
\nabla \cdot \left( \sum_{p=1}^N \left((\alpha_p \mathbf{u} - \mathbf{J}_p) \rho_p C_p^V\right) T \right)
=
\boldsymbol{\tau} : \nabla \mathbf{u}
-
\nabla \cdot \mathbf{Q},
\end{equation}
where $\boldsymbol{\tau}=\sum_{q=1}^N \alpha_q \boldsymbol{\tau}_q$. Eq.~(\ref{Eq:Incompressible-HeatTransfer}) is the heat equation commonly used in incompressible multiphase flows.
The momentum equation in Eq.~(\ref{Eq:GoverningEquations}) remains the same in the incompressible limit.
We note that the proposed model in Eq.~(\ref{Eq:GoverningEquations}) in its incompressible limit recovers the consistent and conservative Phase-Field model for \textit{incompressible} multiphase flows \citep{Huangetal2020,Huangetal2020N,Huangetal2020Solid}, where the incompressible assumption was applied at the beginning.

\section{Numerical approach}\label{Sec:NumericalScheme}
In this section, we extend the consistent and conservative numerical approach for the five-equation models with the Phase-Field mechanism \citep{HuangJohnsen2022,HuangJohnsen2023,HuangJohnsen2024} to the current six-equation model in Eq.~(\ref{Eq:GoverningEquations}).
This approach provides a unified bound-preserving framework to incorporate different formulations of the Phase-Field mechanism into compressible multiphase flows with possible high-order implementation \citep{Whiteetal2025} and adaptive mesh refinement \citep{Huangetal2025}.
Based on the physical mechanism in Eq.~(\ref{Eq:GoverningEquations}), the approach consists of the hyperbolic step, Phase-Field step, and relaxation step, which are discussed separately.
Our emphasis is on the relaxation step, which is the new component in comparison to our previous approach for the five-equation models. Specifically, we theoretically prove that there exists a unique solution that is thermodynamically admissible for the $N$-phase pressure and pressure-temperature relaxations, with the widely-used equation of state by Le M{\'e}tayer et al. \citep{LeMetayeretal2005}.
For the hyperbolic and Phase-Field steps, moderate modifications to our previous approach are needed to adapt to the present six-equation model in Eq.~(\ref{Eq:GoverningEquations}).

\subsection{Relaxation step}\label{Sec:Relaxation}
In this section, we focus on the relaxation process in Eq.~(\ref{Eq:GoverningEquations}), which is formulated as
\begin{equation}\label{Eq:RelaxationStep}
\frac{\partial \mathbf{U}}{\partial t}
=
\mathbf{S}^{RX},
\end{equation}
where the Newton's law of cooling is used for the heat transfer, i.e., $Q_{p \leftarrow q}^I=h_{p,q}^I \delta_{p,q}^I(T_q -T_p)$.
In this section, we limit our analysis to the equation of state (EOS) by Le M{\'e}tayer et al. \citep{LeMetayeretal2005}, which reads
\begin{equation}\label{Eq:EOS}
\rho_p (e_p-e_p^\infty)
=
\frac{1}{\gamma_p-1} P_p + \frac{\gamma_p P_p^\infty}{\gamma_p-1}
=
\rho_p C_p^V T_p + P_p^\infty,\quad
1 \leqslant p \leqslant N,
\end{equation}
where $\gamma_p$, $P_p^\infty$, $C_p^V$, and $e_p^\infty$ are material-dependent parameters. With an appropriate choice of the parameters, Eq.~(\ref{Eq:EOS}) recovers the ideal gas or stiffened gas equation of state. We note that Eq.~(\ref{Eq:EOS}) requires $P_p \geqslant P_p^{\min}=-P_p^\infty$ to have $T_p \geqslant 0$ and real $c_p$.

Given that $\mathbf{U}$ is thermodynamically admissible, i.e., $(\alpha_p\rho_p) \geqslant 0$, $P_p \geqslant P_p^{\min}=-P_p^\infty$, and $0 \leqslant \alpha_p \leqslant 1$, for $1 \leqslant p \leqslant N$, and $\sum_{q=1}^N \alpha_q=1$, we analyze both the instantaneous pressure relaxation and instantaneous pressure-temperature relaxation, proving that there exists after the relaxation a unique solution that is thermodynamically admissible. In the section, the state after the relaxation is denoted as $\mathbf{U}^*$.
Without loss of generality, we order the phases in such a way that $P_1^{\min} \geqslant P_2^{\min} \geqslant ... \geqslant P_N^{\min}$ (or equivalently $P_1^{\infty} \leqslant P_2^{\infty} \leqslant ... \leqslant P_N^{\infty}$). 

\subsubsection{Pressure relaxation}\label{Sec:Relaxation-P}
The instantaneous pressure relaxation has $\zeta_{p,q}^I \rightarrow +\infty$ and $h_{p,q}^I=0$ for $1 \leqslant p,q \leqslant N$, resulting in $P_p^*=P^*$ for $1 \leqslant p \leqslant N$. After integrating Eq.~(\ref{Eq:RelaxationStep}) over time, we obtain
\begin{equation}\label{Eq:Relaxation-Pressure-Integration}
(\alpha_p \rho_p)^* = (\alpha_p \rho_p),\quad
(\rho \mathbf{u})^* = (\rho \mathbf{u}),\quad
(\alpha_p \rho_p E_p)^* - (\alpha_p \rho_p E_p)
=
-P^* (\alpha_p^* - \alpha_p),\quad
\sum_{q=1}^N \alpha_q^*=\sum_{q=1}^N \alpha_q=1.
\end{equation}
Therefore, $(\alpha_p \rho_p)^*=(\alpha_p \rho_p) \geqslant 0$ is admissible.
When integrating the phasic total energy equation (the third equation in Eq.~(\ref{Eq:Relaxation-Pressure-Integration})), $P_{p,q}^I$ is treated implicitly, which is critical to achieve the thermodynamically admissible volume fraction.
Combining Eq.~(\ref{Eq:Relaxation-Pressure-Integration}) with EOS in Eq.~(\ref{Eq:EOS}), we obtain the relaxed volume fraction in terms of the relaxed pressure, which reads
\begin{equation}\label{Eq:Relaxation-Pressure-VolumeFraction}
\alpha_p^*
=
\frac{(P_p + P_p^\infty) + (\gamma_p - 1) (P^* + P_p^\infty)}{\gamma_p (P^* + P_p^\infty)} \alpha_p,
\end{equation}
After enforcing $\sum_{q=1}^N \alpha_q^*=1$, $P^*$ is the zero of $\mathcal{F}^P(\mathcal{P})$ defined as
\begin{equation}\label{Eq:Relaxation-Pressure}
\mathcal{F}^P(\mathcal{P})
=
\sum_{q=1}^N
\frac{(\mathcal{P}-P_q)}{(\mathcal{P} + P_q^\infty)} \frac{\alpha_q}{\gamma_q}.
\end{equation}
As the relaxed pressure is shared by all the phases, $P^* > \max_q(P_q^{\min})$ is needed, which results in $\alpha_p^* \geqslant 0$ for $1 \leqslant p \leqslant N$ from Eq.~(\ref{Eq:Relaxation-Pressure-VolumeFraction}) (recalling that $P_p^{\min}=-P_p^{\infty}$). Since $\sum_{q=1}^N \alpha_q^*=1$ has been enforced, $\alpha_p^*$ is admissible for $1 \leqslant p \leqslant N$.
Furthermore, to be consistent with the second law of thermodynamics, $P^*$ should not exceed $\left[\min_q(P_q),\max_q(P_q)\right]$, as discussed in Section~\ref{Sec:SeconLaw}. Therefore, the final task is to prove that $\mathcal{F}^P(\mathcal{P})=0$ has a unique solution in $\left[\min_q(P_q),\max_q(P_q)\right] \cup \left(\max_q(P_q^{\min}),\max_q(P_q)\right]$.

We first determine the distribution of the zeros of $\mathcal{F}^{P}(\mathcal{P})$, with the following theorem.
\begin{theorem}\label{Theorem:Relaxation-Pressure-Uniqueness}
$\mathcal{F}^P(\mathcal{P})=0$ has $N$ real roots in $(-\infty,P_N^{\min})$, $(P_p^{\min},P_{p+1}^{\min})$ for $1 \leqslant p \leqslant (N-1)$, and $(P_1^{\min},+\infty)$, with precisely one root per interval.
\end{theorem}
\begin{proof}\label{Proof:Relaxation-Pressure-Uniqueness}
$\mathcal{F}^P(\mathcal{P})=0$ is equivalent to root-finding for a polynomial of degree $N$ with respect to $\mathcal{P}$. Therefore, $\mathcal{F}^P(\mathcal{P})=0$ has at most $N$ real roots.
In addition, $\mathcal{F}^P(\mathcal{P})$ has $N$ poles, which are $P_p^{\min}=-P_p^\infty$ for $1 \leqslant p \leqslant N$, resulting in $(N+1)$ continuous branches, which are $(-\infty,P_N^{\min})$, $(P_{p+1}^{\min},P_{p}^{\min})$ for $1 \leqslant p \leqslant N-1$, and $(P_1^{\min},+\infty)$.
It is straightforward to show that
\begin{equation}\label{Eq:Relaxation-Pressure-Uniqueness-Limit}
\lim_{\mathcal{P}\rightarrow \pm\infty} \mathcal{F}^P(\mathcal{P})=\sum_{q=1}^N\frac{\alpha_q}{\gamma_q} > 0,\quad
\lim_{\mathcal{P}\rightarrow \left(P_p^{\min}\right)^+} \mathcal{F}^P(\mathcal{P})=-\infty,\quad
\lim_{\mathcal{P}\rightarrow \left(P_p^{\min}\right)^-} \mathcal{F}^P(\mathcal{P})=+\infty,
\end{equation}
with the fact that $P_p \geqslant P_p^{\min} = -P_p^{\infty}$.
As a result, there must be roots in $(P_{p+1}^{\min},P_{p}^{\min})$ (for $1 \leqslant p \leqslant N-1$) and $(P_1^{\min},+\infty)$, with precisely one root per interval and in total $N$ real roots.
\end{proof}

Theorem~\ref{Theorem:Relaxation-Pressure-Uniqueness} proves that there is only one root (the largest root) of $\mathcal{F}^P(\mathcal{P})=0$ that is larger than $P_1^{\min}=\max_q(P_q^{\min})$. Then, we further prove that this largest root satisfies the second law of thermodynamics with the following theorem.
\begin{theorem}\label{Theorem:Relaxation-Pressure-Admissible}
The largest root of $\mathcal{F}^P(\mathcal{P})=0$ is in $\left[\min_q(P_q),\max_q(P_q)\right] \cup \left(\max_q(P_q^{\min}),\max_q(P_q)\right]$.
\end{theorem}
\begin{proof}\label{Proof:Relaxation-Pressure-Admissible}
From Theorem~\ref{Theorem:Relaxation-Pressure-Uniqueness}, the largest root of $\mathcal{F}^P(\mathcal{P})=0$ must be in $(P_1^{\min},+\infty)$, recalling that $P_1^{\min}=\max_q(P_q^{\min})$.
Since $\max_q(P_q) \geqslant P_p \geqslant P_p^{\min}=-P_p^{\infty}$, for $1 \leqslant p \leqslant N$, $\mathcal{F}^{P}\left(\max_q(P_{q})\right) \geqslant 0$ is true. Along with $\lim_{\mathcal{P}\rightarrow \left(P_1^{\min}\right)^+} \mathcal{F}^P(\mathcal{P})=-\infty$, the largest root of $\mathcal{F}^P(\mathcal{P})=0$ now must be in $\left(\max_q(P_q^{\min}),\max_q(P_q)\right]$.

When $\min_q(P_q) \leqslant P_1^{\min}$, the largest root of $\mathcal{F}^P(\mathcal{P})=0$ is still in $\left(\max_q(P_q^{\min}),\max_q(P_q)\right]$, so we only need to consider $\min_q(P_q) > P_1^{\min}$. In this case, we have $\min_q(P_q) - P_p \leqslant 0$ and $\min_q(P_q) + P_p^{\infty} > 0$, for $1 \leqslant p \leqslant N$, resulting in $\mathcal{F}^{P}\left(\min_q(P_q)\right) \leqslant 0$. Therefore, the largest root of $\mathcal{F}^P(\mathcal{P})=0$ must be in $\left[\min_q(P_q),\max_q(P_q)\right]$.

In conclusion, the largest root of $\mathcal{F}^P(\mathcal{P})=0$ is in $\left[\min_q(P_q),\max_q(P_q)\right] \cup \left(\max_q(P_q^{\min}),\max_q(P_q)\right]$.
\end{proof}
\begin{remark}\label{Remark:Relaxation-Pressure-Admissible}
When there are only two phases ($N=2$), $\mathcal{F}^{P}(\mathcal{P})=0$ is equivalent to the following quadratic equation with the analytical solution
\begin{equation}\label{Eq:Relaxation-Pressure-TwoPhase}
\begin{split}
a^P \mathcal{P}^2 + b^P \mathcal{P} + c^P = 0,\quad
\mathcal{P}^{\pm} = \frac{-b^P \pm \sqrt{(b^P)^2 - 4 a^P c^P}}{2a^P},\\
a^P=\frac{\alpha_1}{\gamma_1} + \frac{\alpha_2}{\gamma_2},\quad
b^P=\frac{\alpha_1 P_2^\infty - (\alpha_1 P_1)}{\gamma_1} + \frac{\alpha_2 P_1^\infty - (\alpha_2 P_2)}{\gamma_2},\quad
c^P= -\left(\frac{(\alpha_1 P_1)}{\gamma_1} P_2^\infty + \frac{(\alpha_2 P_2)}{\gamma_2} P_1^\infty\right).
\end{split}
\end{equation}
From Theorem~\ref{Theorem:Relaxation-Pressure-Admissible}, $P^*=\mathcal{P}^+$ is thermodynamically admissible. It is also straightforward to show that $P^- \leqslant \min(P_1,P_2)$, which is inadmissible.
\end{remark}

\subsubsection{Pressure-Temperature relaxation}\label{Sec:Relaxation-PT}
The instantaneous pressure-temperature relaxation has $\zeta_{p,q}^I \rightarrow +\infty$ and $h_{p,q}^I \rightarrow +\infty$ for $1 \leqslant p,q \leqslant N$, resulting in $P_p^*=P^*$ and $T_p^*=T^*$, respectively, for $1 \leqslant p \leqslant N$. After integrating Eq.~(\ref{Eq:RelaxationStep}) over time, we obtain
\begin{equation}\label{Eq:Relaxation-PressureTemperature-Integration}
(\alpha_p \rho_p)^* = (\alpha_p \rho_p),\quad
(\rho \mathbf{u})^* = (\rho \mathbf{u}),\quad
\sum_{q=1}^N (\alpha_q \rho_q E_q)^*=\sum_{q=1}^N (\alpha_q \rho_q E_q)=(\rho E),\quad
\sum_{q=1}^N \alpha_q^*=\sum_{q=1}^N \alpha_q=1.
\end{equation}
Therefore, $(\alpha_p \rho_p)^*=(\alpha_p \rho_p) \geqslant 0$ is admissible.
Combining Eq.~(\ref{Eq:Relaxation-PressureTemperature-Integration}) with EOS in Eq.~(\ref{Eq:EOS}), we obtain the relaxed temperature and volume fraction in terms of the relaxed pressure, which reads
\begin{equation}\label{Eq:Relaxation-PressureTemperature-TemperatureVolumeFraction}
T^* = \frac{1}{\sum_{q=1}^N \frac{C_q^V (\gamma_q-1) (\alpha_q \rho_q)}{(P^* + P_q^\infty)}},\quad
\alpha_p^*
=
\frac{C_p^V (\gamma_p-1) (\alpha_p \rho_p)}{(P^* + P_p^\infty)} T^*
=
\frac{\frac{C_p^V (\gamma_p-1) (\alpha_p \rho_p)}{(P^* + P_p^\infty)}}{\sum_{q=1}^N \frac{C_q^V (\gamma_q-1) (\alpha_q \rho_q)}{(P^* + P_q^\infty)}}.
\end{equation}
Finally, $P^*$ is determined by energy conservation (the third equation in Eq.~(\ref{Eq:Relaxation-PressureTemperature-Integration})), which reads
\begin{equation}\label{Eq:Relaxation-PressureTemperature-Energy}
\sum_{q=1}^N \frac{P^*+\gamma_q P_q^\infty}{\gamma_q-1} \alpha_q^*
=
\sum_{q=1}^N \frac{P_q+\gamma_q P_q^\infty}{\gamma_q-1} \alpha_q
=
(\rho E)-\frac{1}{2} \frac{(\rho \mathbf{u})\cdot (\rho \mathbf{u})}{\rho} - \sum_{q=1}^N (\alpha_q \rho_q) e_q^\infty.
\end{equation}
Combining Eq.~(\ref{Eq:Relaxation-PressureTemperature-TemperatureVolumeFraction}) and Eq.~(\ref{Eq:Relaxation-PressureTemperature-Energy}), $P^*$ is the zero of $\mathcal{F}^{PT}(\mathcal{P})$ defined as
\begin{equation}\label{Eq:Relaxation-PressureTemperature}
\begin{split}
\mathcal{F}^{PT}(\mathcal{P})
=
\frac{\mathcal{N}^{PT}(\mathcal{P})}{\mathcal{G}^{PT}(\mathcal{P})}
-
\sum_{q=1}^N \frac{(P_q+P_q^\infty) + (\gamma_q-1) P_q^\infty}{\gamma_q-1} \alpha_q,\\
\mathcal{N}^{PT}(\mathcal{P})
=
\sum_{q=1}^N ((\mathcal{P}+P_q^\infty) + (\gamma_q-1) P_q^\infty) C_q^V (\alpha_q \rho_q) \Pi_{k=1, k \neq q}^N (\mathcal{P} + P_k^\infty),\\
\mathcal{G}^{PT}(\mathcal{P})
=
\sum_{q=1}^N C_q^V (\gamma_q-1) (\alpha_q \rho_q) \Pi_{k=1,k \neq q}^N (\mathcal{P} + P_k^\infty).
\end{split}
\end{equation}
As the relaxed pressure is shared by all the phases, $P^* > \max_q(P_q^{\min})$ is needed, which results in $T^* > 0$ and $0 \leqslant \alpha_p^* \leqslant 1$ for $1 \leqslant p \leqslant N$ from Eq.~(\ref{Eq:Relaxation-PressureTemperature-TemperatureVolumeFraction}) (recalling that $P_p^{\min}=-P_p^{\infty}$). As $\sum_{q=1}^N \alpha_q^*=1$ has been enforced, $\alpha_p^*$ is admissible for $1 \leqslant p \leqslant N$. Therefore, the remaining task is to prove that $\mathcal{F}^{PT}(\mathcal{P})=0$ has a unique solution in $\left(\max_q(P_q^{\min}),+\infty\right)$.
We note that, for pressure-temperature relaxation, $P^*$ does not need to be in $\left[\min_q(P_q),\max_q(P_q)\right]$ to satisfy the second law, because there is extra entropy production due to heat transfer. A straightforward two-phase example is that $P_1=P_2=P^{(0)}$ but $T_1 \neq T_2$, and $P^*$ after the pressure-temperature relaxation in general does not remain the value of $P^{(0)}$.

To begin our analysis, the following lemmas prove the important properties of $\mathcal{G}^{PT}(\mathcal{P})$ and $\mathcal{N}^{PT}(\mathcal{P})$. 
\begin{lemma}\label{Lemma:Relaxation-PressureTemperature-Pole}
$\mathcal{G}^{PT}(\mathcal{P})=0$ has $(N-1)$ real roots in $[P_{p+1}^{\min},P_p^{\min}]$, for $1 \leqslant p \leqslant (N-1)$, with precisely one root per interval.
\end{lemma}
\begin{proof}\label{Proof:Relaxation-PressureTemperature-Pole}
$\mathcal{G}^{PT}(\mathcal{P})$ is a polynomial of degree $(N-1)$ with respect to $\mathcal{P}$, and thus, there must be at most $(N-1)$ real roots.
Furthermore, algebraic calculations show that
\begin{eqnarray}\label{Eq:Relaxation-PressureTemperature-Pole}
\mathcal{G}^{PT}(P_p^{\min}) \times \mathcal{G}^{PT}(P_{p+1}^{\min})
=
\underbrace{C_p^V C_{p+1}^V (\gamma_p-1) (\gamma_{p+1}-1) (\alpha_p \rho_p) (\alpha_{p+1} \rho_{p+1})}_{\geqslant 0}
\times
\Pi_{k=1}^{p-1} \underbrace{\left(\underbrace{(P_k^\infty-P_p^{\infty})}_{\leqslant 0}\underbrace{(P_k^\infty-P_{p+1}^{\infty})}_{\leqslant 0}\right)}_{\geqslant 0}\\
\nonumber
\times
\Pi_{k=p+1}^N \underbrace{(P_k^\infty-P_p^{\infty})}_{\geqslant 0}
\times
\Pi_{k=p+2}^N \underbrace{(P_k^\infty-P_{p+1}^{\infty})}_{\geqslant 0}
\times
\underbrace{(P_p^\infty-P_{p+1}^{\infty})}_{\leqslant 0} \leqslant 0.
\end{eqnarray}
Therefore, there must be roots in $[P_{p+1}^{\min},P_p^{\min}]$, for $1 \leqslant p \leqslant (N-1)$, with precisely one root per interval and in total $(N-1)$ real roots.
\end{proof}

We denote the roots of $\mathcal{G}^{PT}(\mathcal{P})=0$ as $P_{p,p+1}^c \in [P_{p+1}^{\min},P_p^{\min}]$, for $1 \leqslant p \leqslant (N-1)$, which are actually the poles of $\mathcal{F}^{PT}(\mathcal{P})$.
We can further determine the sign of $\mathcal{G}^{PT}(\mathcal{P})$ as it approaches $P_{p,p+1}^c$.
\begin{lemma}\label{Lemma:Relaxation-PressureTemperature-Denominator-Sign}
The signs of $\lim_{\mathcal{P} \rightarrow \left(P_{p,p+1}^c\right)^{\pm}} \mathcal{G}^{PT}(\mathcal{P})$ are $(-1)^{p-1}$ and $(-1)^{p}$, respectively, for $1 \leqslant p \leqslant (N-1)$.
\end{lemma}
\begin{proof}\label{Proof:Relaxation-PressureTemperature-Denominator-Sign}
Based on Lemma~\ref{Lemma:Relaxation-PressureTemperature-Pole}, $\mathcal{G}^{PT}(\mathcal{P})$ can be alternatively written as
\begin{equation}\label{Eq:Relaxation-PressureTemperature-Denominator-Alternative}
\mathcal{G}^{PT}(\mathcal{P})
=
\mathcal{G}_{0}^{PT} \Pi_{q=1}^{N-1} (\mathcal{P} - P_{q,q+1}^c),
\end{equation}
where prefactor $\mathcal{G}^{PT}_0$ must be positive from Eq.~(\ref{Eq:Relaxation-PressureTemperature}). As a result,
\begin{equation}\label{Eq:Relaxation-PressureTemperature-Denominator-Limit}
\lim_{\mathcal{P} \rightarrow \left(P_{p,p+1}^c\right)^{\pm}} \mathcal{G}^{PT}(\mathcal{P})
=
\lim_{\varepsilon \rightarrow 0^{\pm}} \mathcal{G}^{PT}(P_{p,p+1}^c + \varepsilon)
=
\lim_{\varepsilon \rightarrow 0^{\pm}} \mathcal{G}_{0}^{PT} \times \varepsilon \times \Pi_{q=1}^{p-1} \underbrace{(P_{p,p+1}^c + \varepsilon - P_{q,q+1}^c)}_{\leqslant 0} \times \Pi_{q=p+1}^{N-1} \underbrace{(P_{p,p+1}^c + \varepsilon - P_{q,q+1}^c)}_{\geqslant 0}.
\end{equation}
Therefore, the signs of $\lim_{\mathcal{P} \rightarrow \left(P_{p,p+1}^c\right)^{\pm}} \mathcal{G}^{PT}(\mathcal{P})$ are $(-1)^{p-1}$ and $(-1)^{p}$, respectively.
\end{proof}

Now, we need to determine the sign of $\mathcal{N}^{PT}(\mathcal{P})$ at $P_{p,p+1}^c$.
\begin{lemma}\label{Lemma:Relaxation-PressureTemperature-Numerator-Sign}
The sign of $\mathcal{N}^{PT}(P_{p,p+1}^c)$ is $(-1)^p$, for $1 \leqslant p \leqslant (N-1)$.
\end{lemma}
\begin{proof}\label{Proof:Relaxation-PressureTemperature-Numerator-Sign}
From Eq.~(\ref{Eq:Relaxation-PressureTemperature}), $\mathcal{G}^{PT}(P_{p,p+1}^c)=0$ results in
\begin{equation}\label{Eq:Relaxation-PressureTemperature-Denominator-Pole}
C_p^V (\gamma_p-1) (\alpha_p \rho_p) \Pi_{k=1,k \neq p}^N (P_{p,p+1}^c + P_k^\infty)
=
-
\sum_{q=1,q \neq p}^N C_q^V (\gamma_q-1) (\alpha_q \rho_q) \Pi_{k=1,k \neq q}^N (P_{p,p+1}^c + P_k^\infty).
\end{equation}
Then, $\mathcal{N}^{PT}(P_{p,p+1}^c)$ becomes
\begin{eqnarray}\label{Eq:Relaxation-PressureTemperature-Numerator-Pole}
\mathcal{N}^{PT}(P_{p,p+1}^c)
=
\left(\Pi_{k=1}^{p} \underbrace{(P_{p,p+1}^c + P_k^\infty)}_{\leqslant 0} \times \Pi_{k=p+1}^N \underbrace{(P_{p,p+1}^c + P_k^\infty)}_{\geqslant 0}\right)\underbrace{\left(\sum_{q=1}^N C_q^V (\alpha_q \rho_q)\right)}_{\geqslant 0}\\
\nonumber
+
\sum_{q=1}^{p-1} \underbrace{C_q^V (\gamma_q-1) (\alpha_q \rho_q)}_{\geqslant 0} \underbrace{(P_q^\infty-P_p^\infty)}_{\leqslant 0} \left(\Pi_{k=1, k \neq q}^p \underbrace{(P_{p,p+1}^c + P_k^\infty)}_{\leqslant 0} \times \Pi_{k=p+1}^N \underbrace{(P_{p,p+1}^c + P_k^\infty)}_{\geqslant 0}\right)\\
\nonumber
+
\sum_{q=p+1}^N \underbrace{C_q^V (\gamma_q-1) (\alpha_q \rho_q)}_{\geqslant 0} \underbrace{(P_q^\infty-P_p^\infty)}_{\geqslant 0} \left(\Pi_{k=1}^p \underbrace{(P_{p,p+1}^c + P_k^\infty)}_{\leqslant 0} \times \Pi_{k=p+1, k \neq q}^N \underbrace{(P_{p,p+1}^c + P_k^\infty)}_{\geqslant 0}\right).
\end{eqnarray}
Therefore, the sign of $\mathcal{N}^{PT}(P_{p,p+1}^c)$ is $(-1)^p$.
\end{proof}

With these lemmas, we can show the distribution of the roots of $\mathcal{F}^{PT}(\mathcal{P})=0$ from the following theorem.
\begin{theorem}\label{Theorem:Relaxation-PressureTemperature-Uniqueness}
$\mathcal{F}^{PT}(\mathcal{P})=0$ has $N$ real roots in $(-\infty,P_{N-1,N}^{c})$, $(P_{p+1,p+2}^{c},P_{p,p+1}^{c})$ for $1 \leqslant p \leqslant (N-2)$, and $(P_{1,2}^{c},+\infty)$, with precisely one root per interval.
\end{theorem}
\begin{proof}\label{Proof:Relaxation-PressureTemperature-Uniqueness}
$\mathcal{F}^{PT}(\mathcal{P})=0$ is equivalent to root-finding for a polynomial of degree $N$ with respect to $\mathcal{P}$. Therefore, $\mathcal{F}^{PT}(\mathcal{P})=0$ has at most $N$ real roots.
In addition, from Lemma~\ref{Lemma:Relaxation-PressureTemperature-Pole}, $\mathcal{F}^{PT}(\mathcal{P})$ has $(N-1)$ poles, which are $P_{p,p+1}^c \in [P_{p+1}^{\min},P_{p}^{\min}]$, for $1 \leqslant p \leqslant (N-1)$, resulting in $N$ continuous branches, which are $(-\infty,P_{N-1,N}^{c})$, $(P_{p+1,p+2}^{c},P_{p,p+1}^{c})$ for $1 \leqslant p \leqslant (N-2)$, and $(P_{1,2}^{c},+\infty)$.
It is straightforward to show that
\begin{equation}\label{Eq:Relaxation-PressureTemperature-Uniqueness-Limit}
\lim_{\mathcal{P} \rightarrow -\infty} \mathcal{F}^{PT}(\mathcal{P}) = -\infty,\quad
\lim_{\mathcal{P} \rightarrow \left(P_{p,p+1}^c\right)^{-}} \mathcal{F}^{PT}(\mathcal{P}) = +\infty,\quad
\lim_{\mathcal{P} \rightarrow \left(P_{p,p+1}^c\right)^{+}} \mathcal{F}^{PT}(\mathcal{P}) = -\infty,\quad
\lim_{\mathcal{P} \rightarrow +\infty} \mathcal{F}^{PT}(\mathcal{P}) = +\infty,
\end{equation}
with the help of Lemma~\ref{Lemma:Relaxation-PressureTemperature-Denominator-Sign} and Lemma~\ref{Lemma:Relaxation-PressureTemperature-Numerator-Sign}.
As a result, there must be roots in $(-\infty,P_{N-1,N}^{c})$, $(P_{p+1,p+2}^{c},P_{p,p+1}^{c})$ for $1 \leqslant p \leqslant (N-2)$, and $(P_{1,2}^{c},+\infty)$, with precisely one root per interval and in total $N$ real roots.
\end{proof}

Theorem~\ref{Theorem:Relaxation-PressureTemperature-Uniqueness} proves that there is only one root (the largest root) of $\mathcal{F}^{PT}(\mathcal{P})=0$ that is larger than $P_{1,2}^c \leqslant P_1^{\min}=\max_q(P_q^{\min})$. Then, we further prove that this largest root is thermodynamically admissible with the following theorem.
\begin{theorem}\label{Theorem:Relaxation-PressureTemperature-Admissible}
The largest root of $\mathcal{F}^{PT}(\mathcal{P})=0$ is in $\left(\max_q(P_q^{\min}),+\infty\right)$.
\end{theorem}
\begin{proof}\label{Proof:Relaxation-PressureTemperature-Admissible}
From Theorem~\ref{Theorem:Relaxation-PressureTemperature-Uniqueness}, the largest root of $\mathcal{F}^{PT}(\mathcal{P})=0$ must be in $(P_{1,2}^{c},+\infty)$, where $P_{1,2}^c \in [P_2^{\min},P_1^{\min}]$. Therefore, $P_1^{\min}=\max_q(P_q^{\min}) \in (P_{1,2}^{c},+\infty)$, and it is straightforward to show
\begin{equation}\label{Eq:Relaxation-PressureTemperature-Admissible-P}
\mathcal{F}^{PT}\left(\max_q(P_q^{\min})\right)
=
\mathcal{F}^{PT}(P_1^{\min})
=
-
\sum_{q=1}^N \left(
(P_q^\infty - P_1^\infty)
+
\frac{(P_q+P_q^\infty)}{\gamma_q-1}
\right) \alpha_q
\leqslant
0,
\end{equation}
with the fact that $P_p \geqslant P_p^{\min}=-P_p^\infty$ for $1 \leqslant p \leqslant N$ and $\sum_{q=1}^N \alpha_q=1$.
Therefore, the largest root of $\mathcal{F}^{PT}(\mathcal{P})=0$ is in $\left(\max_q(P_q^{\min}),+\infty\right)$.
\end{proof}
\begin{remark}\label{Remark:Relaxation-PressureTemperature-Admissible}
When there are only two phases ($N=2$), $\mathcal{F}^{PT}(\mathcal{P})=0$ is equivalent to the following quadratic equation with the analytical solution
\begin{equation}\label{Eq:Relaxation-PressureTemperature-TwoPhase}
\begin{split}
a^{PT} \mathcal{P}^2 + b^{PT} \mathcal{P} + c^{PT} = 0,\quad
\mathcal{P}^{\pm} = \frac{-b^{PT} \pm \sqrt{(b^{PT})^2 - 4 a^{PT} c^{PT}}}{2a^{PT}},\\
a^{PT}=C_1^V (\alpha_1 \rho_1) + C_2^V (\alpha_2 \rho_2),\\
b^{PT}=\left(
C_1^V (\alpha_1 \rho_1) (
P_2^\infty
+
\gamma_1 P_1^\infty
)
+
C_2^V (\alpha_2 \rho_2) (
P_1^\infty
+
\gamma_2 P_2^\infty
)
\right)\\
\nonumber
-
\left( (\rho E)-\frac{1}{2} \frac{(\rho \mathbf{u})\cdot (\rho \mathbf{u})}{\rho} - \sum_{q=1}^{2} (\alpha_q \rho_q) e_q^\infty \right) \left(C_1^V (\gamma_1-1) (\alpha_1 \rho_1) + C_2^V (\gamma_2-1) (\alpha_2 \rho_2) \right),\\
c^{PT}= \left(
C_1^V \gamma_1 (\alpha_1 \rho_1)
+
C_2^V \gamma_2 (\alpha_2 \rho_2)
\right) P_1^\infty P_2^\infty\\
\nonumber
-
\left(
(\rho E)-\frac{1}{2} \frac{(\rho \mathbf{u})\cdot (\rho \mathbf{u})}{\rho} - \sum_{q=1}^{2} (\alpha_q \rho_q) e_q^\infty \right) \left(C_1^V (\gamma_1-1) (\alpha_1 \rho_1) P_2^\infty  + C_2^V (\gamma_2-1) (\alpha_2 \rho_2) P_1^\infty
\right).
\end{split}
\end{equation}
From Theorem~\ref{Theorem:Relaxation-PressureTemperature-Admissible}, $P^*=\mathcal{P}^+$ is thermodynamically admissible.
\end{remark}

\subsection{Hyperbolic step}\label{Sec:Hyperbolic}
In the hyperbolic step, a finite-volume approach is applied to evolve
\begin{equation}\label{Eq:HyperbolicStep}
\frac{\partial \overline{\mathbf{U}}}{\partial t}
+
\nabla \cdot \hat{\mathbf{F}}^{HB}\left(\mathbf{U}^L,\mathbf{U}^R\right)
=
\hat{\mathbf{S}}^{HB}\left(\mathbf{U}^L,\mathbf{U}^R,\overline{\mathbf{U}}\right),
\end{equation}
where $\overline{\mathbf{U}}$ is the cell-averaged data, $\mathbf{U}^{L,R}$ are the reconstructed data at two sides of a grid cell face, and $\hat{\mathbf{F}}^{HB}$ and $\hat{\mathbf{S}}^{HB}$ are the numerical flux and non-conservative vectors, respectively, in the hyperbolic step. As a result, the conservation of mass, momentum, and energy is satisfied. The calculation of the non-conservative vector is tightly connected to that of the flux vector, which reads
\begin{equation}\label{Eq:Hyperbolic-NonConservative}
\begin{split}
\hat{\mathbf{S}}^{HB}=\left[\left\{0\right\}_{p=1}^N,\mathbf{0},\right.\\
\left.
\left\{\sum_{q=1}^N \left(
\frac{\overline{(\alpha_q \rho_q)}}{\overline{\rho}} \left(\nabla \cdot \widehat{(\alpha_p P_p \mathbf{u})} - \overline{(\alpha_p P_p)} \nabla \cdot \hat{\mathbf{u}} \right)
-
\frac{\overline{(\alpha_p \rho_p)}}{\overline{\rho}} \left( \nabla \cdot \widehat{(\alpha_q P_q \mathbf{u})} - \overline{(\alpha_q P_q)} \nabla \cdot \hat{\mathbf{u}} \right)
\right)\right\}_{p=1}^N,\left\{\overline{(\alpha_p)} \nabla \cdot \hat{\mathbf{u}}\right\}_{p=1}^N\right]^T,
\end{split}
\end{equation}
where $\hat{\mathbf{u}}=\hat{\mathbf{F}}^{HB,\alpha_p}|_{\alpha_p^{L,R}=1}$ and $\widehat{(\alpha_p P_p \mathbf{u})}=\hat{\mathbf{F}}^{HB,(\alpha_p\rho_pE_p)}|_{E_p^{L,R}=0}$ are directly obtained from the numerical flux vector, and $\hat{\mathbf{F}}^{HB,\xi}$ denotes the numerical flux for quantity $\xi$.
This choice of $\hat{\mathbf{u}}$ is for satisfying the consistency of reduction, which prevents the production of fictitious phases, local voids, or overfilling, and the analysis is identical to that for the five-equation models in \citep{HuangJohnsen2022,HuangJohnsen2023,HuangJohnsen2024,Whiteetal2025}.
The diffusive flux vector and diffusive non-conservative vector in Eq.~(\ref{Eq:Diffusion}) may be included here to account for the viscosity and heat condition, but these terms are skipped for inertia-dominant problems investigated in the present study.

The HLLC approximate Riemann solver \citep{Toroetal1994,Toro2009} with the wave speeds in \citep{Einfeldtetal1991} is used as the numerical hyperbolic flux vector, and we only present the $x$-component, which reads
\begin{equation}\label{Eq:HLLC}
\begin{split}
\hat{\mathbf{F}}_{x}^{HLLC} = \frac{1 + \mathrm{sign}(S^*)}{2} \left(\mathbf{F}_{x}^{HB}\left(\mathbf{U}^L\right) + S^{-} \left(\mathbf{U}^{L*} - \mathbf{U}^{L}\right)\right)
                        + \frac{1 - \mathrm{sign}(S^*)}{2} \left(\mathbf{F}_{x}^{HB}\left(\mathbf{U}^R\right) + S^{+} \left(\mathbf{U}^{R*} - \mathbf{U}^{R}\right)\right)\\
S^{L} = \min\left(u^L - c^L, u^R - c^R\right),\quad
S^{R} = \max\left(u^L + c^L, u^R + c^R\right),\quad
S^{-} = \min\left(0, S^L\right),\quad
S^{+} = \max\left(0, S^R\right),\\
S^{*} = \frac{\sum_{q=1}^N (\alpha_q P_q)^R - \sum_{q=1}^N (\alpha_q P_q)^L + \rho^R u^R (u^R - S^R) - \rho^L u^L (u^L - S^L)}{\rho^R (u^R - S^R) - \rho^L (u^L - S^L)},\quad
c = \sqrt{\frac{\sum_{q=1}^N \alpha_q \rho_q c_q^2}{\rho}}.
\end{split}
\end{equation}
We modify the star states $\mathbf{U}^{K*}$ ($K=L,R$) derived in \citep{Saureletal2009}, where a two-phase six-equation model based on the phasic internal energy was investigated, so that the star states adapt to the present $N$-phase model based on the phasic total energy, and they read
\begin{equation}\label{Eq:HLLC-StarState}
\begin{split}
\mathbf{U}^{K*}=
\begin{bmatrix}
\left\{(\alpha_p \rho_p)^{K*}\right\}_{p=1}^N\\
(\rho u)^{K*}\\
(\rho v)^{K*}\\
(\rho w)^{K*}\\
\left\{(\alpha_p\rho_p E_p)^{K*}\right\}_{p=1}^N\\
\left\{\alpha_p^{K*}\right\}_{p=1}^N
\end{bmatrix}
=
\frac{u^K-S^K}{S^{*}-S^K}
\begin{bmatrix}
\left\{(\alpha_p \rho_p)^K\right\}_{p=1}^N\\
(\rho u)^K\\
(\rho v)^K\\
(\rho w)^K\\
\left\{(\alpha_p \rho_p E_p)^K\right\}_{p=1}^N\\
\left\{\alpha_p^K\right\}_{p=1}^N
\end{bmatrix}
+
\begin{bmatrix}
\left\{0\right\}_{p=1}^N\\
\frac{u^K-S^K}{S^{*}-S^K} \rho^K (S^{*} - u^K)\\
0\\
0\\
\left\{\frac{u^K (\alpha_p P_p)^K - S^{*} (\alpha_p P_p)^{K*}}{S^{*} - S^K}\right\}_{p=1}^N\\
\frac{S^{*} - u^K}{S^{*}-S^K} \left\{\alpha_p^K\right\}_{p=1}^N
\end{bmatrix},\\
(\alpha_p P_p)^{K*} = (\alpha_p P_p)^K + (u^K - S^K) (\alpha_p \rho_p)^K u^K - (S^* - S^K) (\alpha_p \rho_p)^{K*} S^{*}.
\end{split}
\end{equation}
As a result, the corresponding $\hat{\mathbf{u}}^{HLLC}$ and $\widehat{(\alpha_p P_p \mathbf{u})}^{HLLC}$ are
\begin{equation}\label{Eq:HLLC-Nonconservative}
\begin{split}
\hat{\mathbf{u}}^{HLLC} = \frac{1 + \mathrm{sign}(S^{*})}{2} \mathbf{u}^L + \frac{1 - \mathrm{sign}(S^*)}{2} \mathbf{u}^R,\\
\widehat{(\alpha_p P_p u)}^{HLLC} = \frac{1 + \mathrm{sign}(S^*)}{2} \left( (\alpha_p P_p)^L u^L + \frac{u^L (\alpha_p P_p)^L - S^* (\alpha_p P_p)^{L*}}{S^* - S^L} S^{-} \right)\\
                                  + \frac{1 - \mathrm{sign}(S^*)}{2} \left( (\alpha_p P_p)^R u^R + \frac{u^R (\alpha_p P_p)^R - S^* (\alpha_p P_p)^{R*}}{S^* - S^R} S^{+} \right),\\
\widehat{(\alpha_p P_p v)}^{HLLC} = \frac{1 + \mathrm{sign}(S^*)}{2} (\alpha_p P_p)^L v^L
                                  + \frac{1 - \mathrm{sign}(S^*)}{2} (\alpha_p P_p)^R v^R,\\
\widehat{(\alpha_p P_p w)}^{HLLC} = \frac{1 + \mathrm{sign}(S^*)}{2} (\alpha_p P_p)^L w^L
                                  + \frac{1 - \mathrm{sign}(S^*)}{2} (\alpha_p P_p)^R w^R.
\end{split}
\end{equation}

To obtain $\mathbf{U}^{L,R}$, we follow the strategy that the primitive variables are reconstructed \citep{Abgrall1996,SaurelAbgrall1999-Gamma,JohnsenColonius2006,CoralicColonius2014,BeigJohnsen2015}, which avoids velocity errors across isolated interfaces.
The primitive variables for the present six-equation model are $\mathbf{V}=\left[\{(\alpha_p\rho_p\}_{p=1}^N,\mathbf{u},\{(\alpha_p P_p)\}_{p=1}^N,\{\alpha_p\}_{p=1}^N\right]^T$.
In our previous studies for the five-equation models \citep{HuangJohnsen2022,HuangJohnsen2023,HuangJohnsen2024}, we proposed the \textit{consistent reconstruction} of mass and volume fraction, which requires that $(\alpha_p \rho_p)^{L,R}=\rho_p \times \alpha_p^{L,R}$ for constant $\rho_p$. This consistency requirement guarantees that the interface labeled by $(\alpha_p\rho_p)$ is advected consistently with that by $\alpha_p$. Otherwise, mass spikes are produced near interfaces, as demonstrated in \citep{HuangJohnsen2022}, which may cause simulation failure in large-density-ratio problems. We further prove that the consistency requirement must be satisfied to achieve pressure and temperature equilibria at isolated interfaces for general equations of state \citep{HuangJohnsen2024}.
This consistency requirement is now extended to adapt to the present six-equation model such that $(\alpha_p P_p)^{L,R}=P_p \times \alpha_p^{L,R}$ for constant $P_p$, so that the pressure and temperature equilibria at isolated interfaces is preserved.

With the above HLLC flux and the consistency requirement for reconstruction, it is straightforward to show that the velocity, pressure, and temperature equilibria at isolated interfaces are satisfied, following a procedure similar to \citep{Abgrall1996,SaurelAbgrall1999-Gamma,JohnsenColonius2006,Saureletal2009,CoralicColonius2014,BeigJohnsen2015,HuangJohnsen2023,HuangJohnsen2024}. Therefore, we only highlight the key steps. Given $\rho_p=\rho_p^{(0)}$, $\mathbf{u}=\mathbf{u}^{(0)}$, $P_p=P^{(0)}$, $T_p=T_p^{(0)}=\mathsf{T}_p\left(\rho_p^{(0)},P^{(0)}\right)$ at the beginning, we have $(\alpha_p\rho_p)^{L,R}=\rho_p^{(0)}\times \alpha_p^{L,R}$, $\mathbf{u}^{L,R}=\mathbf{u}^{(0)}$, $(\alpha_p P_p)^{L,R}=P^{(0)}\times \alpha_p^{L,R}$, due to the consistent reconstruction on the primitive variables. Then, we obtain $S^*=u^{(0)}$ from Eq.~(\ref{Eq:HLLC}) with the fact that $\sum_{q=1}^N \alpha_q^{L,R}=1$, and $\mathbf{U}^{K*}=\mathbf{U}^{K}$ ($K=L,R$) from Eq.~(\ref{Eq:HLLC-StarState}). Without loss of generality, we assume $u^{(0)} \geqslant 0$, resulting in $\hat{\mathbf{F}}_x^{HLLC}=\mathbf{F}_x^{HB}(\mathbf{U}^L)$ from Eq.~(\ref{Eq:HLLC}), and $\hat{\mathbf{u}}^{HLLC}=\mathbf{u}^{(0)}$ and $\widehat{(\alpha_p P_p \mathbf{u})}^{HLLC}=\alpha_p^{L} P^{(0)}\mathbf{u}^{(0)}$ from Eq.~(\ref{Eq:HLLC-Nonconservative}). With these simplifications, we obtain $\rho_p=\rho_p^{(0)}$ after comparing the mass equation with the volume fraction equation in Eq.~(\ref{Eq:HyperbolicStep}). Then, the momentum and energy equations in Eq.~(\ref{Eq:HyperbolicStep}) result in $\mathbf{u}=\mathbf{u}^{(0)}$ and $e_p=e_p^{(0)}=\mathsf{e}_p(\rho_p^{(0)},P^{(0)})$, respectively. Finally, from the state principle in thermodynamics, $P_p=P^{(0)}$ and $T_p=T_p^{(0)}$ are obtained due to $\rho_p=\rho_p^{(0)}$ and $e_p=e_p^{(0)}$, which is valid for any equation of state.

\subsection{Phase-Field step}\label{Sec:PhaseField}
In the Phase-Field step, we continuously use the \textit{multiphase reduction-consistent formulation} \citep{HuangJohnsen2022,HuangJohnsen2023,HuangJohnsen2024,Whiteetal2025,Huangetal2025}, which robustly couples different formulations of the Phase-Field mechanism to compressible multiphase flows.
This approach first introduces the auxiliary variables $\{Q_p\}_{p=1}^N$ from
\begin{equation}\label{Eq:Q}
\nabla \cdot \left( \phi_p (1-\phi_p)  \nabla Q_p \right)=\left(\nabla \cdot \mathbf{J}_p\right),
\end{equation}
where the right-hand side is any (discretized) user-selected formulation of the Phase-Field mechanism, and $\phi_p$ is a spread of $\alpha_p$ to avoid rank deficiency during the inversion of the variable-coefficient Laplacian on the left-hand side.
Following \citep{Huangetal2025}, we use $\phi_p=\mathcal{M}_{2\eta}\left(\mathcal{M}_{\eta}^{-1}(\alpha_p)\right)$ in the present study, where $\mathcal{M}_{\eta}(x)=\frac{1}{2}\left(1+\tanh{\left(\frac{x}{\sqrt{2}\eta}\right)}\right)$, $\mathcal{M}_{\eta}^{-1}(x)$ is the inverse function, and $\eta$ controls the interface thickness.
Eq.~(\ref{Eq:Q}) serves as a bridge to connect the selected formulation of the Phase-Field mechanism to the unified calculation of the Phase-Field flux vector that is then applied to evolve
\begin{equation}\label{Eq:PhaseFieldStep}
\frac{\partial \overline{\mathbf{U}}}{\partial t}
=
\nabla \cdot \hat{\mathbf{F}}^{PF}\left(\mathbf{U}^L,\mathbf{U}^R;\{\nabla Q_p\}_{p=1}^N\right)
+
\hat{\mathbf{S}}^{PF}\left(\mathbf{U}^L,\mathbf{U}^R,\overline{\mathbf{U}}\right),
\end{equation}
where $\hat{\mathbf{F}}^{PF}$ and $\hat{\mathbf{S}}^{PF}$ are the numerical flux and non-conservative vectors, respectively, in the Phase-Field step. As a result, the conservation of mass, momentum, and energy is preserved.
Like in the hyperbolic step, $\hat{\mathbf{S}}^{PF}$ has a tight connection to $\hat{\mathbf{F}}^{PF}$, which reads
\begin{equation}\label{Eq:PhaseField-NonConservative}
\begin{split}
\hat{\mathbf{S}}^{PF}=\left[\left\{0\right\}_{p=1}^N,\mathbf{0},\right.\\
\left.
\left\{\sum_{q=1}^N \left(
\frac{\overline{(\alpha_p \rho_p)}}{\overline{\rho}} \left( \nabla \cdot \widehat{\left(\mathbf{J}_q \rho_q \frac{\mathbf{u} \cdot \mathbf{u}}{2}\right)} - \frac{\overline{\mathbf{u}} \cdot \overline{\mathbf{u}}}{2} \nabla \cdot \widehat{(\mathbf{J}_q \rho_q)} \right)
-
\frac{\overline{(\alpha_q \rho_q)}}{\overline{\rho}} \left( \nabla \cdot \widehat{\left(\mathbf{J}_p \rho_p \frac{\mathbf{u} \cdot \mathbf{u}}{2}\right)} - \frac{\overline{\mathbf{u}} \cdot \overline{\mathbf{u}}}{2} \nabla \cdot \widehat{(\mathbf{J}_p \rho_p)} \right)
\right)\right\}_{p=1}^N,\left\{0\right\}_{p=1}^N\right]^T,
\end{split}
\end{equation}
where $\widehat{(\mathbf{J}_p \rho_p)}=\hat{\mathbf{F}}^{PF,(\alpha_p\rho_p)}$ and $\widehat{\left(\mathbf{J}_p \rho_p \frac{\mathbf{u} \cdot \mathbf{u}}{2}\right)}=\hat{\mathbf{F}}^{PF,(\alpha_p\rho_pE_p)}|_{e_p^{L,R}=0}$.

Following the approach in \citep{HuangJohnsen2024}, $\hat{\mathbf{F}}^{PF}$ for the six-equation model is obtained from
\begin{equation}\label{Eq:PhaseField-Flux}
\begin{split}
\hat{\mathbf{F}}^{PF}\left(\mathbf{U}^L,\mathbf{U}^R;\{\nabla Q_p\}_{p=1}^N\right)
=
\psi_{\tilde{\mathbf{U}} \rightarrow \mathbf{U}}\left(
\hat{\tilde{\mathbf{F}}}^{PF}\left(\psi_{\mathbf{U} \rightarrow \tilde{\mathbf{U}}}(\mathbf{U}^L),\psi_{\mathbf{U} \rightarrow \tilde{\mathbf{U}}}(\mathbf{U}^R);\{\nabla Q_p\}_{p=1}^N\right)
\right),\\
\psi_{\mathbf{U} \rightarrow \tilde{\mathbf{U}}}(\mathbf{U})
=
\left[
\left\{(\alpha_p\rho_p)\right\}_{p=1}^N,
\left\{(\alpha_p\rho_p)\frac{(\rho\mathbf{u})}{\sum_{q=1}^N (\alpha_q \rho_q)}\right\}_{p=1}^N,
\left\{(\alpha_p\rho_pE_p)\right\}_{p=1}^N,\left\{\alpha_p\right\}_{p=1}^N
\right]^T,\\
\psi_{\tilde{\mathbf{U}} \rightarrow \mathbf{U}}\left(\tilde{\mathbf{U}}\right)
=
\left[
\left\{(\alpha_p\rho_p)\right\}_{p=1}^N,
\left\{\sum_{q=1}^N (\alpha_q\rho_q\mathbf{u}_q)\right\}_{p=1}^N,
\left\{(\alpha_p\rho_pE_p)\right\}_{p=1}^N,
\left\{\alpha_p\right\}_{p=1}^N\right]^T,
\end{split}
\end{equation}
where $\hat{\tilde{\mathbf{F}}}^{PF}\left(\tilde{\mathbf{U}}^L,\tilde{\mathbf{U}}^R;\{\nabla Q_p\}_{p=1}^N\right)$ is a numerical approximation to 
\[\tilde{\mathbf{F}}^{PF}=\left[
\left\{(\mathbf{J}_p\rho_p)\right\}_{p=1}^N,
\{(\mathbf{J}_p\rho_p) \otimes \mathbf{u}_p\}_{p=1}^N,
\left\{(\mathbf{J}_p \rho_p E_p)\right\}_{p=1}^N,
\left\{\mathbf{J}_p\right\}_{p=1}^N\right]^T,\] for the full sate vector
\[
\tilde{\mathbf{U}}
=
\left[
\left\{(\alpha_p\rho_p)\right\}_{p=1}^N,
\left\{(\alpha_p\rho_p\mathbf{u}_p)\right\}_{p=1}^N,
\left\{(\alpha_p\rho_p E_p)\right\}_{p=1}^N,
\left\{\alpha_p\right\}_{p=1}^N
\right]^T,\] and $\psi_{\mathbf{U} \rightarrow \tilde{\mathbf{U}}}(\mathbf{U}):\mathbb{R}^{3N+d} \to \mathbb{R}^{(3+d)N}$ and $\psi_{\tilde{\mathbf{U}} \rightarrow \mathbf{U}}(\tilde{\mathbf{U}}):\mathbb{R}^{(3+d)N} \to \mathbb{R}^{3N+d}$ are the mappings between $\mathbf{U}$ and $\tilde{\mathbf{U}}$ ($d=1,2,3$ is the problem dimension).
There are multiple options of $\hat{\tilde{\mathbf{F}}}^{PF}$ developed and analyzed in \citep{HuangJohnsen2024}, which are capable of preserving the volume fraction boundedness and mass positivity, and we use the upwind-downwind (UD) flux therein for its ease of implementation.

To preserve the velocity, pressure, and temperature equilibria at isolated interfaces in the Phase-Field step, the same consistency requirement for reconstruction discussed in the hyperbolic step is needed. As shown in \citep{HuangJohnsen2024}, given $\rho_p=\rho_p^{(0)}$, $\mathbf{u}=\mathbf{u}^{(0)}$, $P_p=P^{(0)}$, $T_p=T_p^{(0)}=\mathsf{T}_p\left(\rho_p^{(0)},P^{(0)}\right)$, a key property of $\hat{\tilde{\mathbf{F}}}^{PF}$ therein resulting from the consistency requirement for reconstruction is $\widehat{(\xi_p \mathbf{J}_p)}=\xi_p^{(0)} \times \hat{J}_p$ ($\xi=\rho$, $(\rho \mathbf{u})$, and $(\rho E)$), where $\hat{J}_p=\hat{\tilde{\mathbf{F}}}^{PF,\alpha_p}$. As a result, the non-conservative terms in Eq.~(\ref{Eq:PhaseField-NonConservative}) vanish. Then, combining Eq.~(\ref{Eq:PhaseFieldStep}) and Eq.~(\ref{Eq:PhaseField-Flux}), and using the state principle in thermodynamics, it is straightforward to show that $\rho_p=\rho_p^{(0)}$, $\mathbf{u}=\mathbf{u}^{(0)}$, $P_p=P^{(0)}$, and $T_p=T_p^{(0)}$ are true for any equation of state in the Phase-Field step, following a procedure similar to \citep{HuangJohnsen2022,HuangJohnsen2023,HuangJohnsen2024}.

\subsection{Implementation}\label{Sec:Implementation}
We limit the present implementation of the proposed consistent and conservative numerical approach to compressible two-phase flows ($N=2$), so that both the pressure and pressure-temperature relaxations have analytical solutions that are thermodynamically admissible (Remarks~\ref{Remark:Relaxation-Pressure-Admissible}~and~\ref{Remark:Relaxation-PressureTemperature-Admissible}). As a result, only $\alpha_1$ is stored and $\alpha_2$ is obtained from $(1-\alpha_1)$ whenever needed. The MUSCL reconstruction \citep{LeVeque2002} is used in both the hyperbolic and Phase-Field steps, which satisfies the consistency requirement for reconstruction without extra effort. The third-order TVD Runge-Kutta time integration \citep{Shu1988,GottliebShu1998,Gottliebetal2001} is used to integrate Eq.~(\ref{Eq:HyperbolicStep}) and Eq.~(\ref{Eq:PhaseFieldStep}).

In the present study, we use the conservative Phase-Field model \citep{ChiuLin2011,Mirjalilietal2020}
\begin{equation}\label{Eq:CDI}
\mathbf{J}_p = M \left(\nabla \alpha_p - \frac{\sqrt{2}}{\eta} \alpha_p(1-\alpha_p) \mathbf{n}_p\right),
\end{equation}
as the Phase-Field mechanism, where $M=0.5\eta \max|\mathbf{u}|$ is the mobility and $n_p=\nabla \alpha_p/|\nabla \alpha_p|$ is the normal vector at the interface. It is straightforward to verify that Eq.~(\ref{Eq:CDI}) satisfies $\sum_{q=1}^2 \mathbf{J}_q=\mathbf{0}$, as required by the second law of thermodynamics in Eq.~(\ref{Eq:SecondLaw-PhaseField}). We discretize Eq.~(\ref{Eq:CDI}) with second-order central difference. Moreover, as it is a flux-based formulation, we can explicitly express $\nabla Q_p=\mathbf{J}_p/\left(\phi_p(1-\phi_p)\right)$ from Eq.~(\ref{Eq:Q}) without the need of inverting the variable-coefficient Laplacian \citep{Huangetal2024}.

Overall, starting at time $t$, the solution at $t+\Delta t$ is obtained from:
\begin{enumerate}
\item Perform the hyperbolic step in Section~\ref{Sec:Hyperbolic}, which evolves Eq.~(\ref{Eq:HyperbolicStep}) from $t$ to $t+\Delta t$. The relaxation step in Section~\ref{Sec:Relaxation} is performed after every RK time stage.
    \item Perform the Phase-Field step in Section~\ref{Sec:PhaseField}, which first obtains $\mathbf{J}_p$ and $\nabla Q_p$ from Eq.~(\ref{Eq:CDI}) and Eq.~(\ref{Eq:Q}), respectively, and then evolves Eq.~(\ref{Eq:PhaseFieldStep}) from $\tau=0$ to $\tau=\Delta t$. The relaxation step in Section~\ref{Sec:Relaxation} is performed after every RK time stage.
\end{enumerate}
Here, $\Delta t$ is determined from the CFL condition $\Delta t = \mathrm{CFL} \frac{\Delta x}{\max(|\mathbf{u}|+c)}$, while $\Delta \tau$ in the Phase-Field step is determined by $\Delta \tau = \min\left(\Delta t,\mathrm{CFL}^{PF} \frac{\Delta x}{\max|\nabla Q_p|}\right)$ \citep{HuangJohnsen2022,HuangJohnsen2023,HuangJohnsen2024}. Our preliminary tests indicate that it is more accurate to perform the relaxation step after each RK time stage instead of after an entire RK time step.

\section{Results}\label{Sec:Results}
In this section, we demonstrate the proposed model and numerical approach with various benchmark problems of compressible two-phase flows. The interface thickness is evaluated by
\begin{equation}\label{Eq:InterfaceThickness}
N_I
=
\frac{\int_\Omega \mathcal{H}(\alpha_1) d\Omega}{\Delta x \int_\Omega |\nabla \alpha_1| d\Omega},
\quad
\mathcal{H}(\alpha_1)=\left\{
\begin{array}{cc}
     1,&  (\alpha_1)_{\min} \leqslant \alpha_1 \leqslant (\alpha_1)_{\max},\\
     0,& \mathrm{else},
\end{array}
\right.
\end{equation}
which measures the average number of grid cells across an interface. Following \citep{HuangJohnsen2022,HuangJohnsen2023,HuangJohnsen2024,Whiteetal2025,Huangetal2025}, we set $(\alpha_1)_{\min}=0.05$ and $(\alpha_1)_{\max}=0.95$, which cover 90\% of the diffuse interface, and, unless otherwise specified, use $\eta =\Delta x$ and $\mathrm{CFL}=\mathrm{CFL}^{PF}=0.4$.

\subsection{Advection of an air square in water}\label{Sec:Advection}
We first verify the kinematic, mechanical, and thermal equilibria at isolated interfaces and the conservation of mass, momentum, and energy via a two-phase advection problem.
An air square (Phase~$1$: $\rho_1=1.204\times10^{-3}$, $\gamma_1=1.4$, $P_1^\infty=0$, $C_1^V=0.1024$, and $e_1^\infty=0$) initially at $(x_s,y_s)=(0,0)$ with a side length of $L_s=0.4$ is surrounded by water (Phase~$2$: $\rho_2=1$, $\gamma_2=6.12$, $P_2^\infty=0.1631$, $C_2^V=0.5973$, and $e_2^{\infty}=0$), and both the air and water share the same velocity $\mathbf{u}^{(0)}=(1,1)$ and pressure $P^{(0)}=4.819\times10^{-5}$. The domain is $[-0.5,0.5]\times[-0.5,0.5]$, discretized by $128\times128$ grid cells, and the periodic boundary conditions are applied.

Fig.~\ref{Fig:Advection} shows the air volume fraction at $t=1$ and the time history of the interface thickness, with the pressure relaxation. Both results without and with the Phase-Field mechanism are shown. It is clear that the interface thickness grows unboundedly over time when the Phase-Field mechanism is not activated. On the other hand, the interface thickness is saturated at about $5$ grid cells after activating the Phase-Field mechanism, which is similar to \citep{HuangJohnsen2022,HuangJohnsen2023,HuangJohnsen2024,Huangetal2025}.
\begin{figure}[!t]
	\centering
    \includegraphics[scale=.4]{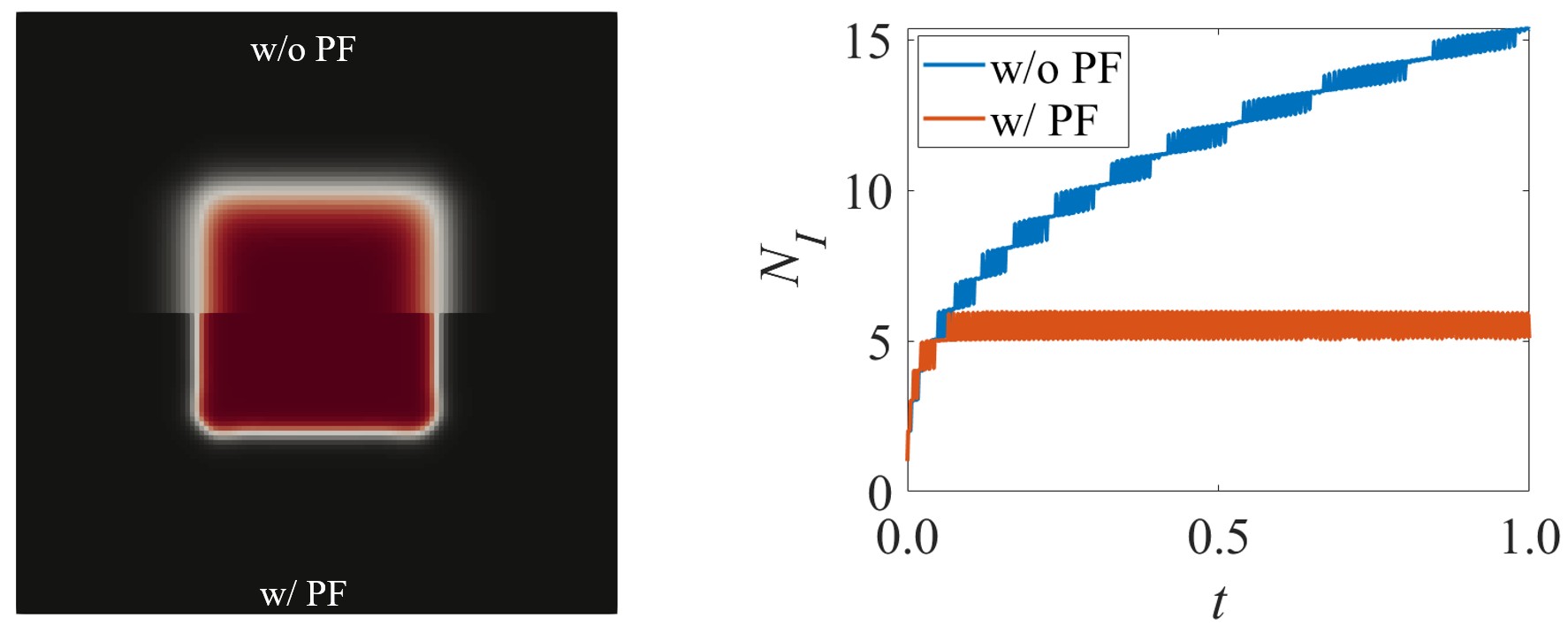}
	\caption{Air volume fraction at $t=1$ (left) and time history of the interface thickness (right) of the advection of an air square in water with the pressure relaxation.}\label{Fig:Advection}
\end{figure}

Fig.~\ref{Fig:Advection-Equilibrium-P} shows the $L^\infty$ errors of $x$-velocity, pressure, and Phase~$2$ temperature versus time with or without the pressure relaxation and with or without the Phase-Field mechanism. The $L^\infty$ errors of $y$-velocity and Phase~$1$ temperature behave similarly to their correspondences in Fig.~\ref{Fig:Advection-Equilibrium-P}, and therefore are not shown for a clear presentation. All the errors are at the round-off level, satisfying the kinematic and mechanical equilibria at isolated interfaces even without the pressure relaxation. However, the thermal equilibrium is not satisfied under the present setup, but individual phases maintain their own initial temperature. As a result, our analysis in Section~\ref{Sec:Hyperbolic} and Section~\ref{Sec:PhaseField} is verified.
\begin{figure}[!t]
	\centering
    \includegraphics[scale=.38]{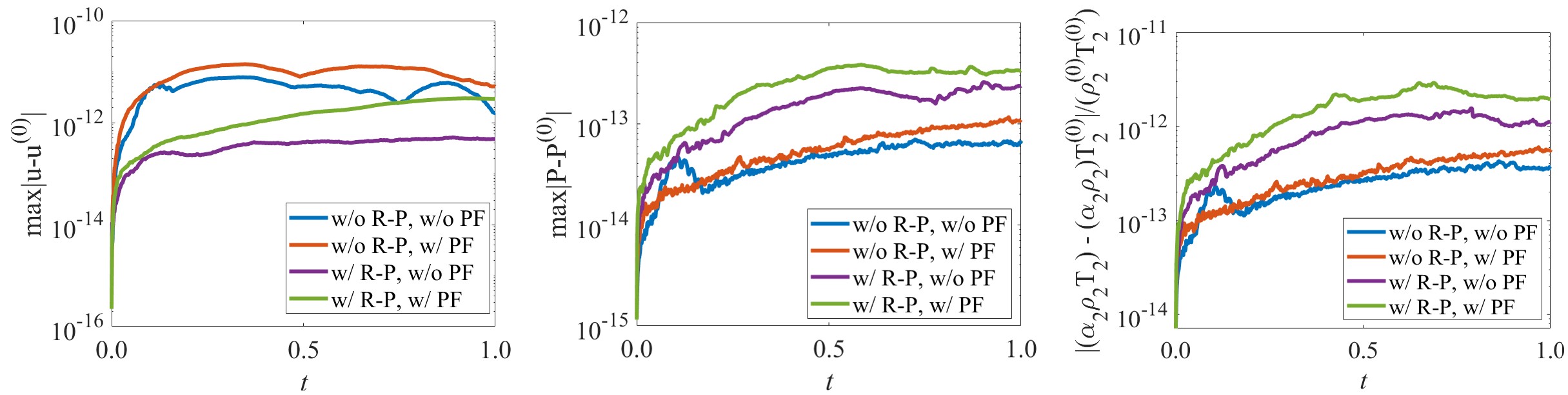}
	\caption{Time histories of the $L^\infty$ errors of $x$-velocity (left), pressure (middle), and Phase~$2$ temperature (right) of the advection of an air square in water with or without the pressure relaxation.}\label{Fig:Advection-Equilibrium-P}
\end{figure}

Fig.~\ref{Fig:Advection-Conservation-P} shows the changes of Phase~$2$ mass, $y$-momentum, and total energy versus time with or without the pressure relaxation and with or without the Phase-Field mechanism. The changes of Phase~$1$ mass and $x$-momentum behave similarly to their correspondences in Fig.~\ref{Fig:Advection-Conservation-P}, and therefore are not shown for a clear presentation. All the changes are at the round-off level, satisfying the conservation of mass, momentum, and energy even with the pressure relaxation. Moreover, the phasic total energies and volume fraction are also conserved under the present setup.
\begin{figure}[!t]
	\centering
    \includegraphics[scale=.38]{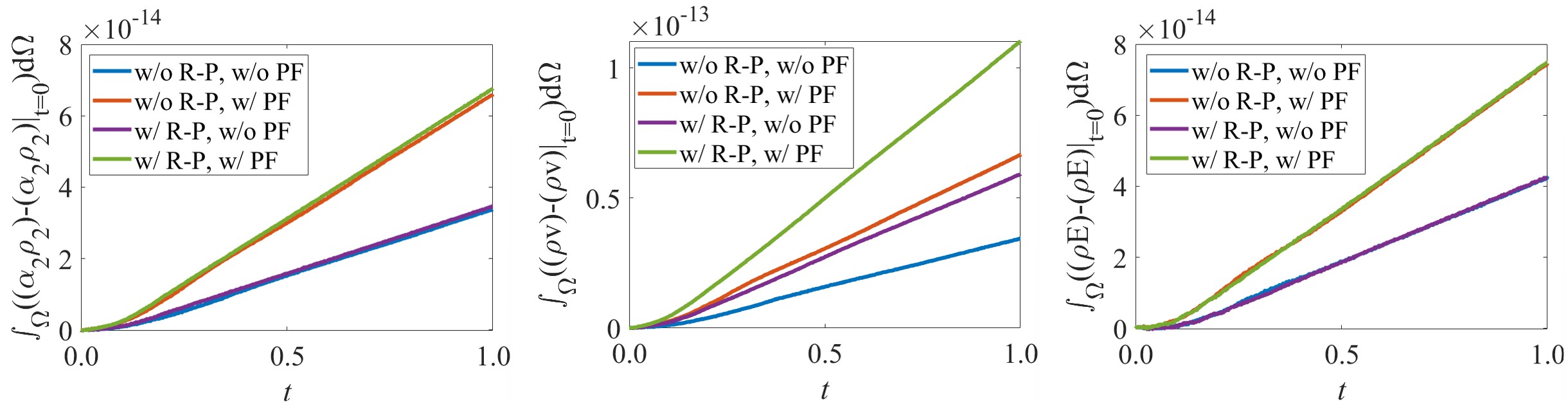}
	\caption{Time histories of the changes of Phase~$2$ mass (left), $y$-momentum (middle), and total energy (right) of the advection of an air square in water with or without the pressure relaxation.}\label{Fig:Advection-Conservation-P}
\end{figure}

Fig.~\ref{Fig:Advection-Equilibrium-PT} shows the $L^\infty$ errors of $x$-velocity, pressure, and temperature versus time with or without the pressure-temperature relaxation and with or without the Phase-Field mechanism, and all the errors are again at the round-off level even without the pressure-temperature relaxation. To achieve the thermal (temperature) equilibrium, we assign $C_p^V = (P^{(0)} + P_p^\infty)/((\gamma_p - 1)\rho_p T^{(0)})$ for $p=1,2$, where $T^{(0)}=1$.
\begin{figure}[!t]
	\centering
    \includegraphics[scale=.38]{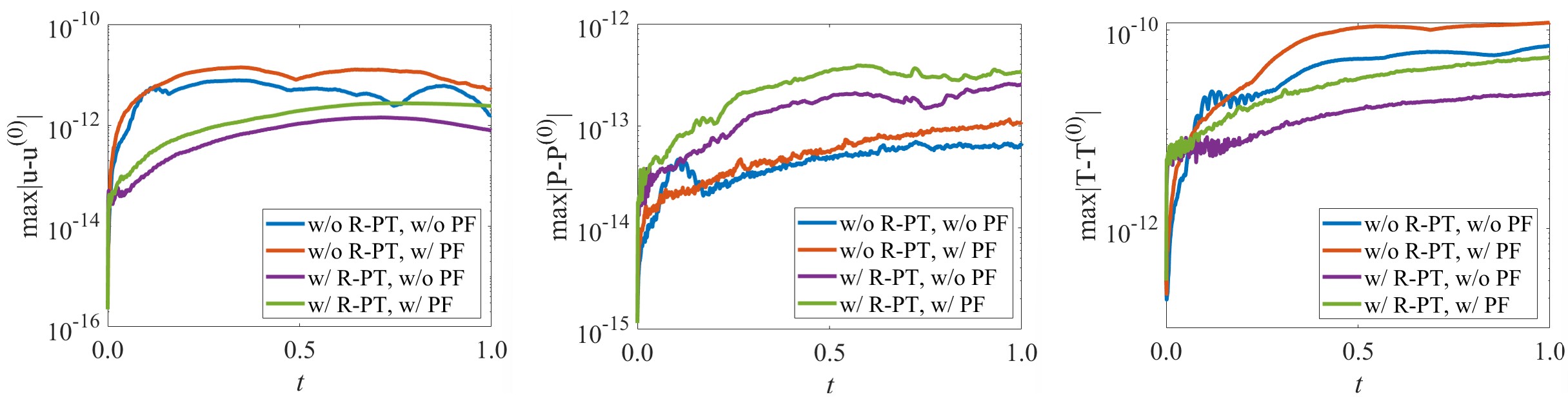}
	\caption{Time histories of the $L^\infty$ errors of $x$-velocity (left), pressure (middle), and temperature (right) of the advection of an air square in water with or without the pressure-temperature relaxation.}\label{Fig:Advection-Equilibrium-PT}
\end{figure}
Fig.~\ref{Fig:Advection-Conservation-PT} shows the changes of Phase~$2$ mass, $y$-momentum, and total energy versus time, which are at the round-off level, verifying the conservation of mass, momentum, and energy even with the pressure-temperature relaxation.
\begin{figure}[!t]
	\centering
    \includegraphics[scale=.38]{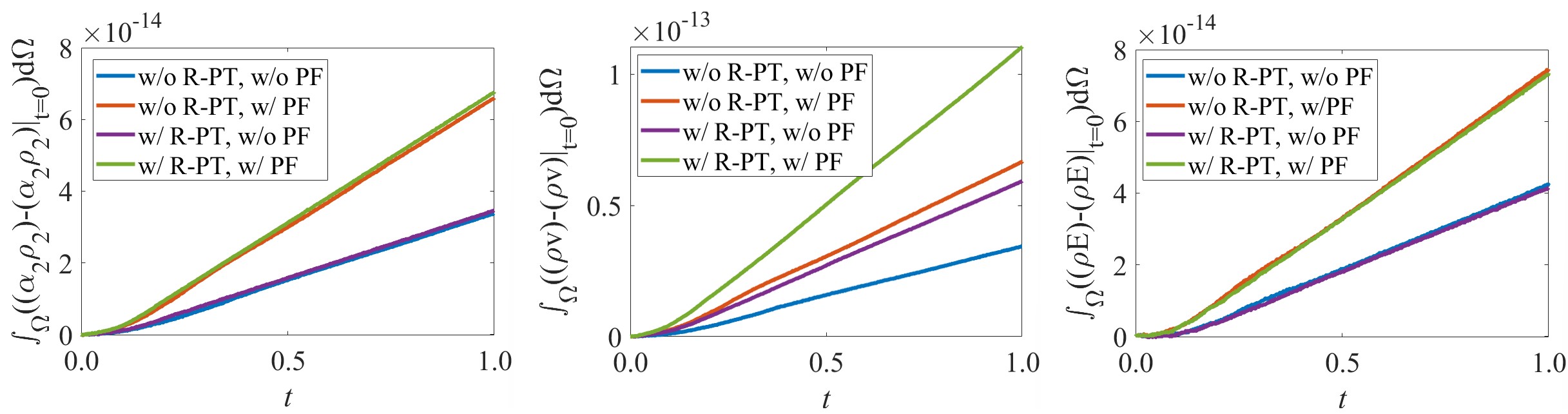}
	\caption{Time histories of the changes of Phase~$2$ mass (left), $y$-momentum (middle), and total energy (right) of the advection of an air square in water with or without the pressure-temperature relaxation.}\label{Fig:Advection-Conservation-PT}
\end{figure}

\subsection{Air-helium shock tube}\label{Sec:AirHelium}
We then consider the air-helium shock tube problem \citep{AbgrallKarni2001,Tiwarietal2013,HuangJohnsen2022} to include shock-interface interactions. Both helium (Phase~$1$: $\rho_1=0.125$, $\gamma_1=1.6$, and $C_1^V = 3115.72$) and air (Phase~$2$: $\rho_2=1$, $\gamma_2=1.4$, and $C_2^V = 717.5$) are modeled as ideal gases, and the initial velocity, pressure, and volume fraction are
\begin{equation}\label{Eq:IC-AirHelium}
\left(u,P,\alpha_1 \right)=\left \{
\begin{array}{cc}
     (0, 1.0, 0 ),&  0 \leqslant x < 1,\\
     (0, 0.1, 1 ),&  1 \leqslant x \leqslant 2.
\end{array}
\right.
\end{equation}
We use $200$ grid cells to discretize the domain with the outflow boundary conditions.

Fig.\ref{Fig:AirHelium-P} shows the mixture density, velocity, pressure, and helium volume fraction at $t=0.4$ with the pressure relaxation. The present result with or without the Phase-Field mechanism agrees well with the exact solution; the shock, rarefaction, and interface are accurately captured.
\begin{figure}[!t]
	\centering
	\includegraphics[scale=0.35]{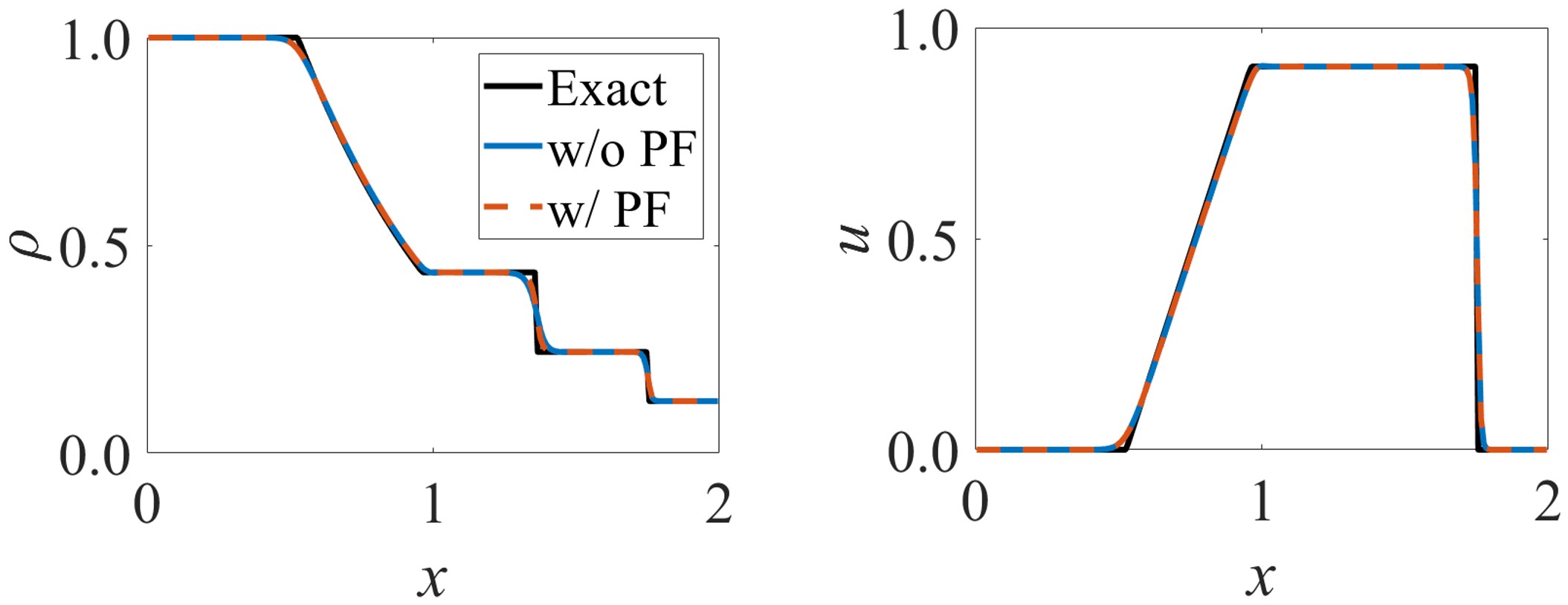}\\
    \includegraphics[scale=0.35]{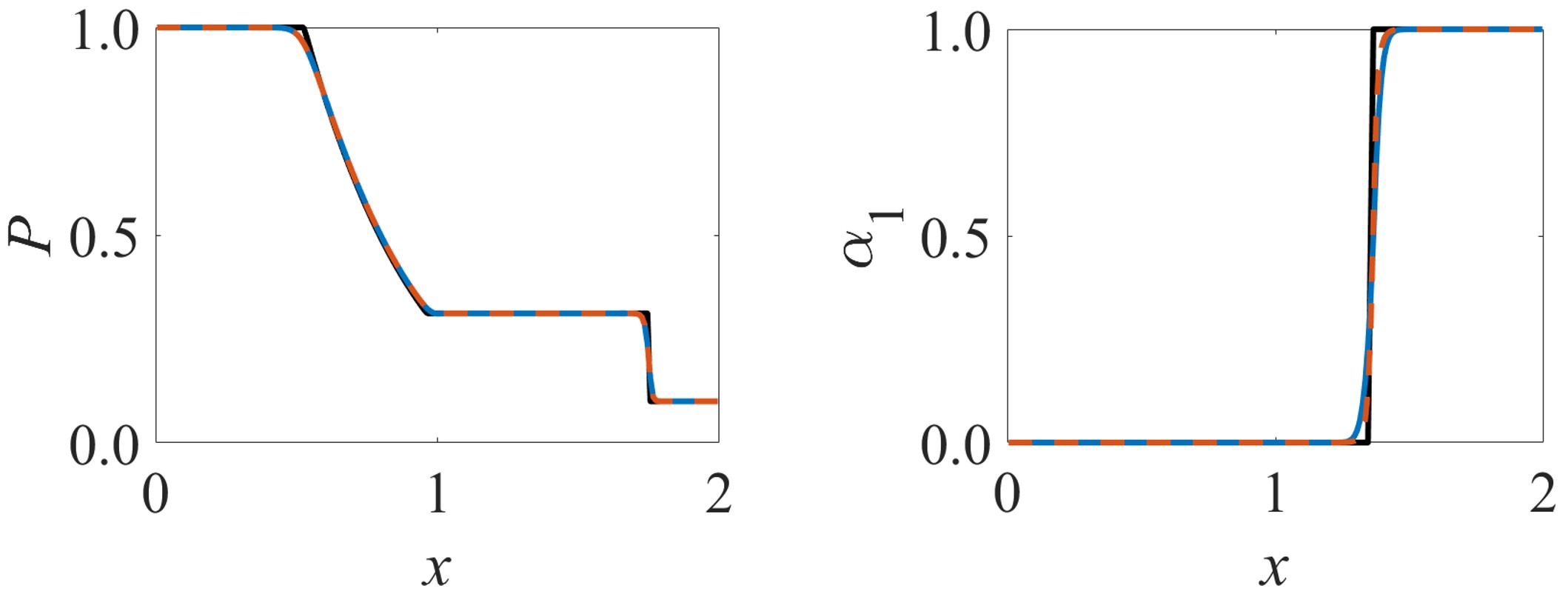}
	\caption{Mixture density (top left), velocity (top right), pressure (bottom left), and helium volume fraction (bottom right) of the air-helium shock tube problem at $t=0.4$ with the pressure relaxation.}\label{Fig:AirHelium-P}
\end{figure}

Fig.~\ref{Fig:AirHelium-PT} shows the same quantities with the pressure-temperature relaxation. Although there is no available exact solution with the pressure-temperature relaxation, we include the one with the pressure relaxation for comparison. The pressure-temperature relaxation slightly slows down the interface, while produces indistinguishable result away from the interface from that with the pressure relaxation.
\begin{figure}[!t]
	\centering
	\includegraphics[scale=0.35]{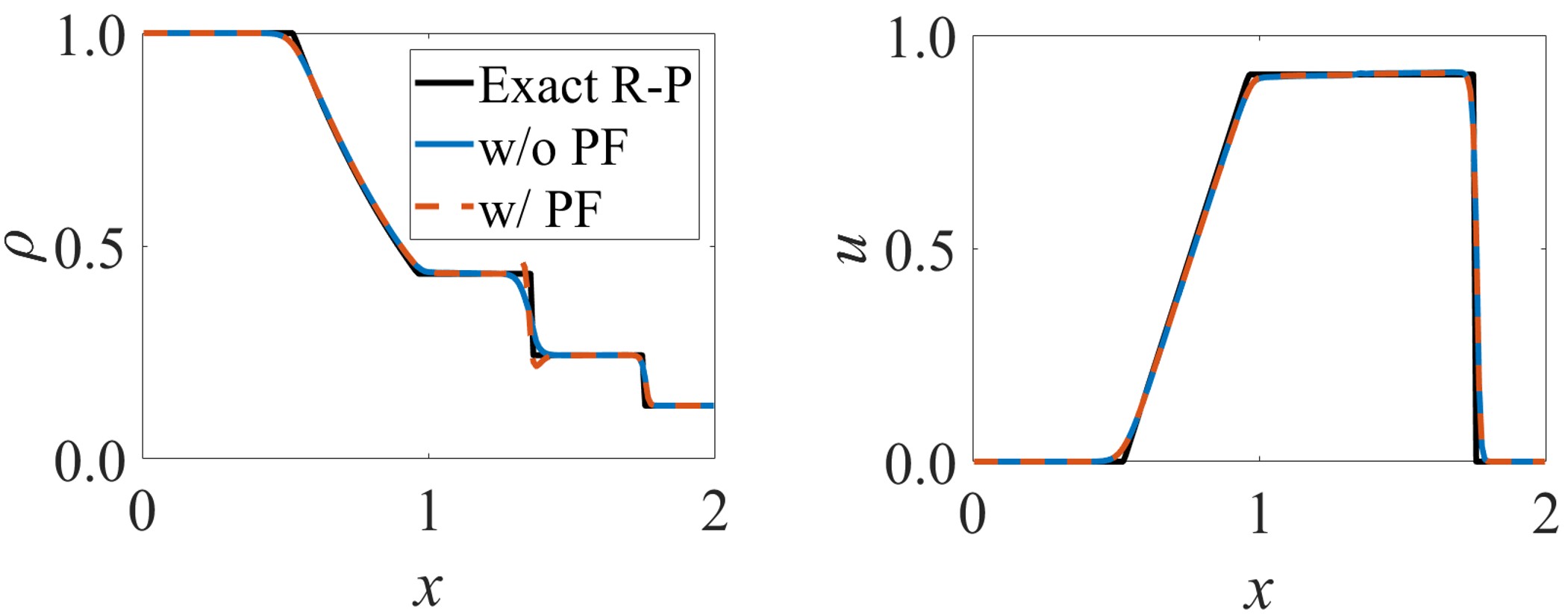}\\
    \includegraphics[scale=0.35]{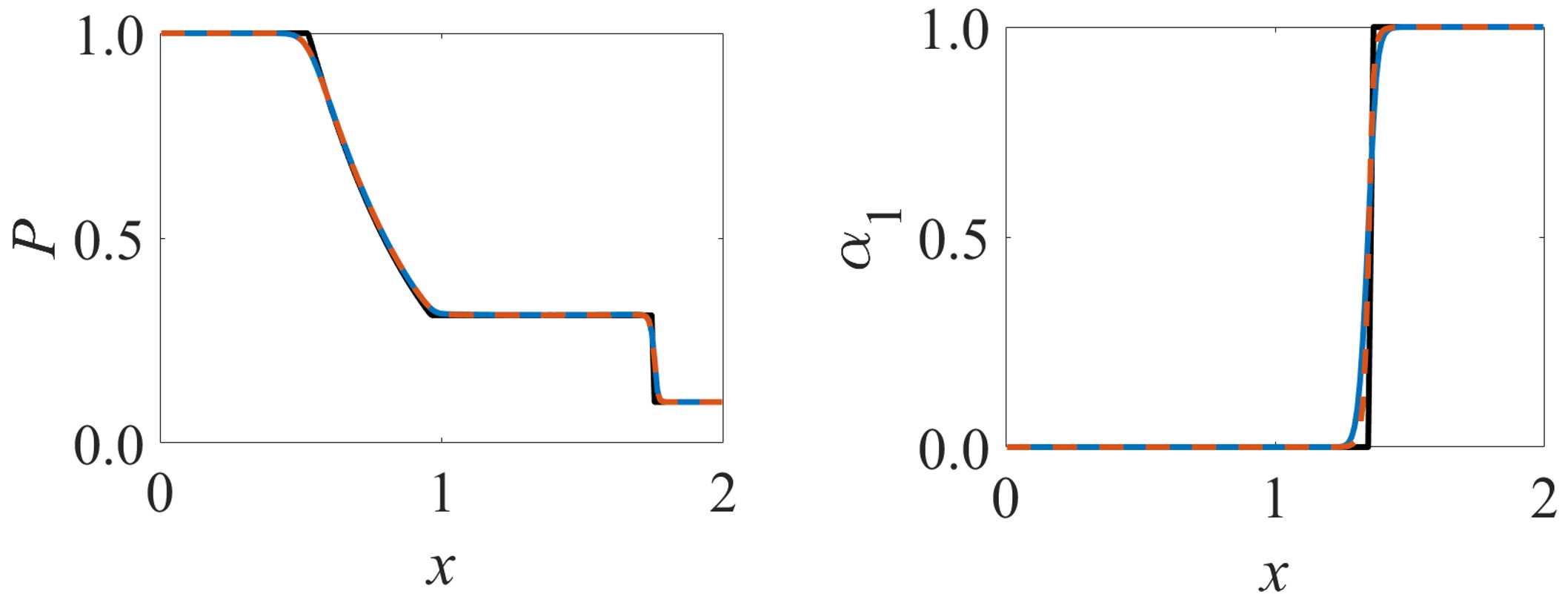}
	\caption{Mixture density (top left), velocity (top right), pressure (bottom left), and helium volume fraction (bottom right) of the air-helium shock tube problem at $t=0.4$ with the pressure-temperature relaxation.}\label{Fig:AirHelium-PT}
\end{figure}

Fig.~\ref{Fig:AirHelium-Thickness} shows the time histories of the interface thickness. Instead of increasing over time, the interface thickness maintains around $6$ grid cells with the Phase-Field mechanism, which demonstrates the effectiveness of the Phase-Field mechanism on preventing interface thickening.
\begin{figure}[!t]
	\centering
	\includegraphics[scale=.35]{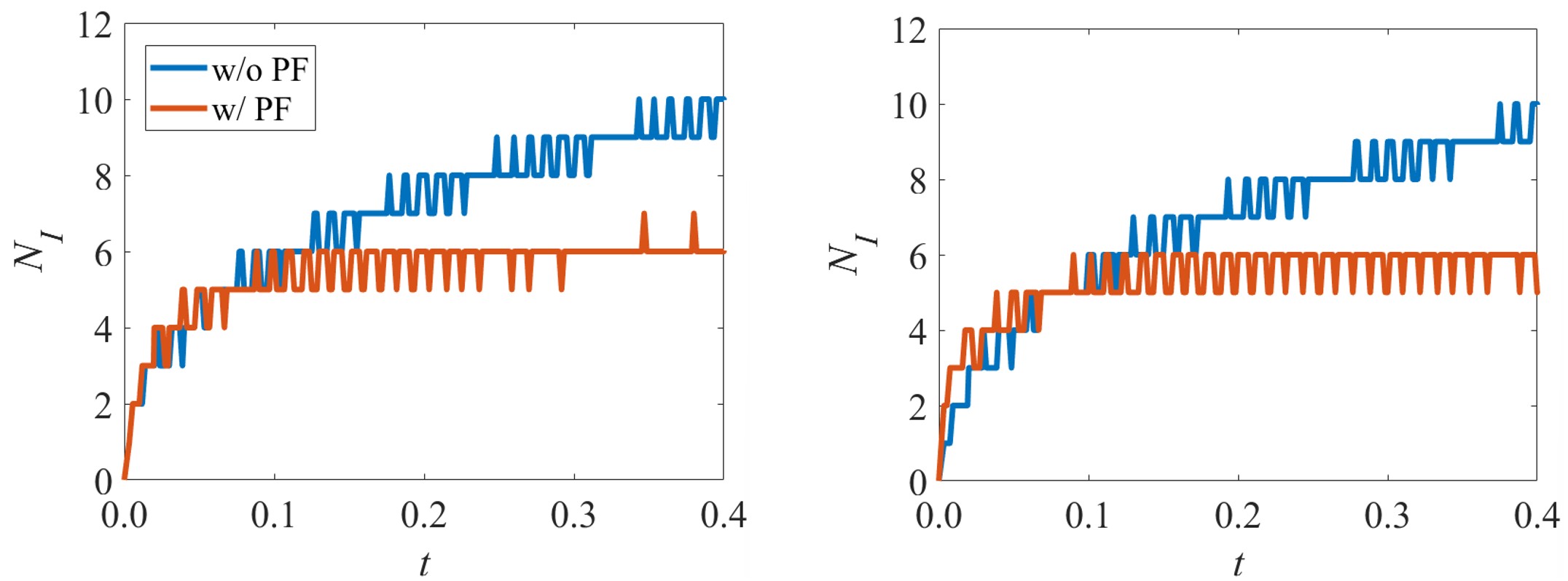}
	\caption{Time histories of the interface thickness of the air-helium shock tube problem with the pressure relaxation (left) and pressure-temperature relaxation (right).}\label{Fig:AirHelium-Thickness}
\end{figure}

\subsection{Cavitation}\label{Sec:Cavitation}
The cavitation problem \citep{Saureletal2009} is considered, which includes interface creation with a large density ratio.
A homogeneous mixture of air (Phase~$1$: $\rho_1 = 1$, $\gamma_1 = 1.4$, $P_1^\infty = 0$, and $C_1^V = 717.5$) and water (Phase~$2$: $\rho_2=1000$, $\gamma_2 = 4.4$, $P_2^\infty = 6\times 10^8$, and $C_2^V = 4186$) fills a unit-length tube in a pressure of $P=1 \times 10^5$. The amount of the air is small; its volume fraction is $\alpha_1 = 1\times10^{-2}$. The velocity is $u=100$ at $x > 0.5$, while $u=-100$ elsewhere, initially. The tube is discretized by $2000$ grid cells with the outflow boundary conditions. 

Fig.~\ref{Fig:Cavitation} shows the mixture density, velocity, pressure, and air volume fraction at $t = 1.85\times10^{-3}$, and both the results with the pressure relaxation and pressure-temperature relaxation are included. Since the air-water interface is not sharp in this problem, the Phase-Field mechanism is not activated \citep{HuangJohnsen2024}. The result with the pressure relaxation agrees well with the exact solution \citep{Petitpasetal2007} of the five-equation model of Kapila et al. \citep{Kapilaetal2001}; the air is produced in the middle of the tube with a density and pressure approaching to zero due to the strong symmetry rarefaction waves traveling from the middle to the two ends of the tube. In comparison to the numerical solution of the five-equation model of Kapila et al. in \citep{Schmidmayeretal2020,HuangJohnsen2024}, the present solution of the six-equation model does not produce pressure bumps at the edge of the rarefactions.

The result with the pressure-temperature relaxation has no significant difference from that with the pressure relaxation, except that the pressure near the two ends of the tube is much smaller. As explained in Section~\ref{Sec:Relaxation-PT}, even though the two phases have the same pressure value, the relaxed pressure of the pressure-temperature relaxation generally takes a different value, due to the additional requirement of matching the temperature of the two phases. In this problem, we have $P_1=P_2=1 \times 10^5$, $T_1=348.4321$, and $T_2=42.1644$, at the beginning, while $T^*=42.1597$ and $P^*=1.2143 \times 10^{4}$ after the pressure-temperature relaxation with a tiny increase of the water volume fraction to ensure the energy conservation.
\begin{figure}[!t]
	\centering
	\includegraphics[scale=.35]{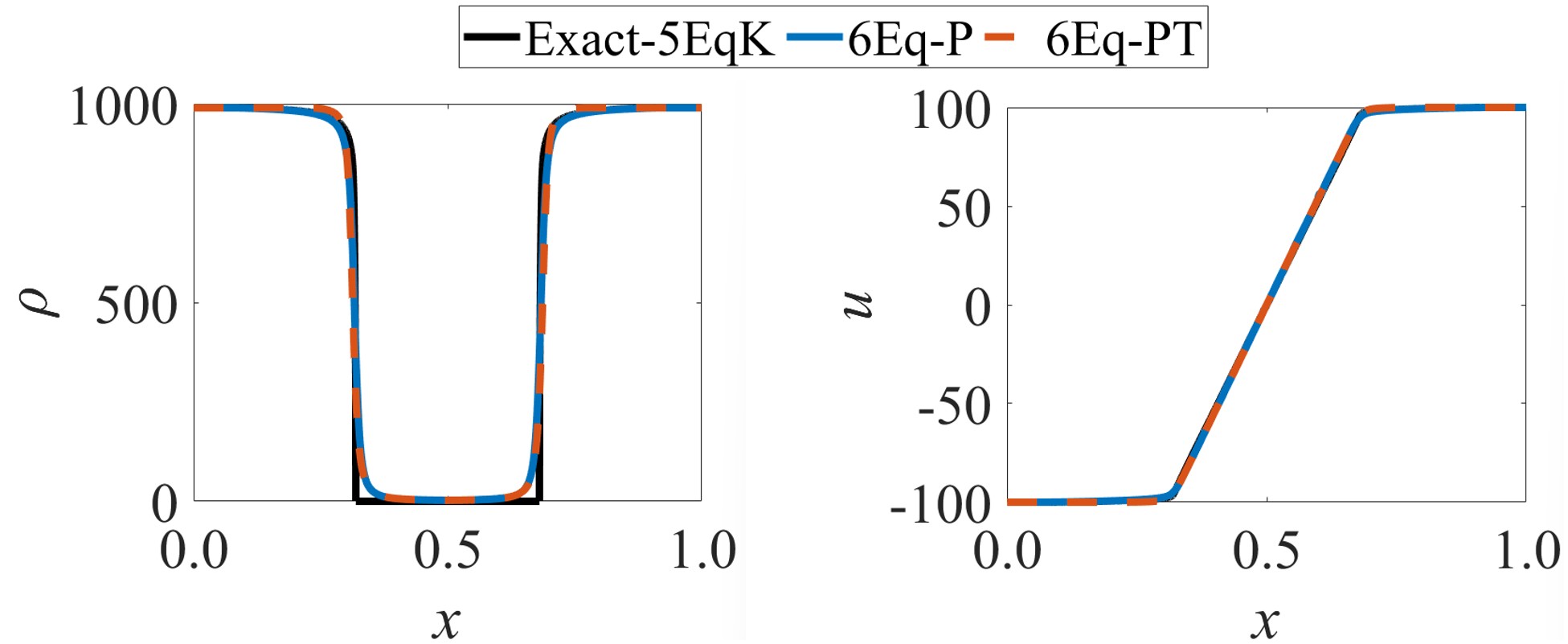}\\
    \includegraphics[scale=.35]{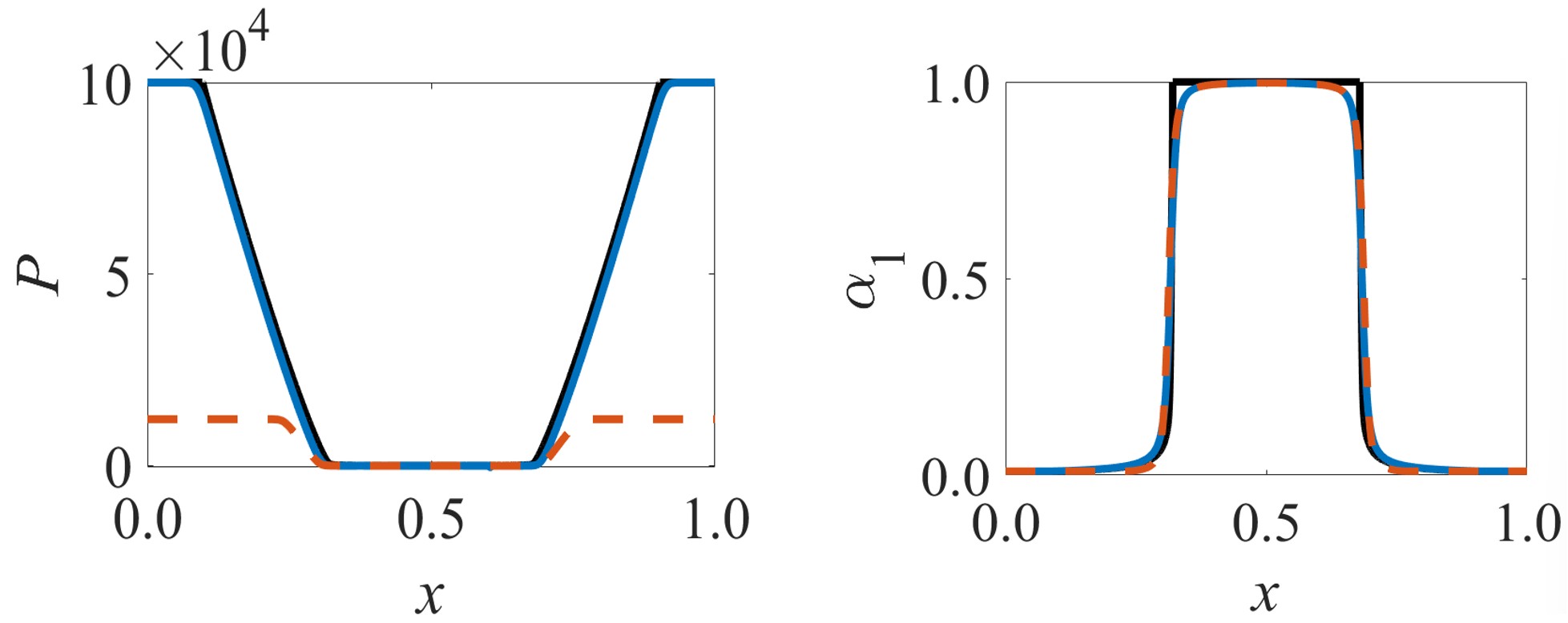}
	\caption{Mixture density (top left), velocity (top right), pressure (bottom left), and air volume fraction (bottom right) of the cavitation problem at $t=1.85\times10^{-3}$ with the pressure relaxation and pressure-temperature relaxation.}\label{Fig:Cavitation}
\end{figure}

\subsection{Water–air shock tube}\label{Sec:WaterAirShockTube}
The water-air shock tube problem \citep{Saureletal2009} is considered to discuss the effect of interface thickness.
In a unit-length tube, air (Phase~$1$: $\rho_1 = 1$, $\gamma_1 = 1.4$, $P_1^\infty  = 0$, and $C_1^v = 717.5$) fills the right chamber ($x > 0.75$) in a pressure of $P = 0.1 \times 10^{6}$, while water (Phase~$2$: $\rho_2 = 1000$, $\gamma_2 = 4.4$, $P_2^\infty = 6 \times 10^{8}$, and $C_2^V = 4186$) fills the rest of the tube in a pressure of $P = 1 \times 10^{9}$. The air and water are not pure; a tiny amount of water is dissolved in the right chamber with a volume fraction $\check{\alpha}=10^{-6}$, and the same volume fraction of air is dissolved in the left chamber. Both the air and water are stationary initially. The domain is discretized by $1000$ grid cells with the outflow boundary conditions.

Fig.~\ref{Fig:WaterAir-P} shows the mixture density, velocity, pressure, and air volume fraction at $t = 240 \times 10^{-6}$ with the pressure relaxation. Both results without and with the Phase-Field mechanism are included, and, overall, they agree well with the exact solution. Without the Phase-Field mechanism, a velocity error near the air-water interface is observed, while this error disappears when we use $\eta/\Delta x=0.1$ in the Phase-Field mechanism. 
\begin{figure}[!t]
	\centering
	\includegraphics[scale=.35]{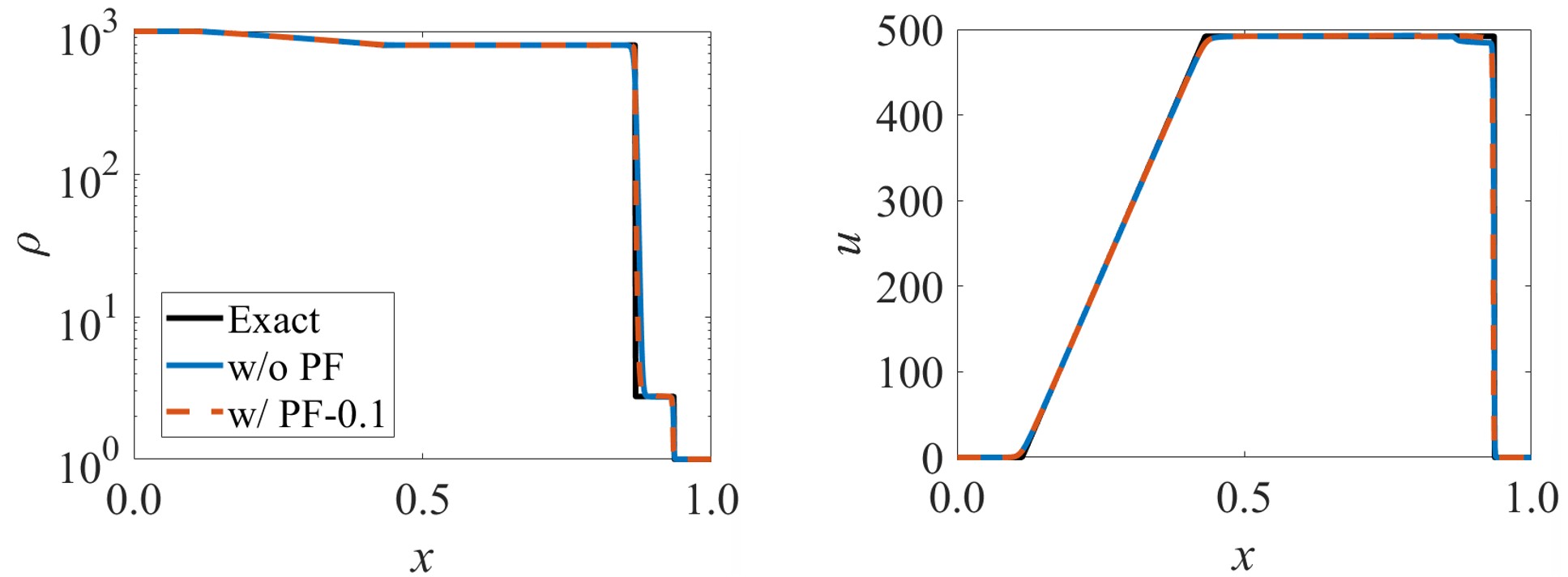}\\
    \includegraphics[scale=.35]{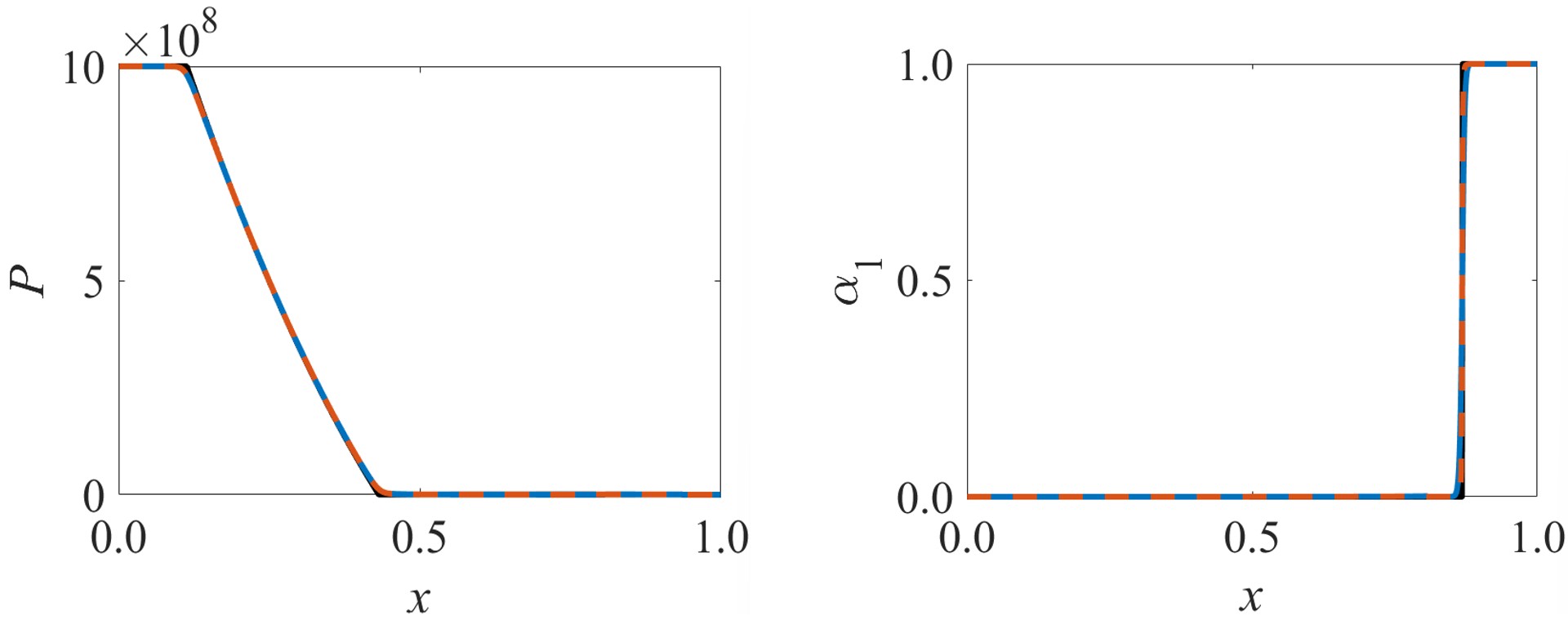}
	\caption{Mixture density (top left), velocity (top right), pressure (bottom left), and air volume fraction (bottom right) of the water-air shock tube problem at $t = 240 \times 10^{-6}$ with the pressure relaxation.}\label{Fig:WaterAir-P}
\end{figure}

This problem poses considerable numerical challenges, as the substantial density and pressure ratios render the solution highly susceptible to the interface thickness.
Fig.~\ref{Fig:WaterAir-P-Compare} compares the velocity at $t = 240 \times 10^{-6}$ and the time history of interface thickness with different $\eta/\Delta x$. As we reduce the interface thickness, i.e., use a smaller value of $\eta/\Delta x$, the velocity error decreases as well even under the same grid resolution, demonstrating the existence of sharp-interface limit.
With $\eta/\Delta x=1$ (the default value), although it stops growing and becomes thinner than that without the Phase-Field mechanism at later time, the interface thickness grows faster in the early stage. As the waves first initiate at the interface, they actually propagate in a thicker interface when the Phase-Field mechanism with $\eta/\Delta x=1$ is activated, resulting in a larger error. In addition to oscillations, the transmitted shock travels faster than it should.
The early-stage growth of the interface thickness with $\eta/\Delta x=0.5$ is almost identical to that without the Phase-Field mechanism, and thus these two results behave similarly. When we use $\eta/\Delta x=0.1$, the interface thickness is always below that without the Phase-Field mechanism, and the best result is obtained.
We test the HLL Riemann solver \citep{Toro2009} and the HLLC Riemann solver in \citep{DeLorenzoetal2018} for the six-equation model and the numerical methods \citep{Tiwarietal2013,Schmidmayeretal2020,HuangJohnsen2022} for the five-equation model of Kapila et al. \citep{Kapilaetal2001}, but none of them successfully solves this problem. A recent study discusses about the effect of different Riemann solvers on the six-equation model based on the phasic total energy, and interested readers can refer to \citep{Orlandoetal2026}.
Moreover, the approach in \citep{Mirjalilietal2020,Jainetal2020} requires $\eta/\Delta x \geqslant (2+\sqrt{2})/2 \approx 1.707$ under the present setup to preserve the volume fraction boundedness, while a much smaller value of $\eta/\Delta x$ is enabled by the present approach without violating bound preservation.
\begin{figure}[!t]
	\centering
	\includegraphics[scale=.35]{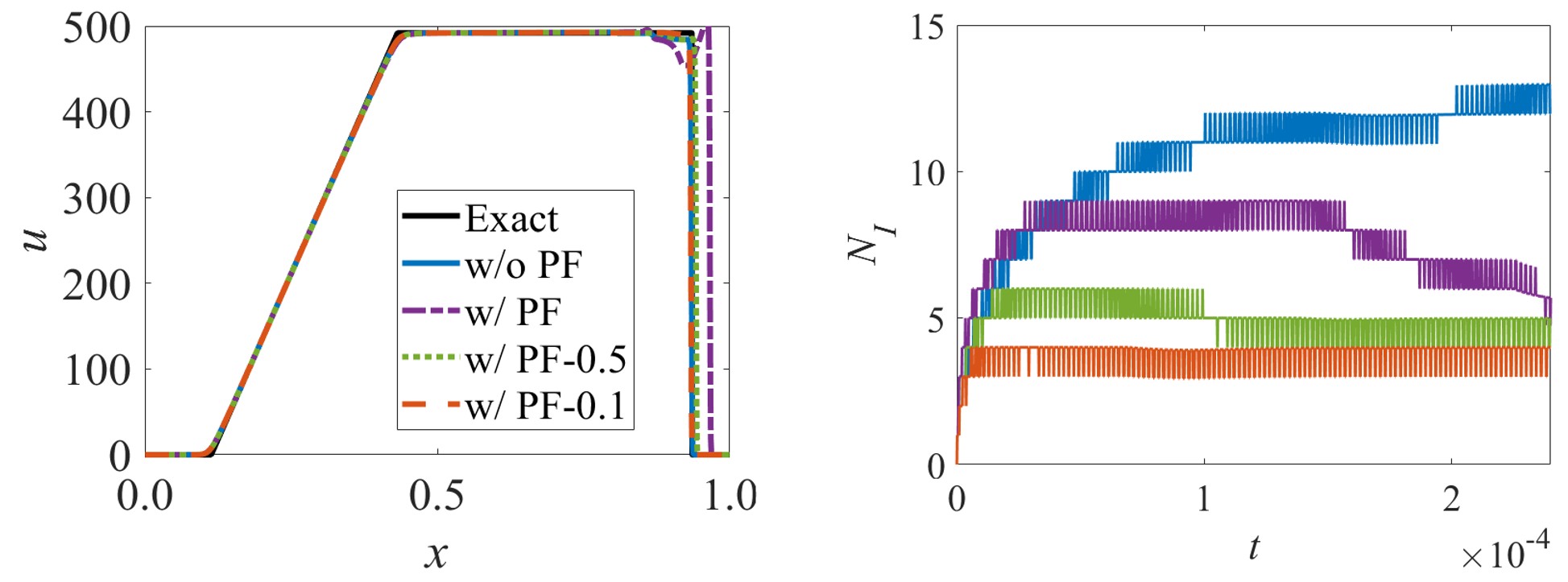}
	\caption{Effect of interface thickness on the solution of the water-air shock tube problem. Left: velocity at $t = 240 \times 10^{-6}$; Right: time history of interface thickness, with different $\eta/\Delta x$.}\label{Fig:WaterAir-P-Compare}
\end{figure}

Fig.~\ref{Fig:WaterAir-PT} shows the mixture density, velocity, pressure, and air volume fraction at $t = 240 \times 10^{-6}$ with the pressure-temperature relaxation. Both results without and with the Phase-Field mechanism are included, along with the exact solution of the one with the pressure relaxation for comparison. On the water-rich side, the difference between the two relaxations is negligible, while the transmitted shock on the air-rich side is not observed with the pressure-temperature relaxation. We test different values of $\check{\alpha}$ down to $10^{-14}$, and all the results overlap.
\begin{figure}[!t]
	\centering
	\includegraphics[scale=.35]{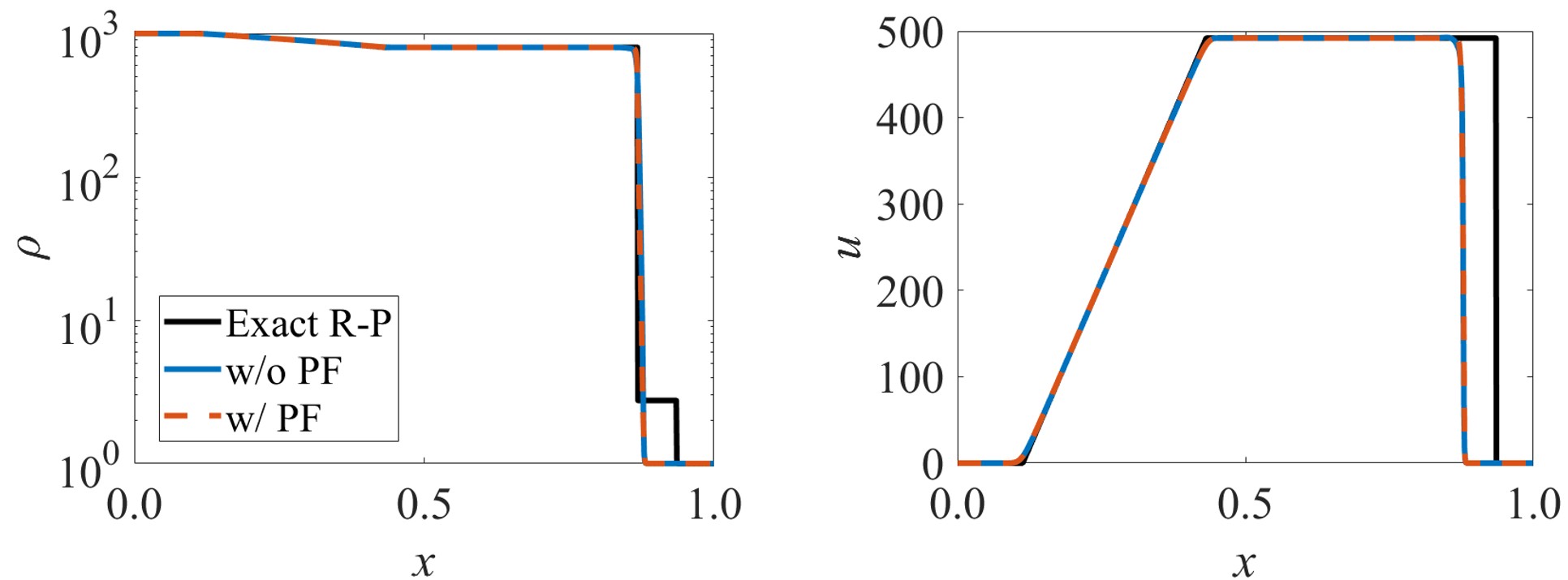}\\
    \includegraphics[scale=.35]{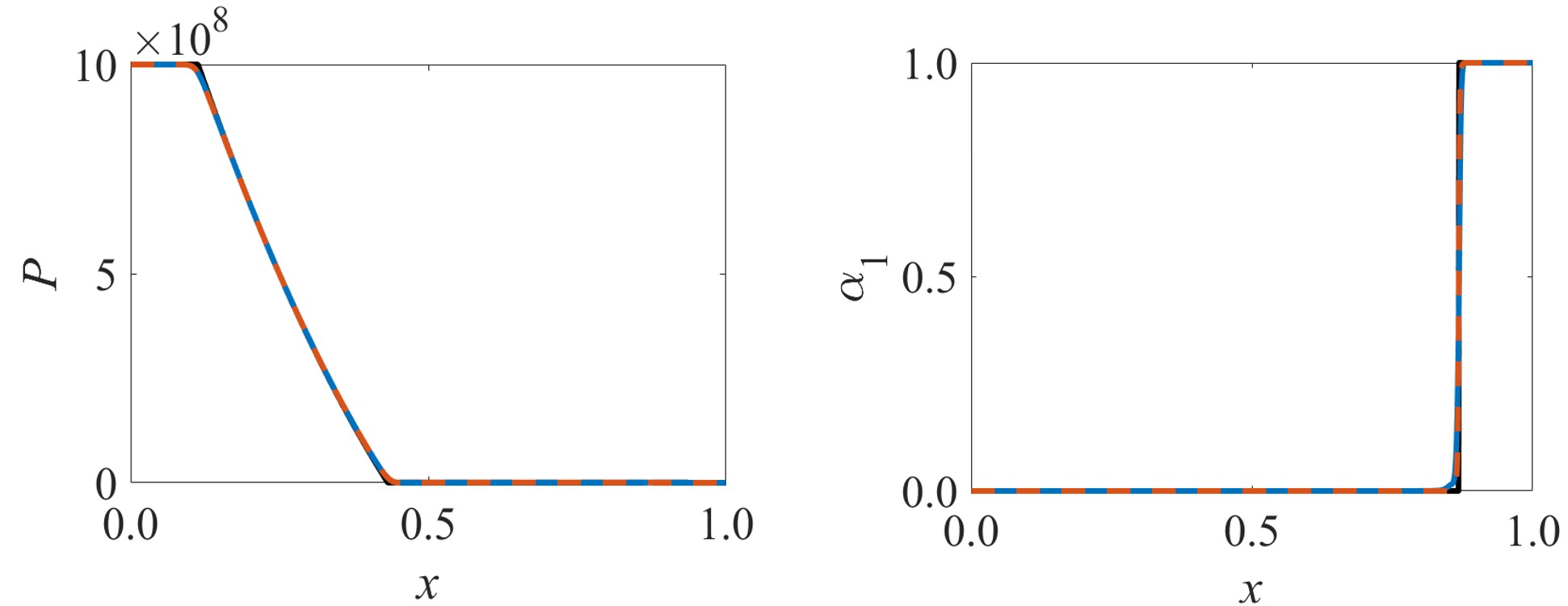}
	\caption{Mixture density (top left), velocity (top right), pressure (bottom left), and air volume fraction (bottom right) of the water-air shock tube problem at $t = 240 \times 10^{-6}$ with the pressure-temperature relaxation.}\label{Fig:WaterAir-PT}
\end{figure}
Fig.~\ref{Fig:WaterAir-PT-Thickness} shows the interface thickness versus time with the pressure-temperature relaxation, and the interface with the Phase-Field mechanism is always sharper than the one without the Phase-Field mechanism even when using the default value $\eta/\Delta x=1$.
\begin{figure}[!t]
	\centering
	\includegraphics[scale=.5]{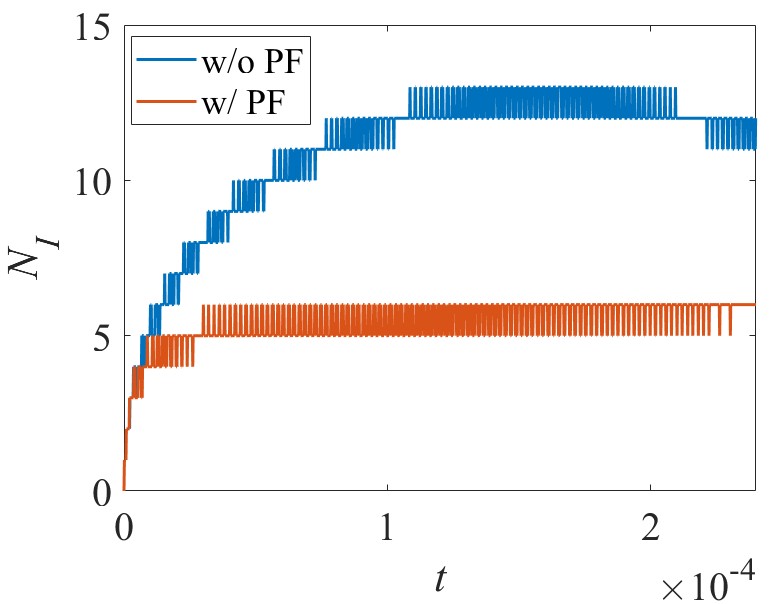}
	\caption{Time history of interface thickness of the water-air shock tube problem with the pressure-temperature relaxation.}\label{Fig:WaterAir-PT-Thickness}
\end{figure}

\subsection{Shock-bubble interaction}\label{Sec:ShockBubble}
We consider the shock-bubble interaction problem \citep{HaasSturtevant1987} to include a large interface deformation. In a channel filled with air (Phase~$1$: $\rho_1=1.204$, $\gamma_1=1.4$, $C_1^V = 717.5$), a stationary cylindrical helium bubble (Phase~$2$: $\rho_2= 0.166$, $\gamma_2=1.667$, and $C_2^V = 3115.72$) at the centerline of the channel is impacted by a $Ma$ $1.22$ shock. Both the unshocked air and helium have a pressure of $P=101325$. The channel height is $2L=0.089$, and the bubble radius is $r=0.025$. Due to symmetry, we only compute the upper half of the domain, which is $[-L,4L]\times[0,L]$, and the shock is initially at $x=-0.03$. The domain is discretized by $960 \times 192$ grid cells, with the outflow boundaries at the left and right, free-slip boundary at the top, and symmetry boundary at the bottom.

Fig.~\ref{Fig:ShockBubble-P} and Fig.~\ref{Fig:ShockBubble-PT} show the mixture density with the pressure relaxation and pressure-temperature relaxation, respectively, at selected moments corresponding to those in \citep{HaasSturtevant1987}. The evolution of both the bubble shape and wave patten is consistent with previous experimental \citep{HaasSturtevant1987} and computational \citep{QuirkKarni1996,MarquinaMulet2003,JohnsenColonius2006,CoralicColonius2014} studies, and the difference between the pressure relaxation and pressure-temperature relaxation is hardly observed in this problem.
The helium bubble starts moving to the right after impacted by the shock. The transmitted shock in the helium bubble and the reflected rarefaction in the air are well captured. The rarefaction then reflects off the channel wall and subsequently interacts with the bubble. A high-speed jet is formed in the middle of the domain, which deforms and finally breaks the bubble.
The conservative Phase-Field formulation \citep{ChiuLin2011,Mirjalilietal2020} of the Phase-Field mechanism used in the present study breaks filaments that are too thin to be resolved by the mesh size into multiple tiny bubbles, as in \citep{Jainetal2023}. This behavior is different from the conservative Allen-Cahn formulation \citep{BrasselBretin2011,Huangetal2020B}, which tends to dissolve these small structures \citep{HuangJohnsen2022,Huangetal2025}. More discussion of the behavior of different Phase-Field formulations in compressible multiphase flows is available in \citep{Huangetal2024}.
\begin{figure}[!t]
	\centering
	\includegraphics[scale=.39]{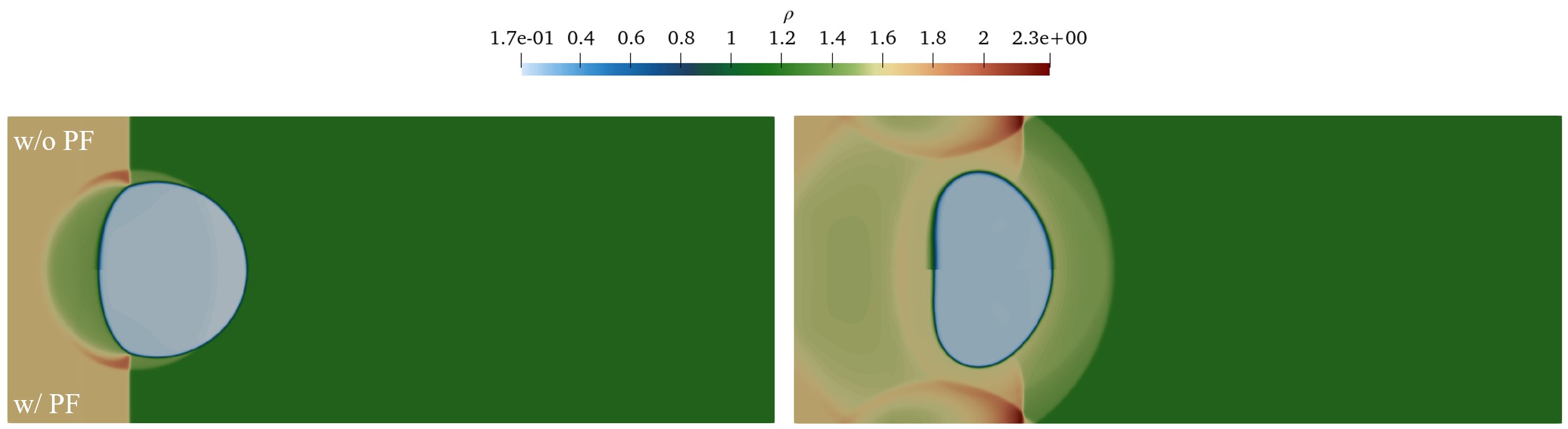}\\
	\includegraphics[scale=.39]{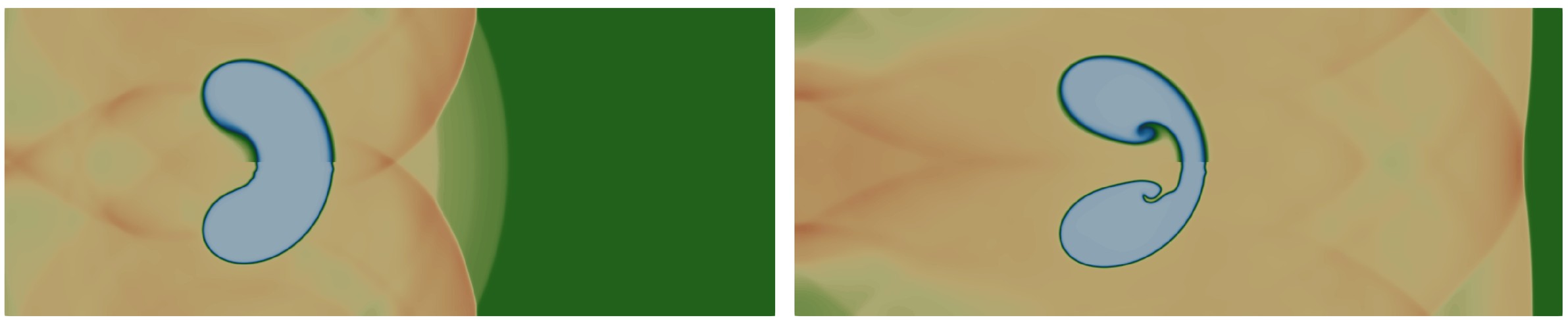}\\
	\includegraphics[scale=.39]{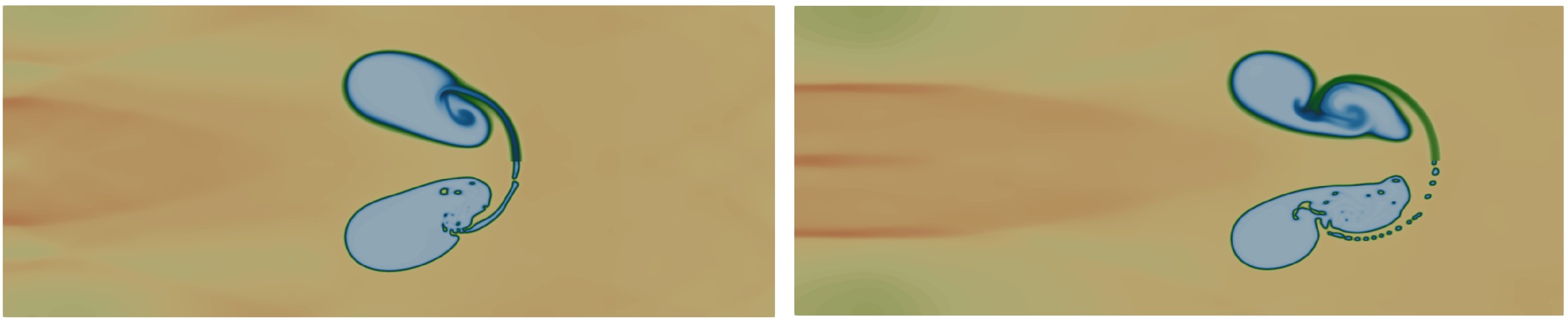}
	\caption{Mixture density of the shock-bubble interaction problem at $t=50$, $120$, $260$, $440$, $640$, and $890 \times 10^{-6}$ (from left to right and top to bottom) with the pressure relaxation.}\label{Fig:ShockBubble-P}
\end{figure}
\begin{figure}[!t]
	\centering
	\includegraphics[scale=.39]{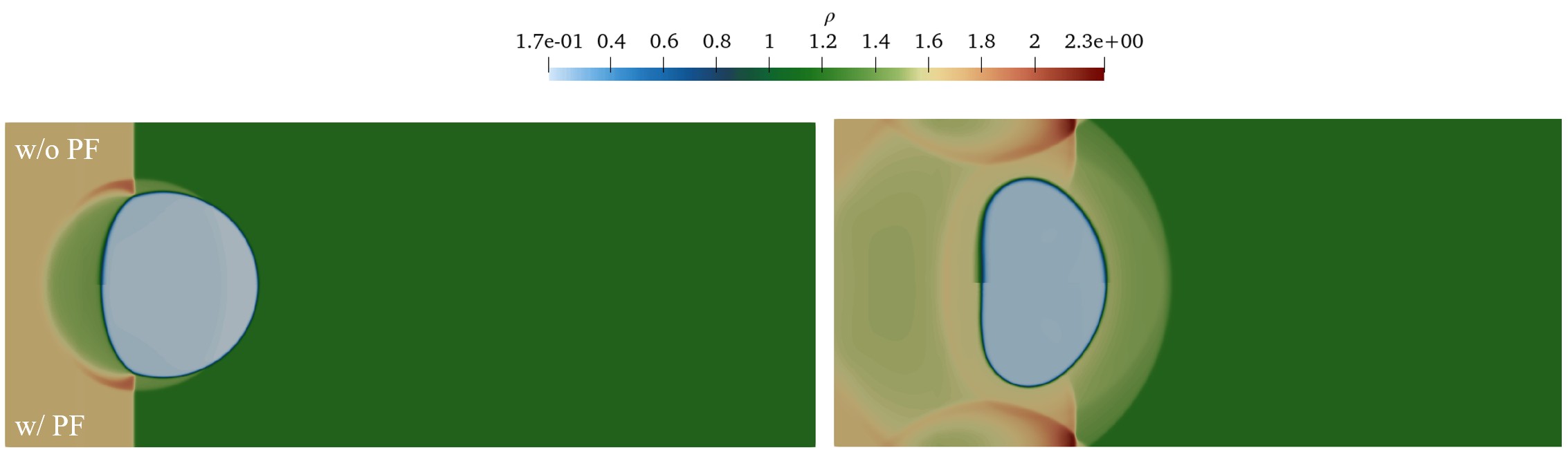}\\
	\includegraphics[scale=.39]{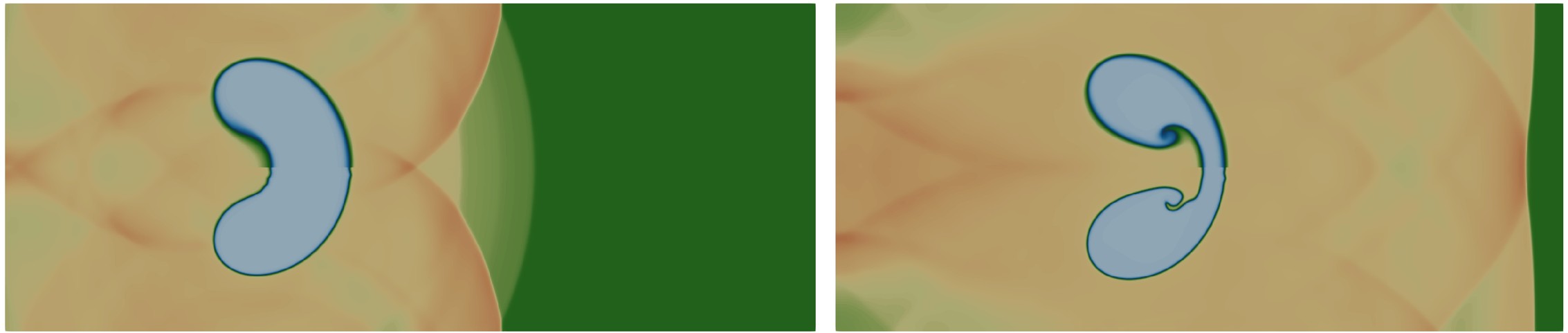}\\
	\includegraphics[scale=.39]{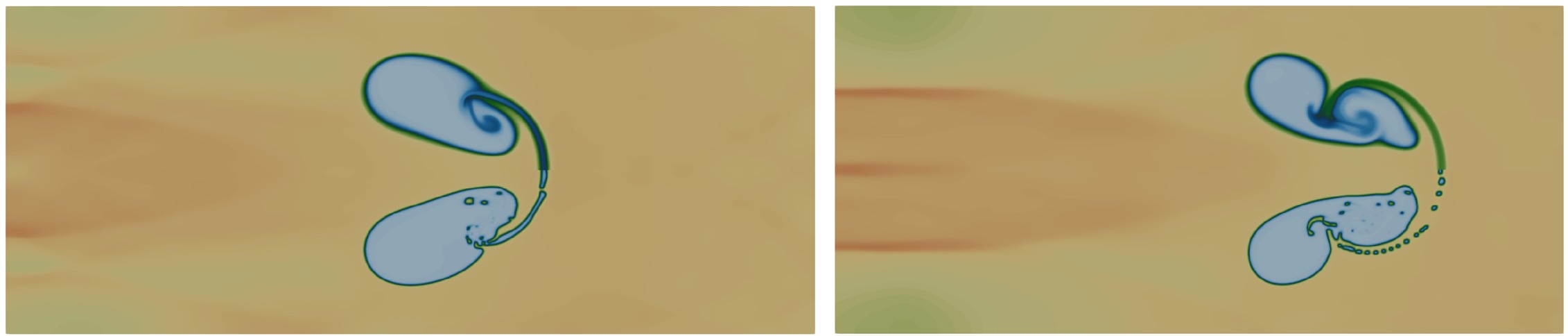}
	\caption{Mixture density of the shock-bubble interaction problem at $t=50$, $120$, $260$, $440$, $640$, and $890 \times 10^{-6}$ (from left to right and top to bottom) with the pressure-temperature relaxation.}\label{Fig:ShockBubble-PT}
\end{figure}

As shown in Fig.~\ref{Fig:ShockBubble-Thickness}, the interface thickness increases quickly without the Phase-Field mechanism, while it remains nearly unchanged with the Phase-Field mechanism even including shock impacts and interface breakups in this problem.
\begin{figure}[!t]
	\centering
	\includegraphics[scale=.39]{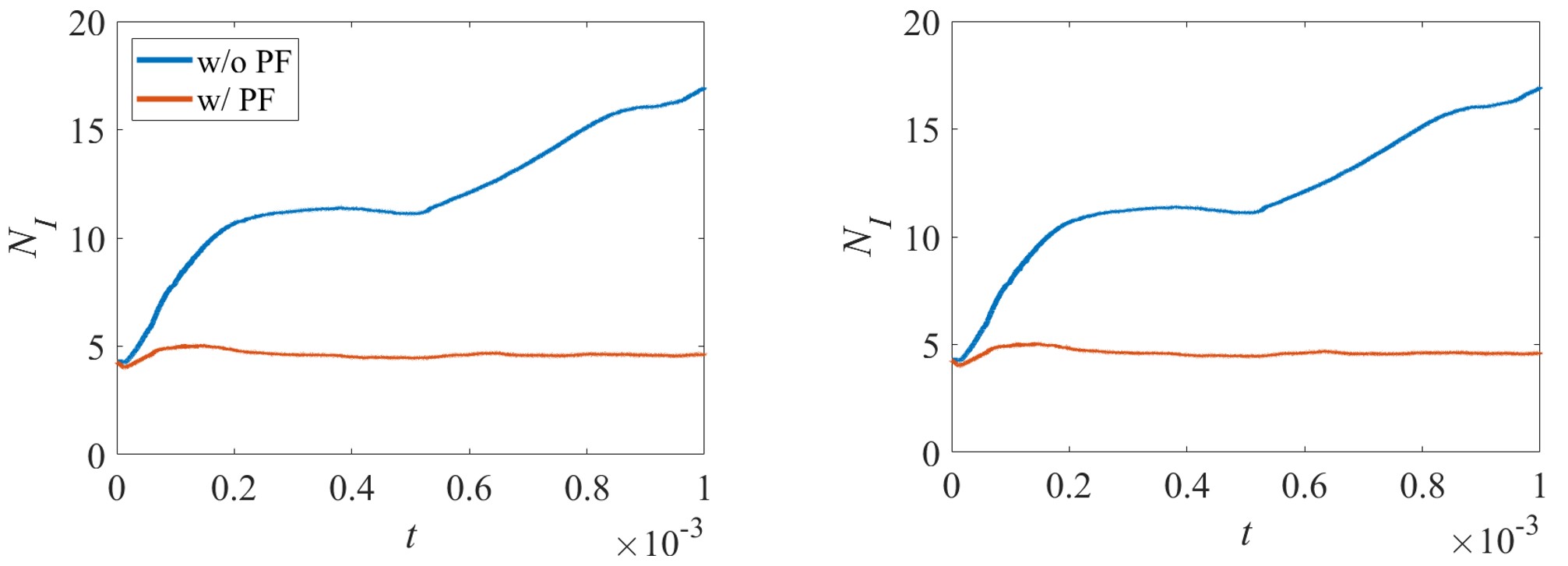}
	\caption{Time histories of the interface thickness of the shock-bubble interaction problem with the pressure relaxation (left) and pressure-temperature relaxation (right).}\label{Fig:ShockBubble-Thickness}
\end{figure}

\subsection{Spherical bubble collapse}\label{Sec:BubbleCollapse}
We finally consider the spherical bubble collapse problem \citep{Tiwarietal2013,Schmidmayeretal2020,Huangetal2025} to further demonstrate the present approach in three dimensions.
A gas bubble (Phase~$1$: $\rho_1 = 1$, $\gamma_1 = 1.4$, $P_1^\infty  = 0$, and $C_1^V = 717.5$) has a uniform pressure $P^B = 1 \times 10^{4}$, while water (Phase~$2$: $\rho_2 = 1000$, $\gamma_2 = 2.35$, $P_2^\infty  = 1 \times 10^{9}$, and $C_2^V = 4186$) surrounding the bubble has a pressure following the Rayleigh-Plesset equation
\begin{equation}\label{Eq:SphericalBubbleCollapse}
P = P^L + \frac{R^B}{r}\left(P^B-P^L\right),
\end{equation}
where $P^L = 10P^B$ is the pressure at infinity, $R^B = 0.0005$ is the bubble radius, and $r$ is the distance from the bubble center.
Due to symmetry, only an octant of the problem is computed, and the bubble center coincides with the coordinate origin. To avoid boundary effects, the domain is $[0,80R^B]\times[0,80R^B]\times[0,80R^B]$, with the symmetry boundary conditions at the inner boundaries and the outflow boundary conditions at the outer boundaries. For this large domain in comparison to the bubble radius, it is computationally unfeasible to use a uniform mesh. A stretching mesh was used in \citep{Schmidmayeretal2020}, while adaptive mesh refinement was used in \citep{Tiwarietal2013,Huangetal2025}. In the present study, we adaptively refine the mesh based on the initial conditions. Specifically, on a $64 \times 64 \times 64$ base mesh, six levels of mesh refinement are performed when the volume fraction is greater than $0.01$, the difference of pressure from its nearest neighbors is greater than $0.5$, or $(x,y,z) \in \left(0,1.5R^B\right) \times \left(0,1.5R^B\right) \times \left(0,1.5R^B\right)$, as shown in Fig.~\ref{Fig:SphericalBubbleCallapse-Mesh}. Once the mesh is generated, it is fixed during the entire computation.
\begin{figure}[!t]
	\centering
	\includegraphics[scale=.35]{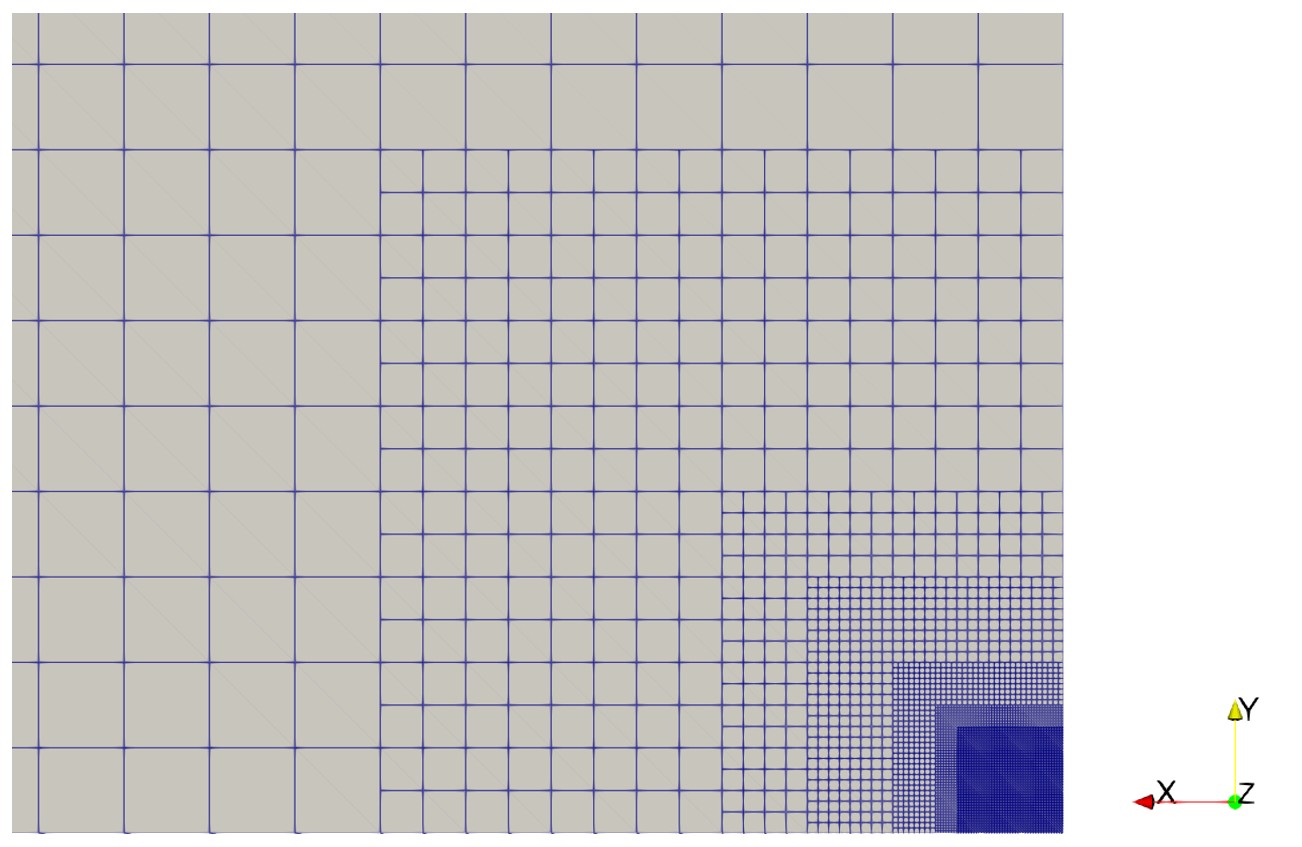}
	\caption{Schematic of the $64 \times 64 \times 64$ base mesh with six levels of mesh refinement near the bubble for the spherical bubble collapse problem.}\label{Fig:SphericalBubbleCallapse-Mesh}
\end{figure}

Fig.~\ref{Fig:SphericalBubbleCallapse-Radius} shows the bubble radius versus time. Here, the bubble radius $R=\left(\frac{3}{4 \pi} \mathcal{V}^B\right)^{1/3}$ is obtained from its volume $\mathcal{V}^B = \int_{\Omega}\alpha_1d\Omega$, and $t^c=0.915 R^B \sqrt{\rho^L/P^L}$ is the Rayleigh collapse time.
The results with the pressure relaxation agree well with the Keller-Miksis model \citep{KellerMiksis1980}, regardless of whether the Phase-Field mechanism is included, although a very minor delay of bubble rebound is observed; the same behavior was reported in \citep{Schmidmayeretal2020}.
When the pressure-temperature relaxation is performed, the bubble is less compressed and rebounds slightly later than the Keller-Miksis model \citep{KellerMiksis1980}, resulting in a larger bubble size after the rebound.
\begin{figure}[!t]
	\centering
	\includegraphics[scale=.4]{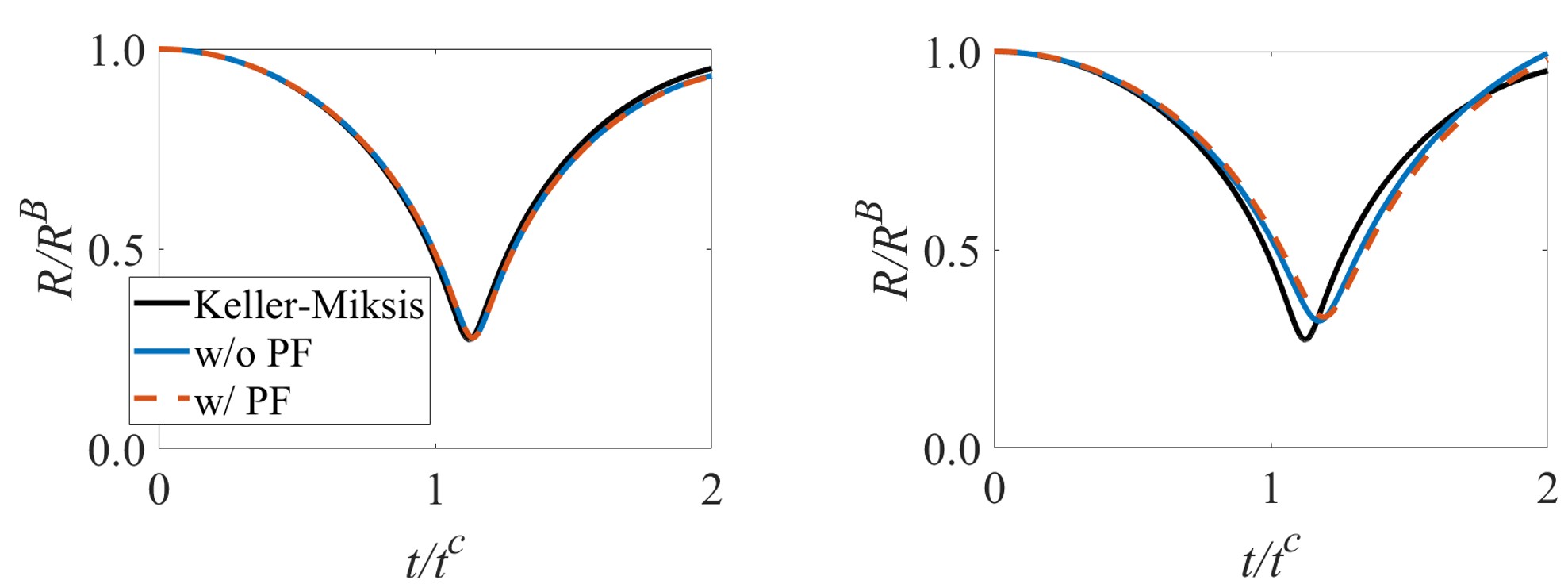}
	\caption{Time histories of bubble radius of the spherical bubble collapse problem with the pressure relaxation (left) and the pressure-temperature relaxation (right).}\label{Fig:SphericalBubbleCallapse-Radius}
\end{figure}

Fig.~\ref{Fig:SphericalBubbleCallapse-Thickness} shows the interface thickness versus time. For both the pressure and pressure-temperature relaxations, the interface thickness with the Phase-Field mechanism has a faster early-stage growth, while maintaining a constant value afterwards. The fast compression and expansion of the bubble prevent the growth of interface thickness when the Phase-Field mechanism is not activated; the interface thickness is always below 6 grid cells, compared to more than 10 grid cells in the cases investigated above. Furthermore, the pressure-temperature relaxation results in a sharper interface than that with the pressure relaxation.
\begin{figure}[!t]
	\centering
	\includegraphics[scale=.4]{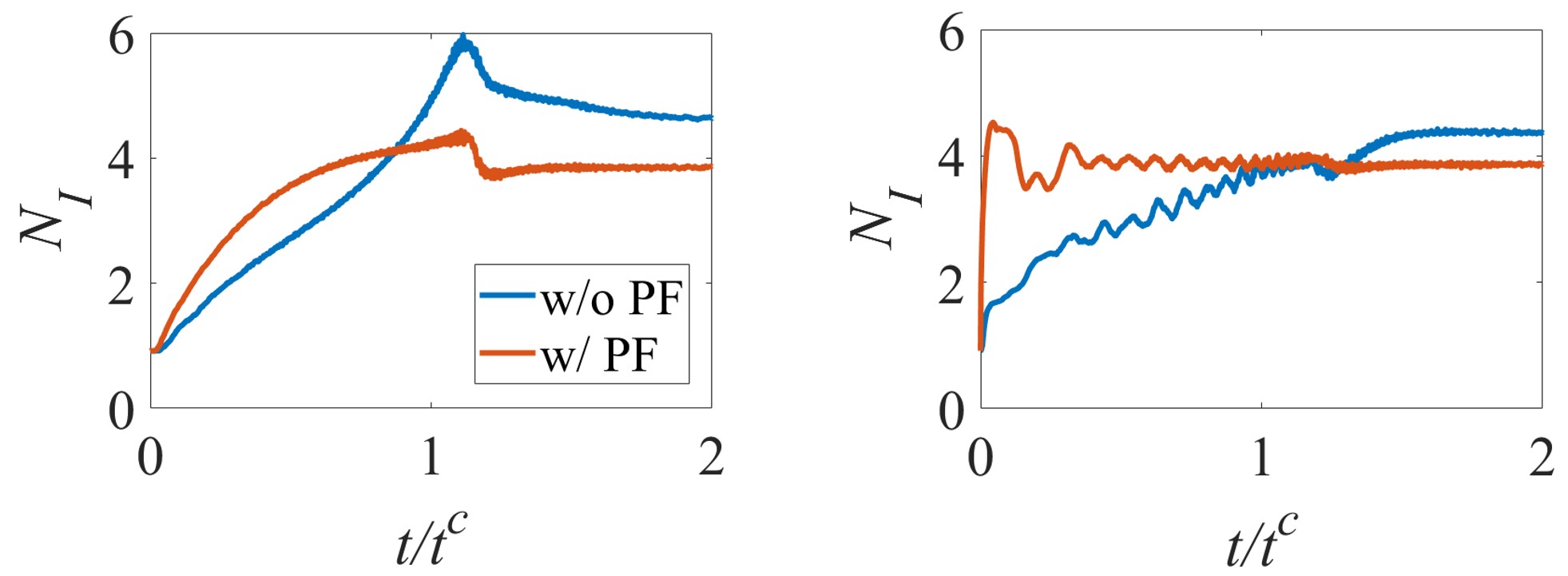}
	\caption{Time histories of interface thickness of the spherical bubble collapse problem with the pressure relaxation (left) and the pressure-temperature relaxation (right).}\label{Fig:SphericalBubbleCallapse-Thickness}
\end{figure}

\section{Conclusion}\label{Sec:Conclusions}
In the present study, the consistent and conservative Phase-Field method for compressible multiphase flows \citep{HuangJohnsen2022,HuangJohnsen2023,HuangJohnsen2024}, originally developed based on the five-equation models, is extended to the six-equation model, with rigorous derivation of the mathematical model, moderate modifications to the numerical approach, and theoretical analysis of the relaxations.

The proposed six-equation model with the Phase-Field mechanism is derived from the consistency conditions \citep{Huangetal2020,Huangetal2020N,Huangetal2020NPMC,Huangetal2020Solid,Huang2021}, which determine the transport of mass, momentum, and energy in the conservation laws, and the second law of thermodynamics, which finalizes the work between phases and the volume fraction equation. Our derivation is general, which does not rely on a specific formulation of the Phase-Field mechanism or have a limit on the number of different phases.
In addition to mass, momentum, and energy conservation as well as volume fraction summation to unity, the resulting model satisfies Galilean invariance, which leads to the kinematic, mechanical, and thermal equilibria at isolated interfaces, and local consistency of reduction \citep{BoyerMinjeaud2014,Dong2018,Huangetal2020N,Huangetal2020B}, which prevents the production of fictitious phases, local voids, or overfilling from modifying the multiphase dynamics. Furthermore, the model automatically recovers the previous two-phase six-equitation models \citep{Saureletal2009,PelantiShyue2014} in the two-phase regions (in the absence of the Phase-Field mechanism) and the Navier-Stokes equations in the single-phase regions.
We also analyze the isobaric closure achieved by the pressure relaxation, which results in a reduced multiphase five-equation model, and discuss the incompressible limit, which recovers the consistent and conservative Phase-Field model for incompressible multiphase flows \citep{Huangetal2020,Huangetal2020N,Huangetal2020Solid}.
Our derivation identifies the contributions of the Phase-Field mechanism not only to the transport of mass, momentum, and energy but also to the work between phases. Furthermore, new non-conservative terms associated with these contributions of the Phase-Field mechanism appear in the volume fraction equation of the reduced five-equation model resulting from the isobaric closure via the pressure relaxation. These terms were not discovered in previous studies \citep{Shuklaetal2010,Tiwarietal2013,Jainetal2020,JainMoin2022,HuangJohnsen2022,HuangJohnsen2023,HuangJohnsen2024}, where the Phase-Field mechanism was incorporated into the five-equation models. As a result, once the Phase-Field mechanism are incorporated, the proposed six-equation model with the pressure relaxation is no longer mathematically equivalent to the five-equation models \citep{Shuklaetal2010,Tiwarietal2013,Jainetal2020,JainMoin2022,HuangJohnsen2022,HuangJohnsen2023,HuangJohnsen2024} with the non-conservative term of Kapila et al. \citep{Kapilaetal2001}, and hence should not be treated as their alternative.

The proposed consistent and conservative numerical approach is modified from our previous approach \citep{HuangJohnsen2022,HuangJohnsen2023,HuangJohnsen2024} based on the five-equation models with the Phase-Field mechanism. To adapt to the proposed six-equation model, the non-conservative terms in the hyperbolic and Phase-Field steps are reformulated to link to the corresponding numerical fluxes; the HLLC flux \citep{Saureletal2009} is modified for the phasic total energy in the hyperbolic step; a new mapping is applied to the UD flux \citep{HuangJohnsen2024} for the volume fraction boundedness and mass positivity in the Phase-Field step; and the consistency requirement for reconstruction is extended to include the phasic pressure weighted by the volume fraction to preserve the velocity, pressure, and temperature equilibria at isolated interfaces for any equation of state.
Both the pressure and pressure-temperature relaxations, the new component of solving the six-equation model, are theoretically analyzed in detail, proving for the first time that there exists a unique thermodynamically admissible solution for these two relaxations under a general multiphase setup with the equation of state by Le M{\'e}tayer et al. \citep{LeMetayeretal2005}.

The proposed model and numerical approach are demonstrated in various compressible two-phase flow benchmarks that include shocks, rarefactions, interfaces, and their interactions. Good agreement with the exact solutions is achieved when the pressure relaxation is applied, while the results with the pressure-temperature relaxation are also supplemented. The interfacial equilibrium condition and the conservation of mass, momentum, and energy are verified. Thanks to the Phase-Field mechanism that quantitatively controls the interface thickness, a discussion about the effect of interface thickness is enabled, which illustrates the existence of a sharp-interface limit in a water-air shock tube problem with large density and pressure ratios. We also validate that the six-equation model based on the phasic total energy captures correct bubble collapse dynamics.

\section*{Acknowledgments}
ZH acknowledges Dr. William J. White, Dr. Baudouin Fonkwa Kamga, and Prof. Eric Johnsen for their fruitful discussions, while ZH was at the University of Michigan and the University of Alabama.

\appendix
\section{Proof of Galilean invariance of the proposed model}\label{Appendix:Galilean}
Given a fixed frame $(\mathbf{x},t)$ and a moving frame $(\mathbf{x}',t')$ that has a constant velocity $\mathbf{u}_0$, the Galilean transformation is
\begin{equation}
\mathbf{x}'=\mathbf{x}-\mathbf{u}_0 t,\quad
t'=t,\quad
\mathbf{u}'=\mathbf{u}-\mathbf{u}_0,\quad
f'(\mathbf{x}',t')=f(\mathbf{x},t),\quad
\frac{\partial f'}{\partial t'}=\frac{\partial f}{\partial t}+\mathbf{u}_0 \cdot \nabla f,\quad
\nabla' f'=\nabla f,
\end{equation}
where $f$ is a scalar function in the fixed frame $(\mathbf{x},t)$ , and $f'$ is the same quantity  in the moving frame $(\mathbf{x}',t')$.
For convenience, we define
\begin{equation}\label{Eq:Force-Interface}
\mathbf{F}_{p \leftarrow q}^I
=
\left(
\frac{(\alpha_p \rho_p)}{\rho} \nabla \cdot (\alpha_q \boldsymbol{\sigma}_q)
-
\frac{(\alpha_q \rho_q)}{\rho} \nabla \cdot (\alpha_p \boldsymbol{\sigma}_p)
\right)
+
\left(
\frac{(\alpha_p \rho_p)}{\rho} (\mathbf{J}_q \rho_q) \cdot  \nabla \mathbf{u}
-
\frac{(\alpha_q \rho_q)}{\rho} (\mathbf{J}_p \rho_p) \cdot \nabla \mathbf{u}
\right),
\end{equation}
resulting in
$
W_{p \leftarrow q}^I
=
\mathbf{u} \cdot \mathbf{F}_{p \leftarrow q}^I
-
\zeta_{p,q}^I P_{p,q}^I (P_p - P_q).
$

Given $\mathbf{J}_p$, $\boldsymbol{\sigma}_p$, $\mathbf{Q}_p$, and $Q_{p \leftarrow q}^I$ objective (frame indifferent), the phasic mass equation in the moving frame is
\begin{eqnarray}\label{ProofGalilean:Mass-Phase}
\frac{\partial (\alpha'_p \rho'_p)}{\partial t'}
+
\nabla' \cdot \left( (\alpha'_p \mathbf{u}' - \mathbf{J}'_p)\rho'_p \right)
=
\underbrace{\mathbf{u}_0 \cdot \nabla (\alpha_p \rho_p)
-
\nabla \cdot \left( \alpha_p \rho_p \mathbf{u}_0 \right)}_{=0}
+
\underbrace{\frac{\partial (\alpha_p \rho_p)}{\partial t}
+
\nabla \cdot \left( (\alpha_p\mathbf{u} - \mathbf{J}_p)\rho_p \right)}_{=0}
=
0,
\end{eqnarray}
the mixture momentum equation in the moving frame is
\begin{eqnarray}\label{ProofGalilean:Momentum}
\frac{\partial (\rho' \mathbf{u}')}{\partial t'}
+
\nabla' \cdot (\mathbf{m}' \otimes \mathbf{u}')
-
\nabla' \cdot \boldsymbol{\sigma}'
=
-
\underbrace{\left( \frac{\partial (\rho \mathbf{u}_0)}{\partial t}
+
\nabla \cdot (\mathbf{m} \otimes \mathbf{u}_0) \right)}_{=\mathbf{u}_0 \left( \frac{\partial \rho}{\partial t}
+
\nabla \cdot \mathbf{m} \right)=\mathbf{0}}\\
\nonumber
+
\underbrace{\mathbf{u}_0 \cdot \nabla ( \rho (\mathbf{u}-\mathbf{u}_0) )
-
\nabla \cdot ( \rho \mathbf{u}_0 \otimes (\mathbf{u}-\mathbf{u}_0))}_{=\mathbf{0}}
+
\underbrace{\frac{\partial (\rho \mathbf{u})}{\partial t}
+
\nabla \cdot (\mathbf{m} \otimes \mathbf{u})
-
\nabla \cdot \boldsymbol{\sigma}}_{=\mathbf{0}}
=\mathbf{0},
\end{eqnarray}
the phasic total energy equation in the moving frame is
\begin{eqnarray}\label{ProofGalilean:Energy-Phase}
\frac{\partial (\alpha'_p \rho'_p E'_p)}{\partial t'}
+
\nabla' \cdot \left( (\alpha'_p \mathbf{u}' - \mathbf{J}'_p) \rho'_p E'_p \right)
-
\nabla' \cdot (\alpha'_p \mathbf{u}' \cdot \boldsymbol{\sigma}'_p)
+
\nabla' \cdot (\alpha'_p \mathbf{Q}'_p)\\
\nonumber
-
\sum_{q=1}^N \mathbf{u}' \cdot (\mathbf{F}_{p \leftarrow q}^I)'
-
\sum_{q=1}^N (Q_{p \leftarrow q}^I)'
+
\sum_{q=1}^N (\zeta_{p,q}^I)' (P_{p,q}^I)' (P'_p - P'_q)\\
\nonumber
=
\underbrace{\left(
\mathbf{u}_0 \cdot \nabla (\alpha_p \rho_p E_p)
- 
\nabla \cdot \left(\alpha_p\rho_p E_p \mathbf{u}_0 \right)
\right)}_{=0}
+
\mathbf{u}_0 \cdot \underbrace{\left(
\nabla \cdot \left( \alpha_p \rho_p \mathbf{u}_0 \otimes \mathbf{u}\right)
-
\mathbf{u}_0 \cdot \nabla (\alpha_p \rho_p \mathbf{u})
\right)}_{=\mathbf{0}}\\
\nonumber
+
\frac{\mathbf{u}_0 \cdot \mathbf{u}_0}{2} \underbrace{\left(
\mathbf{u}_0 \cdot \nabla (\alpha_p \rho_p)
-
\nabla \cdot \left( \alpha_p \rho_p \mathbf{u}_0 \right)
\right)}_{=0}
+
\frac{\mathbf{u}_0 \cdot \mathbf{u}_0}{2} \underbrace{\left(
\frac{\partial (\alpha_p \rho_p)}{\partial t}
+
\nabla \cdot \left( (\alpha_p \mathbf{u} - \mathbf{J}_p) \rho_p \right)
\right)}_{=0}\\
\nonumber
+
\underbrace{\left(
\begin{array}{cc}
     \frac{\partial (\alpha_p \rho_p E_p)}{\partial t}
+
\nabla \cdot \left( (\alpha_p \mathbf{u} - \mathbf{J}_p)\rho_p E_p \right)
-
\nabla \cdot (\alpha_p \mathbf{u} \cdot \boldsymbol{\sigma}_p)
+
\nabla \cdot (\alpha_p \mathbf{Q}_p)&  \\
     -
\sum_{q=1}^N \mathbf{u} \cdot \mathbf{F}_{p \leftarrow q}^I
-
\sum_{q=1}^N Q_{p \leftarrow q}^I
+
\sum_{q=1}^N \zeta_{p,q}^I P_{p,q}^I (P_p - P_q)& 
\end{array}
\right)}_{=0}\\
\nonumber
-
\mathbf{u}_0 \cdot \underbrace{\left(
\frac{\partial (\alpha_p \rho_p \mathbf{u})}{\partial t}
+
\nabla \cdot \left( (\alpha_p \mathbf{u} - \mathbf{J}_p)\rho_p \otimes \mathbf{u} \right)
-
\nabla \cdot (\alpha_p \boldsymbol{\sigma}_p)
- 
\sum_{q=1}^N \mathbf{F}_{p \leftarrow q}^I
\right)}_{=\mathbf{0}}
=
0,
\end{eqnarray}
and the volume fraction equation in the moving frame is
\begin{eqnarray}\label{ProofGalilean:VolumeFraction}
\frac{\partial \alpha'_p}{\partial t'}
+
\nabla' \cdot (\alpha'_p \mathbf{u}' - \mathbf{J}'_p)
-
\alpha'_p \nabla' \cdot \mathbf{u}'
-
\sum_{q=1}^N (\zeta_{p,q}^I)' (P'_p - P'_q)\\
\nonumber
=
\underbrace{\mathbf{u}_0 \cdot \nabla \alpha_p
-
\nabla \cdot (\alpha_p \mathbf{u}_0 )+\alpha_p \nabla \cdot \mathbf{u}_0}_{=0}
+
\underbrace{\frac{\partial \alpha_p}{\partial t}
+
\nabla \cdot (\alpha_p \mathbf{u} - \mathbf{J}_p)
-
\alpha_p \nabla \cdot \mathbf{u}
-
\sum_{q=1}^N \zeta_{p,q}^I (P_p - P_q)
}_{=0}
=0.
\end{eqnarray}
As a result, the proposed model is Galilean invariant.
The derivation of Eq.~(\ref{ProofGalilean:Energy-Phase}) uses
\begin{equation}\label{Eq:Momentum-Phase}
\begin{split}
\frac{\partial (\alpha_p \rho_p \mathbf{u})}{\partial t}
+
\nabla \cdot \left( (\alpha_p \mathbf{u} - \mathbf{J}_p) \rho_p \otimes \mathbf{u} \right)
-
\nabla \cdot (\alpha_p \boldsymbol{\sigma}_p)
-
\sum_{q=1}^N \mathbf{F}_{p \leftarrow q}^I\\
=
\frac{(\alpha_p \rho_p)}{\rho} \underbrace{\left(\frac{\partial (\rho \mathbf{u})}{\partial t} +
\nabla \cdot \left(\sum_{q=1}^N (\alpha_q\mathbf{u} - \mathbf{J}_q)\rho_q \otimes \mathbf{u}\right)-
\nabla \cdot \left(\sum_{q=1}^N \alpha_q \boldsymbol{\sigma}_q\right)\right)}_{=\mathbf{0}}\\
-
\frac{(\alpha_p \rho_p)}{\rho} \mathbf{u} \sum_{q=1}^N \underbrace{\left(
\frac{\partial (\alpha_q\rho_q)}{\partial t}
+
\nabla \cdot \left( (\alpha_q\mathbf{u} - \mathbf{J}_q)\rho_q \right)
\right)}_{=0}
+
\mathbf{u} \underbrace{\left(\frac{\partial (\alpha_p \rho_p)}{\partial t} +
\nabla \cdot \left( (\alpha_p\mathbf{u} - \mathbf{J}_p)\rho_p \right) \right)}_{=0}\\
+
\underbrace{\left(
(\alpha_p \rho_p) \mathbf{u} \cdot \nabla \mathbf{u}
-
\frac{(\alpha_p \rho_p)}{\rho} \sum_{q=1}^N \alpha_q \rho_q \mathbf{u} \cdot \nabla \mathbf{u}
\right)}_{=\mathbf{0}}\\
+
\sum_{q=1}^N \underbrace{\left(
\frac{(\alpha_p \rho_p)}{\rho} \nabla \cdot \left(\alpha_q \boldsymbol{\sigma}_q\right)
-
\frac{(\alpha_q \rho_q)}{\rho} \nabla \cdot \left(\alpha_p \boldsymbol{\sigma}_p\right)
+
\frac{(\alpha_p \rho_p)}{\rho} (\mathbf{J}_q \rho_q) \cdot \nabla \mathbf{u}
-
\frac{(\alpha_q\rho_q)}{\rho} (\mathbf{J}_p \rho_p) \cdot \nabla \mathbf{u}
\right)}_{=\mathbf{F}_{p \leftarrow q}^I}
-
\sum_{q=1}^N \mathbf{F}_{p \leftarrow q}^I
=
\mathbf{0},
\end{split}
\end{equation}
which is the phasic momentum equation of the six-equation model.

\bibliographystyle{plain}
\bibliography{refs.bib}

\end{document}